\documentclass[11pt]{article}
\usepackage[a4paper,margin=1in]{geometry}
\usepackage[T1]{fontenc}
\usepackage[utf8]{inputenc}
\usepackage{lmodern}
\usepackage{microtype}
\usepackage{amsmath,amssymb,amsthm,amsfonts,array,mathrsfs,ifthen,mathtools}
\usepackage{dsfont}
\usepackage{setspace}
\usepackage{authblk}

\usepackage{graphicx}
\usepackage{xcolor}
\usepackage{bm}
\usepackage{enumitem}
\usepackage{url}
\usepackage{hyperref}
\definecolor{warmdarkred}{RGB}{120, 20, 30} 
\hypersetup{
	colorlinks=true,
	linkcolor=warmdarkred,
	urlcolor=warmdarkred,
	citecolor=warmdarkred
}
\usepackage[nameinlink,capitalise,noabbrev]{cleveref}
\usepackage[round]{natbib}
\usepackage{comment}
\usepackage{xpatch}
\xpatchcmd{\proof}{\itshape}{\normalfont\proofnamefont}{}{}

\newcommand{\proofnamefont}{\bfseries}

\newtheorem{proposition}{Proposition}
\newtheorem{lemma}{Lemma}
\newtheorem{corollary}{Corollary}
\newtheorem{assumption}{Assumption}
\theoremstyle{remark}

\newtheorem{definition}{Definition}
\crefname{assumption}{assumption}{assumptions}
\Crefname{assumption}{Assumption}{Assumptions}

\crefname{example}{example}{examples}
\Crefname{example}{Example}{Examples}

\crefname{remark}{remark}{remarks}
\Crefname{remark}{Remark}{Remarks}

\crefname{section}{section}{sections}
\Crefname{section}{Section}{Sections}

\crefname{subsection}{subsection}{subsections}
\Crefname{subsection}{Subsection}{Subsections}

\crefname{equation}{equation}{equations}
\Crefname{equation}{Equation}{Equations}

\usepackage{authblk}

\usepackage[colorinlistoftodos]{todonotes}

\newcommand{\FI}{\mathrm{FI}}
\newcommand{\DE}{\mathrm{AE}}
\newcommand{\M}{\mathrm{M}}
\newcommand{\NWBR}{\mathrm{NWBR}}

\title{{\bf A Case for Competition in Information Provision}\thanks{This paper builds on a previously circulated draft titled “Direct Competition in News Markets.” We thank Yair Antler, Rossella Argenziano, Ennio Bilancini, Daniele Condorelli, Andrea Galeotti, Jayant Ganguli, Christian Ghiglino, Aniol Llorente-Saguer, Andrea Mattozzi, Oleg Rubanov, Alex Smolin, Santiago Oliveros, Federico Trombetta, and seminar participants at the University of Essex, the 2018 RES Symposium of Junior Researchers, the IV Leicester International PhD Conference in Economics, the 2019 Annual LSE–NYU Conference, and the IX IBEO Workshop for valuable comments on this and earlier versions. All errors are our own. Financial support from the Economic and Social Research Council (ESRC grant n.~1366373) is gratefully acknowledged by Vaccari. During the preparation of this work, we used OpenAI's ChatGPT as a form of research assistance, and \href{https://www.refine.ink/}{Refine.ink} for feedback. We reviewed the output as needed, and take full responsibility for the content of the article.}}

\date{}

\author{
	\begin{minipage}[t]{0.45\textwidth}
		\centering
		Bianca Sanesi\\
		{\small University of Bologna}\\
    Department of Economics\\
		{\small \texttt{bianca.sanesi@unibo.it}}
	\end{minipage}
	\begin{minipage}[t]{0.45\textwidth}
		\centering
		Federico Vaccari\\
		{\small University of Bergamo}\\
    Department of Economics\\
		{\small \texttt{federico.vaccari@unibg.it}}
	\end{minipage}
}

\begin{document}
	\maketitle

	\begin{abstract}
We study how competition among biased news sources affects information and welfare when sources can misrepresent facts at a cost. Monopolistic and competitive market structures admit many equilibria. We develop a common belief-based selection criterion that applies to both and makes welfare comparisons possible. Under the refined outcomes, adding an oppositely biased source improves the receiver’s welfare when that source faces sufficiently high misreporting costs. Competition disciplines the incumbent while introducing a few distortions of its own. Better information need not increase total welfare, and distorted advice from a monopolist can raise total welfare. Competition may improve decision-making without being socially beneficial.
\end{abstract}
	
	\bigskip
	
	 \bigskip
	
	
	\newpage
	
	\begingroup
\hypersetup{linkcolor=black}
\tableofcontents
\endgroup
	
	\newpage

	
	\section{Introduction}

Competition in news markets is valued because it exposes audiences to different accounts of the same facts. Media-ownership rules, for example, are partly motivated by the idea that independent viewpoints help audiences distinguish information from persuasion.\footnote{See, for example, \cite{FCC2018,FCC2023}.} The same idea underlies the ``marketplace of ideas'' argument, and appears in settings ranging from expert advice to adversarial legal proceedings \citep{GentzkowShapiro2008}. When interested sources can challenge one another, decision makers should become better informed. This paper asks when that intuition is correct.

The answer is not immediate when claims are neither costless nor fully verifiable. News sources, experts, and advocates can often misrepresent facts, but larger distortions expose them to greater reputational, legal, or technological costs. Competitive information provision has two effects. First, a second source can challenge the incumbent and expose its attempts at persuasion. Second, it can also introduce distortions of its own, due to, e.g., a persuasion rat-race. Whether competition improves information depends on the balance between these forces.

We study this question in a costly communication game. One or two informed senders observe a continuous state and report it to a receiver, who chooses one between two actions. Reports need not be truthful, but the cost of a report increases with its distance from the truth. Under monopoly, a single sender biased toward one action communicates with the receiver. Under competition, a second sender with the opposite bias observes the same state and reports simultaneously. The model isolates the informational effect of giving the receiver access to an additional, strategically opposed source.

The main challenge in performing welfare comparisons is equilibrium multiplicity. Under monopoly, different beliefs about unexpected reports support outcomes ranging from substantial pooling to almost full revelation. Competition also admits outcomes with different informational properties. Comparing a favorable monopoly equilibrium with an unfavorable competitive equilibrium, or the reverse, can produce any welfare ranking. An informative comparison requires a common way of selecting outcomes in the two market structures.

Our first contribution is to put outcome selection in the two market structures on a common footing. We develop a belief-based criterion that applies uniformly across configurations and selects a single outcome in each. By subjecting both benchmarks to the same discipline on beliefs, it permits a welfare comparison that does not rest on separate selection arguments for monopoly and competition.

Although we focus on a specific environment, our model captures the central difficulty of comparing monopoly and competition in information provision. Both market structures admit many equilibria. The continuous state makes the size and cost of misreporting meaningful, while two senders provide the simplest setting in which one source can discipline another. We establish our refinement only for this class of models and make no claim of a general selection result. Within this setting, it provides a common basis for comparison, showing when competition improves the receiver’s decisions and why that improvement need not raise total welfare.

Our second contribution is to characterize when competition benefits the receiver. Competition does not dominate every monopoly equilibrium. A sufficiently skeptical receiver obtains almost full information from a monopolist and can be better off than under competition. The conclusion changes under the outcomes selected by our common refinement criterion. With symmetric primitives, competition gives the receiver a strictly higher welfare compared to the selected monopoly outcome. More generally, competition benefits the receiver when the additional sender faces sufficiently high misreporting costs. As those costs rise, the mistakes caused by the second sender disappear, while his presence continues to constrain the incumbent's attempts at persuasion. For any fixed finite bias of the additional sender and any fixed non-revealing monopoly equilibrium, competition gives the receiver higher welfare once the additional sender's misreporting cost is sufficiently high.\footnote{The required cost threshold depends on the selected monopoly equilibrium, and thus no single finite threshold applies to the entire equilibrium family.}

Our third contribution is to show that this informational gain need not raise total welfare. The receiver chooses the action that maximizes her own payoff, not necessarily the sum of all players' payoffs. A monopolist may induce an action that harms the receiver but increases aggregate action payoffs. If that gain exceeds the associated reporting cost, persuasion is socially valuable even though it makes the receiver worse off. In the linear-quadratic model specification, the selected monopoly outcome can generate greater total welfare than the selected competitive one when the monopolist's preference for the positive action is sufficiently strong relative to the second sender's opposition and the monopolist's misreporting-cost parameter is sufficiently high. Conversely, when the second sender's opposition is strong enough to align the receiver's preferred action with the action that maximizes aggregate surplus on the states distorted by monopoly, sufficiently disciplined competition improves both welfare criteria.

To sum up, competition is valuable to the receiver when an additional source can challenge one-sided persuasion but has limited scope to manipulate information itself. Its social value also depends on whose payoffs enter the welfare criterion. More sources can improve decisions without increasing total welfare. Plurality and accountability are complements, and better information is not by itself a complete measure of the benefits of competition in the whole economy.


\subsection{Related literature}

This paper contributes to work on communication by multiple informed senders. In cheap-talk models, consulting a second expert can improve communication, but the gain depends on the senders' preferences, the dimensionality of the state, and the timing of reports \citep{KrishnaMorgan2001,Battaglini2002,LiRantakariYang2016}.\footnote{Related work studies how preference alignment and the receiver's own information affect equilibrium informativeness in single-sender communication \citep{ChenGordon2015,IshidaShimizu2019}.} When interested parties can disclose verifiable evidence or choose information structures, competition can also increase revelation \citep{MilgromRoberts1986,BhattacharyaMukherjee2013,GentzkowKamenica2017a,GentzkowKamenica2017b,AuKawai2020,LiNorman2021}. These models provide two useful benchmarks: costless messages and hard evidence or committed experiments. Our reports lie between them. They are literal claims about a common state and may be false, but larger distortions cost more. Competition can restrain one source while inducing costly misreporting by the other.

The closest work studies communication with lying or reporting costs. In single-sender models, \citet{KartikOttavianiSquintani2007} and \citet{Kartik2009} give otherwise cheap messages a direct payoff consequence, while \citet{OttavianiSquintani2006} and \citet{Chen2011} study credulity and behavioral perturbations. Close to our monopoly problem, \citet{HodlerLoertscherRohner2014} and \citet{Vaccari2023ET} study a biased source who reports a continuous state to a receiver facing a binary decision and pays more for larger lies. With several interested parties, \citet{EmonsFluet2019} attach reporting costs to the disclosure of verifiable facts, while \citet{SkaperdasVaidya2012} model persuasion as costly evidence production by contending parties. In a two-sender costly-talk environment, \citet{Vaccari2023JET} characterizes the adversarial equilibrium that serves as our competitive benchmark. We compare that benchmark with monopoly under a common selection criterion and evaluate both receiver welfare and total welfare.

Equilibrium selection is part of the paper's contribution. Belief-based refinements restrict how a receiver interprets unexpected messages in single-sender signaling games \citep{GrossmanPerry1986,MailathOkunoFujiwaraPostlewaite1993}. The Never-a-Weak-Best-Response (NWBR) test extends forward-induction reasoning to infinite signaling games \citep{Manelli1997}. For multi-sender games, \citet{VidaHonryo2021} connect strategic stability to unprejudiced beliefs, and \citet{VidaHonryoAzacis2026} develop a forward-induction criterion for monotonic multi-sender signaling games. We use these ideas for a comparative purpose. The same restrictions on beliefs and continuation play are imposed under monopoly and competition, so the welfare ranking does not come from selecting a favorable equilibrium separately in each market structure.

The paper also relates to media bias and media competition. \citet{Stromberg2015} and \citet{GentzkowShapiroStone2015} provide broad surveys. Reader demand and reputation can generate slant \citep{MullainathanShleifer2005,GentzkowShapiro2006,BernhardtKrasaPolborn2008}. Other work studies bias arising from journalists' information, political capture, or the interaction between media slant and electoral incentives \citep{Baron2006,BesleyPrat2006,DugganMartinelli2011}. We take sources' preferences as given and abstract from prices, audience allocation, and the endogenous choice of slant. This isolates how access to an additional, oppositely biased source changes a receiver's decision when both sources observe the same facts.

The effect of media competition is not unambiguously positive in the existing literature. Competition may improve accuracy, but it can also change outlets' incentives to acquire, differentiate, or supply information \citep{GentzkowShapiro2008,PiolattoSchuett2015,PeregoYuksel2022}. Empirically, newspaper entry can reduce hard-news provision and political participation \citep{Cage2020}. Our mechanism is different. Competition changes neither audience demand nor the allocation of attention across outlets. It changes the strategic contest over one receiver's action. We identify when cross-checking improves that decision under partial verifiability, and why greater decision accuracy need not increase aggregate welfare.


\section{Model}

We study communication between an uninformed receiver and one or two informed senders. The senders observe a continuous state and report it to the receiver, who then chooses between two actions. Reports are literal but not verifiable. A sender may misrepresent the state, although doing so entails a cost that increases with the size of the misreport. The senders have opposing biases. Sender~1 favors the positive action relative to the receiver, whereas sender~2 favors the negative action.

The model allows us to isolate the informational effect of access to a competing source. Under monopoly, only sender~1 communicates with the receiver. Under competition, sender~2 also observes the state and submits a separate report. The receiver observes all reports before choosing her action. We first introduce the primitives common to both environments, and then describe the monopoly and competitive outcomes used in the subsequent analysis.


	\subsection{Common primitives}\label{subsec:common-primitives}
    
	There is a continuous state, $\theta\in\Theta\subseteq\mathbb{R}$, with density $f$ and full support on $\Theta$. The receiver chooses one of two actions, denoted by $a^-$ and $a^+$. We normalize the receiver's utility from action $a^-$ to zero and write $u_r(\theta)$ for the utility from action $a^+$.	We assume that there is a unique threshold in $\theta$, normalized to zero, such that the receiver's payoff is $u_r(\theta)<0$ for $\theta<0$, $u_r(\theta)>0$ for $\theta>0$, and $u_r(0)=0$. For each player $i$, $u_i$ is continuous and strictly increasing in $\theta$. Under full information, the receiver chooses $a^+$ if and only if $\theta\geq 0$. Whenever the receiver is indifferent between $a^+$ and $a^-$ under her posterior belief, she chooses $a^+$.

	In every market configuration, senders observe $\theta$ perfectly, and then send a literal report $r_j\in\Theta$. Let $u_j(\theta)$ denote sender $j$'s utility from action $a^+$, again with utility from action $a^-$ normalized to zero. Each sender has a unique threshold $\tau_j$ satisfying $u_j(\tau_j)=0$ and, throughout, sender~$1$ is upward biased while sender~$2$ is downward biased.\footnote{The threshold $\tau_j$ provides a convenient reduced-form summary of sender $j$'s bias and position relative to the other players, since it identifies the state at which $j$ is indifferent between the two actions. The analysis extends naturally to cases in which $u_j(\theta)$ never crosses zero, so that sender $j$ always weakly prefers the same action.} That is,
	\[
	\tau_1<0<\tau_2.
	\]
	If sender $j$ sends report $r_j$ when the state is $\theta$, he incurs misreporting cost $k_j C_j(r_j,\theta)$, where $k_j>0$ and the cost function $C_j$ satisfies the following properties,
	\begin{enumerate}
		\item[$(i)$] $C_j(\theta,\theta)=0$ for all $\theta$;
		\item[$(ii)$] $C_j$ is continuous on $\Theta^2$ and differentiable at
    every point $(r_j,\theta)$ with $r_j\neq\theta$;
		\item[$(iii)$] for every fixed state $\theta$, $C_j(r_j,\theta)>C_j(r_j',\theta)$ whenever $|r_j-\theta|>|r_j'-\theta|$;
        \item[$(iv)$] for every fixed report $r_j$, $C_j(r_j,\theta)>C_j(r_j,\theta')$ whenever $|r_j-\theta|>|r_j-\theta'|$;
		\item[$(v)$] $C_j(r_j,\theta)\to\infty$ as $|r_j-\theta|\to\infty$. 
	\end{enumerate}

A mixed strategy for sender $j$ assigns, at each state $\theta$, a probability distribution $\phi_j(\cdot\mid\theta)$ over reports. We use the standard independent-private-randomization convention. Conditional on the state, the senders' randomize independently. There is no public correlating device.
    
	Hence, sender $j$'s payoff is
	\[
	w_j(r_j,a,\theta)=\mathds{1}_{\{a=a^+\}}u_j(\theta)-k_j C_j(r_j,\theta).
	\]

	For every state $\theta\in\Theta$ with $\theta\geq\tau_1$, we define the \emph{reach} of sender~$1$ as the largest upward report that is weakly better than truthful reporting when action $a^+$ is induced with certainty, i.e.,
	\[
	\bar r_1(\theta)\coloneqq \max\{r\in\Theta \mid u_1(\theta)=k_1 C_1(r,\theta)\}.
	\]
	Likewise, for every state $\theta\in\Theta$ with $\theta\leq\tau_2$, the \emph{reach} of sender~$2$ is the smallest downward report that is weakly better than truthful reporting when action $a^-$ is induced with certainty, i.e.,
	\[
	\underline{r}_2 (\theta)\coloneqq \min\{r\in\Theta \mid -u_2(\theta)=k_2 C_2(r,\theta)\}.
	\]
	We assume that the report space is large enough for these reaches to exist on their respective domains and that both reach functions are continuous and strictly increasing there. Whenever needed, we denote the \emph{inverse reaches} by $\bar r_1^{-1}$ and $\underline{r}_2 ^{-1}$. We assume that $\bar r_1$ and $\underline r_2$ are continuous and strictly increasing on the domains used below. Their increasing inverses are well defined on the corresponding images.\footnote{The domain restrictions are necessary because $u_1(\theta)<0$ when $\theta<\tau_1$, whereas $k_1C_1(r,\theta)\geq0$, and $-u_2(\theta)<0$ when $\theta>\tau_2$, whereas $k_2C_2(r,\theta)\geq0$. Thus, the corresponding reach equations need not have solutions outside the stated domains.}
	
	We assume throughout that the state and report spaces are sufficiently large to contain all equilibrium and comparison-relevant reports generated by the parameter values under consideration. To analyze limits such as $|\tau_j|\to \infty$ or $k_j\to 0^+$, we require an unbounded report space, $\Theta\equiv\mathbb{R}$. Therefore, any limitation on information transmission is endogenous and arises from strategic incentives and misreporting costs rather than from exogenous bounds on feasible reports.

	\subsection{The monopolistic environment}

	The monopoly environment has a single sender, namely sender~$1$. The family of \emph{generic} monopolistic equilibria is indexed by a parameter $\lambda\in[0,\bar\lambda]$, where $\lambda$ indexes the receiver's skepticism.\footnote{We use ``skepticism'' as a label for the equilibrium index. Formally, $\lambda$ is the receiver's expected-payoff advantage from choosing $a^+$ at the pooling report, so a larger $\lambda$ makes $a^+$ more attractive there. Within the generic equilibrium family, it is associated with a smaller pooling region and, hence, a more informative equilibrium.} For every $\lambda$, the equilibrium takes the following form,\footnote{For a complete equilibrium characterization of this configuration, see \cite{Vaccari2023ET}. The generic family coincides exactly with the set of equilibria that survive the Intuitive Criterion \citep{ChoKreps1987}.}
	\begin{enumerate}[label=($\roman*$)]
		\item There is a pooling report, $r^*(\lambda)>0$;
		\item There is a lower pooling type, $\ell(\lambda)\coloneqq \bar r_1^{-1}(r^*(\lambda))<0$;
		\item Sender~$1$ reports truthfully when $\theta\notin(\ell(\lambda),r^*(\lambda))$;
		\item Sender~$1$ pools all states in $(\ell(\lambda),r^*(\lambda))$ at report $r^*(\lambda)$.
	\end{enumerate}
We index this family by the receiver's continuation value at the pooling report. Thus,
\begin{equation}
\lambda \coloneqq \frac{ \displaystyle\int_{\ell(\lambda)}^{r^*(\lambda)} u_r(\theta)f(\theta)\,d\theta }{ \displaystyle\int_{\ell(\lambda)}^{r^*(\lambda)} f(\theta)\,d\theta }.
\label{eq:lambda-pooling-value}
\end{equation}
	The least informative equilibrium is also the sender-preferred equilibrium, indexed by $\lambda=0$. We write
	\[
	r^M\coloneqq r^*(0),
	\]
	\[
	\ell^M\coloneqq \ell(0)=\bar r_1^{-1}(r^M).
	\]
	By construction, the receiver is indifferent at the least informative pooling report, i.e.,
	\begin{equation*}
		\int_{\ell^M}^{r^M} u_r(\theta)f(\theta)\,d\theta=0.
	\end{equation*}
	At the other extreme, the most informative equilibrium, indexed by $\lambda=\bar\lambda$, yields the receiver's full-information payoff.

    
	\subsection{The competitive environment and the adversarial equilibrium}\label{subsec:competitive-environment}
	
	The competitive environment has both sender~$1$ and sender~$2$. We focus on the \emph{adversarial equilibrium} of such a setting characterized by \citet{Vaccari2023JET}. That characterization applies to the primitives in \Cref{subsec:common-primitives}. Our maintained assumptions that the reaches are continuous and strictly increasing support the inverse-reach notation used below.

\begin{proposition}
\label{prop:AE-characterization}
Under the primitives in \Cref{subsec:common-primitives}, the two-sender game has an adversarial equilibrium. Adversarial equilibria are essentially unique in outcomes and strategies, as they induce the same reporting strategies and on-path receiver actions, although they may differ in off-path beliefs. In every adversarial equilibrium,
\begin{enumerate}[label=(\roman*)]
    \item the receiver's continuation value satisfies conditions $(wM)$, $(sM)$, and $(Dom)$ of \citet{Vaccari2023JET};

    \item there is a unique strictly decreasing ``swing function''
    \[
    s: [\underline r_2(0),\bar r_1(0)] \longrightarrow [\underline r_2(0),\bar r_1(0)]
    \]
    satisfying $s(s(r))=r$;

    \item there are two truthful cutoffs, $\theta_L,\theta_H$, such that $\tau_1<\theta_L<0<\theta_H<\tau_2$. Both senders report truthfully outside $(\theta_L,\theta_H)$ and mix inside that interval;

    \item conditional on a state in the mixing interval, sender $j$'s report law has the representation
    \[
    \phi_j(dr_j\mid\theta) = \alpha_j(\theta)\delta_\theta(dr_j) + \psi_j(r_j\mid\theta)\,dr_j,
    \]
    where the only possible atom is at the truthful report, $\delta_\theta$ is a Dirac's delta function. and $\psi_j(\cdot\mid\theta)$ is the density of the remaining atomless component;

    \item an adversarial equilibrium is robust to vanishing report noise, or is outcome- and strategy-equivalent to one with that robustness property.
\end{enumerate}
\end{proposition}

The proposition restates the characterization, strategy representation, and essential uniqueness result in \citet{Vaccari2023JET}. The swing function, truthful cutoffs, and mixed-strategy decomposition are conclusions of that result.

Hereafter, we say that sender~$j$ \emph{swings} the opponent's report $r_{-j}$ when the report pair induces sender~$j$'s preferred action. Thus, sender~$2$ swings a positive report $r_1$ if and only if $r_2<s(r_1)$, while sender~$1$ swings a negative report $r_2$ if and only if $r_1\geq s(r_2)$.

The truthful cutoffs are characterized by $s\!\left(\bar r_1(\theta_L)\right)=\theta_L$ and $s\!\left(\underline r_2(\theta_H)\right)=\theta_H$. Using $s(s(r))=r$, we can implicitly define the cutoffs by $s(\theta_L)=\bar r_1(\theta_L)$ and $s(\theta_H)=\underline r_2(\theta_H)$. Both senders report truthfully when $\theta\notin(\theta_L,\theta_H)$. At interior states, each sender mixes between truthful reporting and misreporting. We write $\alpha_j(\theta)$ for sender~$j$'s probability of truthful reporting and $\psi_j(\cdot\mid\theta)$ for the density of his continuous reporting component.

	Our central comparison is between the least informative monopoly equilibrium and the adversarial equilibrium. This is the natural welfare comparison, because the most informative monopoly equilibrium reaches the full-information benchmark and is trivially best for the receiver. 
	
\subsection{Polar communication benchmarks}
\label{subsec:polar-benchmarks}

Before analyzing costly misreporting, it is useful to recall the two polar communication technologies. Full results and proofs are collected in \Cref{app:polar-benchmarks}. 

Under cheap talk, equilibrium multiplicity makes the comparison between monopoly and competition inherently ambiguous. Competition can reproduce a monopoly outcome if the receiver ignores one source, which then babbles. In the standard advocacy equilibrium, however, an additional oppositely biased source weakly improves the receiver's welfare by providing a counterweight to the incumbent source.

This informational improvement need not increase total welfare. The receiver chooses the action that maximizes her own expected payoff, whereas total welfare also includes the action payoffs of both sources. When the receiver's preferred action differs from the action maximizing aggregate surplus, the additional information generated by competition can reduce total welfare.

At the opposite pole, when reports are fully verifiable, competition among oppositely biased sources implements the receiver's full-information action in every equilibrium. These benchmarks isolate the role of partial verifiability studied below. Costly misreporting disciplines communication without eliminating either strategic distortion or equilibrium multiplicity.
	
	
	\section{Monopoly and the limits of standard refinements}\label{sec:monopoly}
	
	In this section, we study the monopoly side of the model and ask which equilibria remain plausible once off-path deviations are disciplined by refinement. This step is important for two reasons. First, the monopolistic environment admits a continuum of equilibria, indexed by the receiver's skepticism, so the welfare comparison with competition requires a clear understanding of which monopoly outcomes survive refinement. Second, the refinement analysis is of independent interest. It shows how the sender's incentives to mimic different types vary across off-path reports, and how robust the equilibrium family is to a natural deletion procedure.

	Our focus is on the Never-a-Weak-Best-Response ($\NWBR$) refinement.\footnote{$\NWBR$ is a demanding refinement criterion. In standard signaling environments it is stronger than the Intuitive Criterion, divinity, universal divinity, and the refinements $D_1$ and $D_2$. The set of $\NWBR$ survivors contains strategically stable outcomes \citep{KohlbergMertens1986}.} The refinement is attractive in the present setting because it does not require committing to a specific off-path belief selection, but instead removes type-report pairs that are uniformly dominated as explanations for a deviation. We begin by characterizing the relevant deviation sets for a generic	monopoly equilibrium. We then show that, under a simple single-crossing condition on the attractiveness of off-path deviations, every generic monopoly equilibrium survives the refinement. The condition is transparent, and it is automatically satisfied in the linear-quadratic specification used later.
	
	We next provide primitive conditions under which the single-crossing property holds, including weakly convex distance costs, and specialize the result to the linear-quadratic environment. The section closes by explaining the implication for the monopoly benchmark. Under these conditions, $\NWBR$ leaves the entire generic equilibrium family intact. The case without single crossing is characterized report by report in Appendix~\ref{app:nwbr-arbitrary-qr}.

	
	\subsection{A generic NWBR result}
	Fix a generic monopoly equilibrium indexed by $\lambda$, with pooling report $r^*(\lambda)$ and lower pooling type $\ell(\lambda)$. Consider an off-path report $r\in(\ell(\lambda),r^*(\lambda))$. Since the receiver's action set is binary, any mixed reply after observing $r$ is summarized by a single number $\sigma\in[0,1]$, interpreted as the probability assigned to action $a^+$.

	Given such a report $r$, define the deviation payoff of type $\theta$ as
	\[
	w(r,\sigma,\theta)\coloneqq \sigma u_1(\theta)-k_1 C_1(r,\theta).
	\]
	Let $W_\lambda(\theta)$ denote the sender's equilibrium payoff in the monopoly equilibrium indexed by $\lambda$. Define
	\[
	D_0(\theta,r)\coloneqq \{\sigma\in[0,1] \mid w(r,\sigma,\theta)=W_\lambda(\theta)\},
	\]
	and
	\[
	D(\theta,r)\coloneqq \{\sigma\in[0,1] \mid w(r,\sigma,\theta)>W_\lambda(\theta)\}.
	\]
	Under the Never-a-Weak-Best-Response test, the type-report pair $(\theta,r)$ is deleted whenever
	\[
	D_0(\theta,r)\subseteq \bigcup_{\theta'\neq \theta} D(\theta',r).
	\]
	If every receiver mixed reply that makes type $\theta$ just indifferent makes some other type strictly prefer the same deviation, then the receiver should not assign the deviation to type $\theta$.
	
	We now turn to the refinement argument itself. The key step is to understand, for any off-path report, which sender types could rationally	benefit from deviating to it. This is summarized by the function $q_r$, which gives the receiver action probability that makes a type exactly indifferent between following the equilibrium strategy and deviating to report $r$. The shape of $q_r$ is what determines whether a deviation can be ruled out by the Never-a-Weak-Best-Response test. The next assumption imposes a simple single-crossing property on this object. It is the only substantive restriction needed for the result, and it will later be shown to hold automatically in the linear-quadratic specification and other standard cases.

	\begin{assumption}\label{ass:sc}
		Fix a generic monopoly equilibrium indexed by $\lambda$, with pooling report $r^*(\lambda)$ and lower pooling type $\ell(\lambda)$.
		For every off-path report $r\in(\ell(\lambda),r^*(\lambda))$, let $\vartheta(r,\lambda)$ be the unique type in $(r,r^*(\lambda))$ satisfying
		\[
		C_1(r,\vartheta(r,\lambda))=C_1(r^*(\lambda),\vartheta(r,\lambda)).
		\]
		Define
		\[
		q_r(\theta)\coloneqq 
		\begin{cases}
			\dfrac{k_1 C_1(r,\theta)}{u_1(\theta)}, & \theta\in[\bar r_1^{-1}(r),\ell(\lambda)],\\[1.1em] 1-\dfrac{k_1\left(C_1(r^*(\lambda),\theta)-C_1(r,\theta)\right)}{u_1(\theta)}, & \theta\in[\ell(\lambda),\vartheta(r,\lambda)].
		\end{cases}
		\]
		Assume that for every $r\in(\ell(\lambda),r^*(\lambda))$, the map $q_r$ is strictly decreasing on $[\bar r_1^{-1}(r),\ell(\lambda)]$ and strictly increasing on $[\ell(\lambda),\vartheta(r,\lambda)]$.
	\end{assumption}
	

	Assumption~\ref{ass:sc} concerns the shape of the indifference cutoff $q_r$ on the two segments of types that can potentially rationalize an off-path report. These two segments have a different economic meaning. On the first segment, types lie below the lower pooling type and compare an off-path deviation to an equilibrium payoff equal to zero. In that region, the argument is straightforward. As the type rises toward the report, the deviation becomes less costly, while the gain from inducing action $a^+$ becomes larger. Hence, the probability of action $a^+$ needed to make the sender indifferent should fall with the type. This part of the single-crossing property is already	built into the primitive structure of the model, and thus always satisfied.
	
	The second segment is more delicate. There, a type compares deviating to $r$ with pooling at $r^*(\lambda)$, so what matters is no longer the level of the deviation cost but the difference between the cost of deviating and the cost of pooling, relative to the sender's benefit from inducing $a^+$. The sign of that comparison is not pinned down by distance-monotonicity alone. It depends on how the cost function varies across reports as the type changes, and this is precisely where the substantive content of Assumption~\ref{ass:sc} lies.
	
	The next lemma isolates the easy part of the argument. It shows that the first piece of $q_r$ is automatically decreasing under the primitive monotonicity assumptions already imposed on preferences and misreporting costs. As a consequence, all of the real work in verifying Assumption~\ref{ass:sc} is on the second piece.

	\begin{lemma}
		\label{lem:first-branch-automatic}
		Fix $\lambda\in[0,\bar\lambda]$ and an off-path report $r\in\left(\bar r_1^{-1}(r^*(\lambda)),\,r^*(\lambda)\right)$. Suppose that $u_1(\theta)$ is strictly increasing on the relevant interval and that, for every fixed report $r_1$, the cost function satisfies $C_1(r_1,\theta)>C_1(r_1,\theta')$ whenever $|r_1-\theta|>|r_1-\theta'|$. Then, the first part of $q_r$,
		\[
		q_r(\theta)=\frac{k_1C_1(r,\theta)}{u_1(\theta)},
		\]
		is strictly decreasing in $\theta$ for all $\theta\in\left[\bar r_1^{-1}(r),\,\bar r_1^{-1}(r^*(\lambda))\right]$.
	\end{lemma}

	\Cref{lem:first-branch-automatic} shows that the only substantive content of \Cref{ass:sc} lies on the second segment of $q_r$. There, the issue is how the cost difference between deviating to $r$ and pooling at $r^*(\lambda)$ varies with the type. A natural way to control this is to assume that costs depend only on the distance between report and state, with weakly convex distance disutility. Under that restriction, the required single-crossing property follows directly. The next proposition states this formally.
	
	\begin{proposition}
		\label{prop:Gamma-costs}
		Fix $\lambda\in[0,\bar\lambda]$ and $r\in\left(\bar r_1^{-1}(r^*(\lambda)),\,r^*(\lambda)\right)$. Suppose that $C_1(r_1,\theta)=\Gamma(|r_1-\theta|)$, where $\Gamma:\mathbb R_+\to\mathbb R_+$ is strictly increasing, weakly convex, and continuously differentiable on $(0,\infty)$. Then, \Cref{ass:sc} is satisfied. In particular, this holds for the quadratic loss cost $C_1(r_1,\theta)=(r_1-\theta)^2$, and more generally for every $\Gamma(x)=x^\alpha$ with $\alpha\geq 1$.
	\end{proposition}

	\Cref{prop:Gamma-costs} provides a simple primitive route to \Cref{ass:sc}. In particular,	it covers the standard quadratic specification used later in the paper and, more generally, a broad class of distance-based cost functions. The refinement argument itself, however, does not require this specific structure. It only requires the single-crossing property of $q_r$. We state the next result directly under \Cref{ass:sc}, so as to keep the analysis fully general and separate the logic of the refinement from any particular functional-form restriction.

	
	\begin{proposition}\label{thm:nwbr}
		Under \Cref{ass:sc}, every generic monopolistic equilibrium survives the Never-a-Weak-Best-Response refinement.
	\end{proposition}

	
	\Cref{thm:nwbr} establishes a strong robustness property of the monopoly side of the model. Under \Cref{ass:sc}, the Never-a-Weak-Best-Response test does not shrink the generic family of monopolistic equilibria at all. On the contrary, every generic equilibrium survives. This matters for the rest of the paper for two reasons.

	First, it shows that the multiplicity of monopoly outcomes is not an artifact of weak off-path discipline. Even after imposing a non-trivial refinement, the entire generic family remains viable. In particular, the welfare comparison with competition cannot be reduced to a simple claim that refinement eliminates all but one monopoly equilibrium. \Cref{thm:nwbr} implies that, under a broad and economically natural condition, the	monopolistic environment continues to admit a whole range of robust outcomes, from the least to the most informative in the generic family. This is one reason why the welfare analysis must keep track of which benchmark is being used.
	
	Second, the result clarifies what the refinement is doing economically. For any intermediate off-path report, the sender type that survives the deletion procedure is the lower pooling type. The reason is the single-crossing structure embodied in $q_r$. Among all types that could rationalize the deviation, the lower pooling type requires the smallest probability of the favorable receiver action in order to be willing to deviate. This makes it the most plausible candidate to which the receiver can assign the deviation. Once beliefs are concentrated on that type, the receiver's best reply is pessimistic, and the deviation is no longer profitable. The generic equilibria survive not because off-path deviations are harmless in themselves, but because the refinement consistently attributes them to the type for which they are easiest to justify.
	
	At the same time, \Cref{thm:nwbr} should not be read as saying that refinement is uninformative. Rather, it identifies a broad class of environments in which $\NWBR$ does not by itself select among generic monopoly equilibria. This observation helps organize the rest of the analysis. When the single-crossing condition holds, all generic monopoly equilibria remain admissible and welfare comparisons must be made against an explicit monopoly benchmark.
	
	Before turning to that more selective case, it is useful to specialize \Cref{thm:nwbr} to the linear-quadratic environment. This specialization serves two purposes. On the one hand, it shows that the single-crossing condition is not an abstract technical requirement. It is automatically satisfied in the workhorse specification most commonly used in	applications. On the other hand, that same specification is the one employed in the later comparative-statics and welfare analysis. The next corollary links the general refinement result directly to the benchmark environment used in the remainder of the paper.
	
	
	\begin{corollary}\label{cor:lq-nwbr}
		Suppose that $u_1(\theta)=\theta-\tau_1$ and $C_1(r,\theta)=(r-\theta)^2$. Then, \Cref{ass:sc} holds. As a result, every generic monopoly equilibrium survives the Never-a-Weak-Best-Response refinement in the linear-quadratic specification.
	\end{corollary}

\Cref{cor:lq-nwbr} shows that, in the linear-quadratic environment, the single-crossing condition of \Cref{ass:sc} is automatically satisfied. Hence, in the workhorse specification used in the welfare analysis, the Never-a-Weak-Best-Response refinement does not select a unique monopoly equilibrium from the generic family. Instead, every generic monopoly equilibrium survives.
	
Even after imposing a demanding one-sender refinement, the monopolistic environment continues to admit a range of robust outcomes, from the least informative equilibrium to equilibria that are arbitrarily close to full information. Absent a stronger or cross-environment selection principle, monopoly may generate receiver welfare above the competitive adversarial benchmark simply because the receiver may coordinate on a sufficiently informative monopoly equilibrium.\footnote{The proof of \Cref{thm:nwbr} also clarifies how $\NWBR$ works in the benchmark cases. For any intermediate off-path report, the single-crossing property of $q_r$ makes the lower pooling type the unique survivor of one round of the deletion procedure. Once the receiver assigns the deviation to that type, the receiver's best reply is pessimistic and the deviation is not profitable.}

	If \Cref{ass:sc} is dropped, this uniform conclusion need no longer hold. The identity of the surviving types depends on the shape of $q_r$ for each off-path report. Appendix~\ref{app:nwbr-arbitrary-qr} records the corresponding report-by-report characterization. Those results are not part of the main selection argument of the paper. The main selection argument is developed later in Section~\ref{sec:common_refinement}, where a canonical refinement is applied across monopoly and competition.

	
	\subsubsection{Implication for the monopoly benchmark}
	\label{subsec:monopoly-selection}
	
	Within the full set of equilibria of the monopolistic environment, the generic family studied here coincides with the set of equilibria that survive the Intuitive Criterion \citep{ChoKreps1987}. Any non-generic monopolistic equilibrium is discarded on those grounds.\footnote{For a proof, see \citet{Vaccari2023ET}.} The results above show that, under \Cref{ass:sc}, the same generic family also survives the stronger $\NWBR$ refinement. In particular, in the linear-quadratic specification, $\NWBR$ does not eliminate the more informative generic monopoly equilibria.

This observation is important for the welfare comparison. It shows that the Intuitive Criterion and Never a Weak Best Response do not select among the generic monopoly equilibria under the conditions studied here. Thus, neither approach provide the selection needed for our comparison. The next section introduces a belief-based criterion that can be imposed in both market structures. The criterion selects the least-informative monopoly equilibrium and the adversarial-equilibrium outcome under competition. Its main advantage is that of providing a sharp justification for either benchmark while simultaneously placing them under a common selection discipline.
	
	The role of the present section is complementary, as it shows why a common refinement is needed. In the monopolistic environment alone, standard refinements leave too much multiplicity: the receiver may obtain anything from the least informative monopoly payoff up to the full-information payoff. The welfare comparison between monopoly and competition is necessarily ill posed unless one first specifies a selection principle that applies to both market configurations. We turn to that selection principle next.

	
\section{A common skeptical-envelope selection criterion}
\label{sec:common_refinement}

The previous section shows that standard one-sender refinements do not select a unique equilibrium from the generic monopoly family. A comparison with competition requires a selection rule that can be stated in the same way in both environments. We propose such a rule here.

The refinement is built on three simple ideas. First, the receiver should assign positive probability only to states from which the observed reports could rationally have been sent. Second, holding all other reports fixed, a higher report should make $a^+$ weakly more attractive to the receiver and, around her decision threshold, strictly more attractive. Third, an exact report or report profile that occurs with positive probability should not be interpreted more favorably to an active sender merely because it is used in equilibrium. We call such an exact report profile an \emph{atom}. A monopoly pooling report is the leading example.

The first two ideas have a natural behavioral ground. The third one has a \emph{local-robustness} interpretation. Suppose that an exact report profile may be recorded or understood with arbitrarily small imprecision. At an atom, the receiver uses the neighboring continuation value least favorable to the direction of a sender's bias. For an upward-biased sender, this is the smallest neighboring value of the receiver's expected payoff difference between $a^+$ and $a^-$. For a downward-biased sender, it is the largest. When opposed senders are active, both requirements apply. They are compatible only if all components bordering the atom deliver the same limiting value. Otherwise, the equilibrium is inadmissible. The rule does not resolve conflicting advocacy by arbitrarily favoring one sender.

This \emph{skeptical-envelope} condition is a selection assumption rather than a consequence of Bayes' rule. It imposes a local robustness discipline on the treatment of report atoms and rules out an atomicity premium in the direction favored by any active sender.

The definition below is stated for any finite set of active senders. Under monopoly, only the upward-biased sender is active, so the skeptical envelope is the lower envelope. Under two-sender competition, the active senders have opposed biases. As shown below, however, every joint report atom borders a single atom-free component, so the two directional requirements necessarily coincide.

For this section, we assume that $u_r$ is continuously differentiable, with $u_r'(\theta)>0$, on the states that can be inferred from the reports studied below. We also assume that the reach functions are continuous and strictly increasing on their relevant domains. These conditions ensure that the inverse reaches used below are well defined. As in the rest of the paper, the receiver chooses $a^+$ when she is indifferent between the two actions.


\subsection{The selection criterion}
\label{subsec:skeptical-envelope-criterion}

Let $N$ be a finite set of active senders and let $\bm r\coloneqq(r_j)_{j\in N}$ denote a report profile. We call a report $r_j$ \emph{directionally rational} at state $\theta$ if it weakly misrepresents the state in the direction of sender~$j$'s bias. That is, $r_j\geq\theta$ for an upward-biased sender, and $r_j\leq\theta$ for a downward-biased sender. Truthful reporting is included. Let $R_j(\theta)$ denote the directionally rational reports that are also within sender~$j$'s reach at state $\theta$. In the environments studied here, we obtain $R_1(\theta)=[\theta,\bar r_1(\theta)]$ and $R_2(\theta)=[\underline r_2(\theta),\theta]$. Moreover, define
\[
Q_j(r_j) \coloneqq \{\theta\in\Theta\mid r_j\in R_j(\theta)\}
\]
and
\[
P_N(\bm r) \coloneqq \bigcap_{j\in N}Q_j(r_j).
\]
Thus, $P_N(\bm r)$ contains exactly the states that can explain all the reports in $\bm r$. A state is excluded if at least one sender would have had to report in the wrong direction or incur a cost beyond that sender's reach.

The reports that are plausible when the true state is exactly the receiver's threshold, $\theta=0$, form the box
\[
Z_N \coloneqq \prod_{j\in N}R_j(0).
\]
These are the reports most relevant for equilibrium selection because they can leave the receiver uncertain about which action is better. For an upward-biased sender, $R_j(0)=[0,\bar r_j(0)]$. Likewise, for a downward-biased sender, $R_j(0)=[\underline r_j(0),0]$. To state a strict local response, we use the same box but remove each sender's farthest reachable endpoint, i.e.,
\[
K_N \coloneqq \prod_{j\in N}I_j,
\]
where
\[
I_j
\coloneqq
\begin{cases}
[0,\bar r_j(0)), & \tau_j<0,\\[3pt]
(\underline r_j(0),0], & \tau_j>0.
\end{cases}
\]
Both $Z_N$ and $K_N$ depend only on preferences and reporting costs. They are fixed before any candidate equilibrium is considered.

Consider a perfect Bayesian equilibrium, and call it $E$. Let $\nu^E$ be the ex ante distribution of report profiles induced by the prior and the senders' equilibrium strategies, and define
\[
\mathcal A(\nu^E) \coloneqq \left\{ \bm r\in Z_N \mid \nu^E(\{\bm r\})>0 \right\}.
\]
Thus, $\mathcal A(\nu^E)$ is the set of exact joint report profiles that occur with positive probability before the state is realized.\footnote{An atom is a point mass in the distribution of the full report profile observed by the receiver. A point mass in one sender's marginal report distribution is not a joint atom if another sender's report varies continuously.} There can be at most countably many such points. Let $\mathscr C(\nu^E)$ denote the set of connected components of $K_N\setminus\mathcal A(\nu^E)$.\footnote{In a one-dimensional report space, an atom can separate a lower region from an upper region. In a two-dimensional rectangle, removing isolated points need not disconnect the report space.}

Throughout this section, closures are taken relative to $Z_N$. We restrict the criterion to atom configurations satisfying
\[
\left\{ C\in\mathscr C(\nu^E) \mid \bm a\in\overline C \right\} \neq\varnothing
\]
for every $\bm a\in\mathcal A(\nu^E)$. Thus, every atom borders at least one atom-free component. This ensures that the envelope values defined below are well defined. The restriction is automatically satisfied in the monopoly and two-sender applications.

For a receiver posterior system $\mu=(\mu_{\bm r})_{\bm r}$, write
\[
V(\bm r) \coloneqq \int_\Theta u_r(\theta)\,d\mu_{\bm r}(\theta)
\]
for the receiver's expected payoff from choosing $a^+$ minus her expected payoff from choosing $a^-$. Hence, she chooses $a^+$ when $V(\bm r)\geq0$ and $a^-$ when $V(\bm r)<0$.

We now state the three requirements formally. In words, reports must be interpreted consistently with what senders could rationally report. Higher reports must favor $a^+$, and, between atoms, the expected payoff difference $V$ must rise smoothly and strictly with every report.

For every $C\in\mathscr C(\nu^E)$, define
\[
D_C \coloneqq C\cup\left(\mathcal A(\nu^E)\cap\overline C\right).
\]
Let $\bm e_j$ denote the unit vector in report coordinate $j$. For $\bm r\in D_C$, define the set of feasible coordinate increments
\[
H_j(\bm r\mid D_C) \coloneqq \left\{ h\in\mathbb R\setminus\{0\} \mid \bm r+h\bm e_j\in D_C \right\}.
\]
We say that the component-wise partial derivative of a function $g:D_C\to\mathbb R$ exists at $\bm r$ if zero is an accumulation point of $H_j(\bm r \mid D_C)$ and
\[
\partial_j^C g(\bm r) \coloneqq \lim_{\substack{h\to0\\ h\in H_j(\bm r\,\mid\, D_C)}} \frac{g(\bm r+h\bm e_j)-g(\bm r)}{h}
\]
exists. At a boundary point, this definition uses the feasible one-sided limit.

We say that $(\mu,V)$ is \emph{component-wise regular at report atoms} if it has the following properties.
\begin{enumerate}[label=($\roman*$)]
    \item \textbf{Plausible-state support.} For every $\bm r\in Z_N$, we have $\mu_{\bm r}\!\left(P_N(\bm r)\right)=1$. Thus, the receiver assigns probability zero to any state from which at least one of the observed reports could not have been sent directionally rationally and within the corresponding sender's reach. In particular, if $P_N(\bm r)=\{\theta\}$, then $\mu_{\bm r}=\delta_\theta$;

    \item \textbf{Higher reports favor $a^+$.} For every sender $j$, every fixed $\bm r_{-j}$, and every $r_j'\geq r_j$ such that both report profiles belong to $\Theta^N$, we have $V(r_j',\bm r_{-j}) \geq V(r_j,\bm r_{-j})$. Thus, increasing any one sender's report, while holding the other reports fixed, cannot make $a^+$ less attractive anywhere in the report space;

    \item \textbf{Smooth and strict response within regions and at bordering atoms.} For every $C\in\mathscr C(\nu^E)$, there is a function $\overline V_C:D_C\to\mathbb R$ that agrees with $V$ throughout $C$ and is continuous, relative to $D_C$, at every atom in $D_C$. For every $\bm r\in D_C$ and every $j\in N$, the component-wise partial derivative exists and satisfies
\[
\partial_j^C\overline V_C(\bm r)>0.
\]
Thus, every higher report makes $a^+$ strictly more attractive within an atom-free component, and the same strict response extends, as viewed from that component, through every atom bordering it.
\end{enumerate}

Requirement~$(iii)$ is imposed separately for each atom-free component. If several components border the same atom, their extensions may have different values and derivatives there. The skeptical-envelope condition determines whether those limiting values can be reconciled with skepticism toward the active senders. The requirement does not otherwise impose global differentiability across an atom that separates distinct components.

For an assessment in this domain, define the lower- and upper-envelope maps on $Z_N$ by
\begin{align*}
(\mathcal L_{\nu^E}V)(\bm r)
&\coloneqq
\begin{cases}
V(\bm r),
& \bm r\notin\mathcal A(\nu^E),\\[6pt] \displaystyle \inf_{\substack{C\in\mathscr C(\nu^E)\\ \bm r\in\overline C}} \overline V_C(\bm r),
& \bm r\in\mathcal A(\nu^E),
\end{cases}
\\[6pt]
(\mathcal U_{\nu^E}V)(\bm r)
&\coloneqq
\begin{cases}
V(\bm r),
& \bm r\notin\mathcal A(\nu^E),\\[6pt] \displaystyle \sup_{\substack{C\in\mathscr C(\nu^E)\\ \bm r\in\overline C}} \overline V_C(\bm r),
& \bm r\in\mathcal A(\nu^E).
\end{cases}
\end{align*}

A perfect Bayesian equilibrium $E$ satisfies the \emph{skeptical-envelope selection criterion}, or is \emph{skeptical-envelope admissible}, if its receiver assessment is component-wise regular at report atoms and, for every $\bm a\in\mathcal A(\nu^E)$,
\begin{equation*}
\begin{cases}
V^E(\bm a) = (\mathcal L_{\nu^E}V^E)(\bm a),
& \text{if }\tau_j<0\text{ for every }j\in N,\\[5pt] V^E(\bm a) = (\mathcal U_{\nu^E}V^E)(\bm a),
& \text{if }\tau_j>0\text{ for every }j\in N,\\[5pt] (\mathcal L_{\nu^E}V^E)(\bm a) = V^E(\bm a) = (\mathcal U_{\nu^E}V^E)(\bm a),
& \text{if there exist }j,k\in N \text{ with }\tau_j<0<\tau_k.
\end{cases}
\end{equation*}

At a non-atomic report profile, both maps leave $V$ unchanged. At an atom, they return the smallest and largest values obtained from the components bordering that point. Thus, an upward-biased monopolist faces the lower envelope, a downward-biased monopolist faces the upper envelope, and opposed senders can jointly generate an admissible atom only when the two envelope values coincide.


\subsection{Selection under monopoly}
\label{subsec:common-refinement-monopoly}

Under monopoly, we have $Z_1=[0,\bar r_1(0)]$ and $K_1=[0,\bar r_1(0))$. Consider a generic monopoly equilibrium $E_\lambda$, with pooling report $r^*(\lambda)$ and lower pooling type $\ell(\lambda)=\bar r_1^{-1}(r^*(\lambda))$. By the indexing introduced in the model section, we have that
\[
V^{E_\lambda}(r^*(\lambda))=\lambda.
\]
The equation says that, after observing the pooling report, the receiver's expected gain from $a^+$ rather than $a^-$ is $\lambda$. Because an interval of sender types all use $r^*(\lambda)$, this exact report occurs with positive probability and is an atom of $\nu^{E_\lambda}$. Every other exact report in $Z_1$ occurs with probability zero.

Because sender~$1$ is upward biased, the skeptical-envelope condition reduces in this environment to the lower-envelope condition $V^{E_\lambda} = \mathcal L_{\nu^{E_\lambda}}V^{E_\lambda}$ on $Z_1$.

Bayes' rule determines beliefs after reports used in equilibrium, but it does not determine beliefs after unused reports. A \emph{belief completion} fills in those off-path beliefs. The next proposition says that such beliefs can be filled in consistently with the criterion for $E_0$, but not for any $E_\lambda$ with $\lambda>0$.

\begin{proposition}
\label{thm:common-refinement-monopoly}
Within the generic monopoly family $\{E_\lambda\}_{\lambda\in[0,\bar\lambda]}$, the skeptical-envelope selection criterion selects the least informative equilibrium $E_0$. That is,
\begin{enumerate}[label=(\roman*)]
    \item $E_0$ has a belief completion that is skeptical-envelope admissible;
    \item no $E_\lambda$ with $\lambda>0$ is skeptical-envelope admissible.
\end{enumerate}
\end{proposition}

The logic is straightforward. In $E_0$, reports just below the pooling report lead to $a^-$, so pooled types do not want to use them. In any $E_\lambda$ with $\lambda>0$, the criterion requires reports just below the pooling report to lead to $a^+$. A pooled type can then obtain the same action with a smaller lie and a lower reporting cost. Only $E_0$ survives. Notice that the criterion does not mention the monopoly pooling report in advance. That report matters only because the equilibrium places positive probability on it.


\subsection{Selection under competition}
\label{subsec:common-refinement-competition}

With the two opposed senders, we obtain $Z_2 = [0,\bar r_1(0)] \times [\underline r_2(0),0]$ and $K_2 = [0,\bar r_1(0)) \times (\underline r_2(0),0]$. Each point in $Z_2$ is now a pair of reports: a non-negative report from the upward-biased sender, and a non-positive report from the downward-biased sender. The next lemma establishes three facts. It identifies the states that can explain any report pair, finds the two report pairs that identify $\theta=0$ uniquely, and shows that removing any countable collection of atoms does not split the remaining rectangle into separate pieces.

\begin{lemma}
\label{lem:operator-competition-geometry}
For every $(r_1,r_2)\in Z_2$,
\[
P_2(r_1,r_2) = \left[ \max\{\bar r_1^{-1}(r_1),r_2\}, \min\{r_1,\underline r_2^{-1}(r_2)\} \right].
\]
Moreover,
\[
P_2(r_1,r_2)=\{0\} \quad\Longleftrightarrow\quad (r_1,r_2) \in \{(0,0),(\bar r_1(0),\underline r_2(0))\}.
\]
Finally, for any probability distribution $\nu$ on $Z_2$, the set $K_2\setminus\mathcal A(\nu)$ is path connected and dense in $Z_2$. In particular,
\[
Z_2=\overline{K_2\setminus\mathcal A(\nu)}.
\]
\end{lemma}

The connectedness result is the key difference from monopoly. A pooling atom can divide a one-dimensional line into a lower and an upper interval. By contrast, isolated atoms do not divide the two-dimensional report rectangle. The strict relationship between reports and the receiver's expected payoff difference must hold throughout one connected region. The two corner profiles play a separate role. Because only $\theta=0$ can explain either corner, the receiver must be certain that the state is zero there.

To show that the adversarial equilibrium itself satisfies the criterion, we also need to know whether its joint distribution of reports has any atoms. The next lemma shows that no fixed pair of reports occurs with positive probability.

\begin{lemma}
\label{lem:AE-atomless}
In an adversarial equilibrium, no fixed pair of reports occurs with positive ex ante probability.
\end{lemma}

The two lemmas now connect the selection criterion to the existing characterization of the adversarial equilibrium. Connectedness makes the receiver's expected payoff difference strictly increasing in each report throughout the relevant rectangle, including along its axes. It also implies that every joint report atom borders a single component, so the lower and upper envelopes coincide there. The two corner profiles pin the receiver's expected payoff difference to zero. These are exactly the receiver-side conditions used by \citet{Vaccari2023JET}. Conversely, because the adversarial equilibrium has no joint-report atoms, both envelope maps leave its receiver assessment unchanged.

\begin{proposition}
\label{cor:common-refinement-AE}
Under the common primitives in \Cref{subsec:common-primitives},
\begin{enumerate}[label=(\roman*)]
    \item every skeptical-envelope admissible perfect Bayesian equilibrium has the adversarial equilibrium outcome;
    \item an adversarial equilibrium has a belief completion that is skeptical-envelope admissible.
\end{enumerate}
Hence, the skeptical-envelope selection criterion uniquely selects the adversarial equilibrium outcome under two-sender competition.
\end{proposition}

The proposition selects an equilibrium \emph{outcome}, not a unique posterior at every off-path report pair. Different belief completions may still be possible at zero-probability reports, but every equilibrium satisfying the criterion generates the adversarial-equilibrium outcome. Together with \Cref{thm:common-refinement-monopoly}, this gives the welfare analysis a common basis. The same selection criterion yields the least-informative monopoly equilibrium and the adversarial-equilibrium outcome under competition.


\subsection{Interpretation and scope}
\label{subsec:interpretation-scope}

The criterion is easiest to understand as a restriction on the persuasive value of a focal report. The receiver excludes states from which the observed reports would require a sender to misreport in the wrong direction or beyond his reach. She also interprets higher reports as stronger evidence in favor of $a^+$. The additional restriction concerns reports used with positive probability. Such a report should not become more persuasive merely because equilibrium play places a mass point on it. The skeptical-envelope condition rules out this \emph{atomicity premium}. An atom is evaluated using the nearby interpretation least favorable to the direction in which the sender seeks to move the receiver's action. When senders have opposed biases, neither source is given priority. The two requirements are compatible only when the neighboring interpretations agree.

Under monopoly, the receiver observes a single claim from a source that would like her to choose $a^+$ more often. In a generic equilibrium, a range of states may pool on the same report. Because reports lie on a line, this pooling report separates lower claims from higher ones, and the receiver's beliefs can approach it differently from the two sides. Granting the pooling report a favorable jump in its persuasive value can sustain a more informative equilibrium. The skeptical-envelope condition rules out that jump. If the pooling report induced $a^+$ with a strictly positive payoff advantage for the receiver, reports just below it would have to do so as well. A pooled type could obtain the same action with a smaller lie and a lower reporting cost. This logic eliminates every $E_\lambda$ with $\lambda>0$ and leaves the least informative equilibrium $E_0$.

Competition changes this inference problem. The receiver now observes two claims, chosen independently by sources that favor opposite actions. Taken together, the claims restrict the states from which both reports could have been made. Moreover, the set of report pairs is two-dimensional. An isolated report pair can be approached from many directions and does not divide the report space into a lower and an upper region. The receiver cannot attach a favorable jump to a joint atom on one side of the market while remaining skeptical of the source on the other side. The neighboring interpretations must agree. Combined with the two report pairs that can arise only at $\theta=0$, this discipline delivers the receiver-side conditions that characterize the adversarial-equilibrium outcome. The economic force behind the result is that she can compare independently chosen claims made by sources with opposing interests.

This argument treats the dimension of the report space as an economic object, not as a choice of notation. If a monopolist's report $r$ were merely recorded as $(r,h(r))$, the two coordinates would still represent a single reporting decision. The feasible pairs would lie on a one-dimensional graph, and a pooling atom could still separate reports on that graph. By contrast, allowing the second coordinate to be chosen independently and observed by the receiver would create an additional signaling instrument and hence a different economic game, even if that coordinate did not enter payoffs directly.\footnote{More generally, skeptical-envelope admissibility is invariant to a strictly increasing $C^1$ reparameterization of each sender's report with strictly positive derivative. Such a reparameterization preserves atoms, connected components, and the signs of the responsiveness conditions.}

The criterion does not require the receiver to distrust all reports, nor does it assert that local skepticism is the only reasonable way to complete beliefs. Our welfare comparisons are conditional on this selection discipline. The common criterion places our two market configurations of interest under the same discipline and thereby provides a coherent basis for comparing their welfare consequences.

We restrict the formal applications to the one- and two-sender environments used in the welfare analysis. Although the criterion can be stated for any finite number of senders, extending the selection result to three or more senders would require a separate equilibrium characterization.

	
	\section{Receiver welfare: monopoly versus competition}\label{sec:welfare}

The previous section provides a common basis for comparing market structures. The skeptical-envelope criterion selects the least-informative equilibrium under monopoly, and the adversarial-equilibrium outcome under two-sender competition. This section compares the receiver's welfare under these selected outcomes. When we consider other monopoly equilibria, we state this explicitly. All competitive objects used below refer to the essentially unique adversarial-equilibrium outcome in \Cref{prop:AE-characterization}. They do not depend on which admissible off-path belief completion is chosen.

Competition has two opposing effects on the receiver's decisions. Under the selected monopoly outcome, sender~1 pools a range of negative states and induces the receiver to choose $a^+$ throughout that range. Introducing an oppositely biased sender allows the receiver to compare conflicting reports and reduces this systematic pooling error. Competition, however, also creates new mistakes. Strategic disagreement may lead the receiver to choose $a^-$ in positive states. Whether competition benefits the receiver depends on whether the reduction in the monopoly pooling loss outweighs the additional errors generated by competition.

Two conclusions organize the analysis. First, competition does not dominate every monopoly equilibrium. As we have already discusses, when the receiver is sufficiently skeptical, monopoly becomes highly informative and can yield greater welfare than the adversarial equilibrium. Second, competition can improve upon the selected least-informative monopoly benchmark. We establish this ranking in a symmetric environment and, under suitable regularity conditions, whenever the additional sender faces a sufficiently high misreporting cost. In the latter case, competitive mistakes vanish as misreporting becomes more costly, while the pooling loss under selected monopoly remains.

We proceed by first characterizing receiver welfare across the family of monopoly equilibria. We then derive tractable bounds on welfare in the adversarial equilibrium and use them to establish the symmetric comparison. Next, an exact decomposition separates the competitive mistakes made in positive and negative states and delivers the high-cost result. Finally, we specialize the analysis to the linear-quadratic class to study how the additional sender's bias and misreporting cost affect the equilibrium objects that govern the welfare comparison.
	
	
	\subsection{Monopoly welfare}
	
	Under monopoly, the skepticism parameter $\lambda$ determines the size of the pooling region and hence the receiver's welfare loss. Expressing welfare relative to the full-information benchmark makes this relationship immediate and prepares the comparison with the adversarial equilibrium.

	It is convenient to measure monopoly welfare relative to the full-information benchmark. Let
	\[
	W_r^{\FI}\coloneqq \int_{\{\theta\geq 0\}} u_r(\theta)f(\theta)\,d\theta
	\]
	denote receiver welfare under full information.
	\begin{proposition}\label{prop:mono-welfare}
		For every monopoly equilibrium indexed by $\lambda$, we have that
		\[
		W_r^{\M}(\lambda)=W_r^{\FI}+\underbrace{\int_{\ell(\lambda)}^0 u_r(\theta)f(\theta)\,d\theta}_{\leq 0}.
		\]
		In particular, $W_r^{\M}(\lambda)\leq W_r^{\FI}$, with equality if and only if $\lambda=\bar\lambda$.
	\end{proposition}

	
	\Cref{prop:mono-welfare} provides a simple representation of the receiver's loss under monopoly. Relative to full information, mistakes occur only in the negative states contained in the pooling interval $[\ell(\lambda),0)$. In those states, pooling induces the receiver to choose $a^+$ rather than $a^-$. As skepticism increases and $\ell(\lambda)$ approaches zero, the pooling interval shrinks and the associated welfare loss falls. At $\lambda=\bar\lambda$, the pooling interval disappears and the receiver obtains her full-information payoff.

It naturally follows that competition cannot dominate every monopoly equilibrium. When the receiver is sufficiently skeptical, monopoly becomes arbitrarily close to full information, whereas the adversarial equilibrium remains only partially revealing. The next proposition formalizes this observation.

	\begin{proposition}\label{thm:high-skepticism}
		Fix the adversarial equilibrium of the competitive game and let $W_r^{\DE}$ denote the receiver's welfare in that equilibrium. There exists $\lambda_0<\bar\lambda$ such that $W_r^{\M}(\lambda)>W_r^{\DE}$ for every $\lambda\in(\lambda_0,\bar\lambda]$. 
	 \end{proposition}
	
 \Cref{thm:high-skepticism} establishes an important limitation on the welfare case for competition. As $\lambda$ approaches $\bar\lambda$, monopoly welfare approaches $W_r^{\FI}$. Because welfare in the adversarial equilibrium is strictly below this benchmark, sufficiently informative monopoly equilibria must outperform competition.

Our main comparison is between the outcomes selected by the common skeptical-envelope criterion, that is, the adversarial equilibrium under competition and the least-informative equilibrium under monopoly. We proceed by deriving tractable bounds on welfare in the adversarial equilibrium. These bounds will allow us to compare competition with the selected monopoly benchmark without relying immediately on the full expression for competitive welfare.

	
	\subsection{Welfare bounds in the adversarial equilibrium}
	
	Receiver welfare in the adversarial equilibrium can be computed exactly, as we show in \Cref{subsec:exact-welfare}. For the present comparison, however, the exact expression contains more detail than is needed because it depends on the full pattern of mistakes generated by the senders' mixed strategies. Therefore, we begin with bounds that depend on only a few equilibrium objects.

The bounds have simple interpretations. The lower bound evaluates a feasible decision rule under which the receiver ignores sender~2 and follows the sign of sender~1's report. Because the receiver's equilibrium decision is optimal given both reports, her equilibrium welfare cannot be lower than the payoff generated by this rule. The upper bound instead counts one event on which the receiver necessarily makes a mistake: sender~1 misreports while sender~2 reports truthfully in a negative state. Ignoring all other possible mistakes provides an upper bound on welfare.

These bounds are conservative because they do not use all the information contained in the two reports. Their advantage is tractability, as they express the relevant welfare losses in terms of the truthful cutoff and the senders' truthful-reporting probabilities. This will yield a transparent sufficient condition under which competition benefits the receiver.
	
	\begin{proposition}\label{prop:bounds}
		The receiver's welfare in the adversarial equilibrium satisfies
		\[
		\underline W_r^{\DE}\leq W_r^{\DE}\leq \overline W_r^{\DE},
		\]
		where
		\[
		\underline W_r^{\DE}\coloneqq W_r^{\FI}+\int_{\theta_L}^0 u_r(\theta)f(\theta)\left(1-\alpha_1(\theta)\right)\,d\theta,
		\]
		and
		\[
		\overline W_r^{\DE}\coloneqq W_r^{\FI}+\int_{\theta_L}^0 u_r(\theta)f(\theta)\left(1-\alpha_1(\theta)\right)\alpha_2(\theta)\,d\theta.
		\]
		Moreover, if $\theta_L<0$ and $\alpha_1,\alpha_2\in(0,1)$ on a set of positive measure, then
		\[
		\underline W_r^{\DE}<\overline W_r^{\DE}<W_r^{\FI}.
		\]
	\end{proposition}

For the welfare comparison, the lower bound is the key object. Under the least informative monopoly equilibrium, the receiver bears the full pooling loss over the relevant negative states. Under competition, sender~1 reports truthfully with positive probability in the negative mixing region. The feasible rule used to construct the lower bound allows the receiver to avoid some of the mistakes that she makes under monopoly. If this lower bound already exceeds monopoly welfare at $\lambda=0$, equilibrium welfare under competition must do so as well. The resulting condition is sufficient rather than necessary because the bound disregards some of the information supplied by sender~2.

The upper bound serves a different purpose. Whenever sender~1 misreports and sender~2 reports truthfully in the relevant negative states, the adversarial-equilibrium decision is incorrect. This unavoidable source of error shows why welfare remains strictly below the full-information benchmark and provides an upper limit on the receiver's payoff. We now use the lower bound to compare competition with the selected monopoly outcome.

	
	\subsection{Competition versus the least-informative monopoly equilibrium}
	
	We now compare the adversarial equilibrium with the least-informative monopoly benchmark. Under selected monopoly, the receiver chooses $a^+$ throughout $[\ell^M,0)$, even though $a^-$ is optimal in these negative states. The lower bound in \Cref{prop:bounds} shows how competition can reduce this loss.

The comparison turns on the left truthful cutoff $\theta_L$. Under the feasible decision rule used to construct the lower bound, the receiver avoids the monopoly error entirely on $[\ell^M,\theta_L)$ whenever $\theta_L>\ell^M$. Within the remaining interval $(\theta_L,0)$, sender~1 reports truthfully with probability $\alpha_1(\theta)$, allowing the receiver to avoid the monopoly error at least some of the time. Thus, if $\theta_L\geq\ell^M$ and truthful reporting occurs on a non-negligible range of states, even this conservative lower bound exceeds receiver welfare under selected monopoly. The following lemma formalizes this sufficient condition.
	
	\begin{lemma}
		\label{lem:cutoff-sufficient-condition}
		Suppose that the adversarial equilibrium has left truthful cutoff $\theta_L$ satisfying $\theta_L\geq \ell^M$. Suppose also that sender $1$'s truthful-reporting probability satisfies $\alpha_1(\theta)>0$	on a subset of $(\theta_L,0)$ with positive measure whenever $\theta_L<0$. Then,
		\[
		W_r^{\DE}>W_r^{\M}(0).
		\]
	\end{lemma}

    The lemma is stated in terms of the endogenous cutoffs $\theta_L$ and $\ell^M$, so it does not yet provide a condition expressed entirely in primitives. The symmetric benchmark supplies such a condition. Symmetry makes the negative competitive cutoff coincide with the lower cutoff of the monopoly pooling interval, so that $\theta_L=\ell^M$. The welfare ranking follows directly from the lemma.

	\begin{proposition}\label{thm:symmetric}
		Assume the following symmetric benchmark,
		\begin{enumerate}[label=(\roman*)]
			\item $f(\theta)=f(-\theta)$ for every $\theta$;
			\item $u_r(\theta)=-u_r(-\theta)$ for every $\theta$;
			\item $u_2(\theta)=-u_1(-\theta)$, so $\tau_2=-\tau_1$;
			\item $k_1C_1(r,\theta)=k_2C_2(-r,-\theta)$ for every $r,\theta\in\Theta$.
		\end{enumerate}
		Then, $\theta_L=\ell^M$, and thus $\underline W_r^{\DE}>W_r^{\M}(0)$.
	\end{proposition}

\Cref{thm:symmetric} provides a clean welfare case for competition, but its scope should be kept in view. Specifically, that result does not imply that competition eliminates all decision errors or that it dominates more informative monopoly equilibria. Rather, it shows that, under symmetric primitives, even the conservative competitive lower bound strictly exceeds receiver welfare in the selected least-informative monopoly equilibrium. The next subsection examines the competitive errors directly by deriving exact receiver welfare.

	

	\subsection{Exact receiver welfare in the adversarial equilibrium}
	\label{subsec:exact-welfare}
	
	The preceding bounds provide tractable sufficient conditions for the welfare comparison, but they do not describe the receiver's equilibrium mistakes exactly. We now derive an exact decomposition of receiver welfare into losses from mistakes in positive and negative states. Besides clarifying the welfare trade-off generated by competition, this decomposition provides the basis for the high-misreporting-cost result in the next subsection.

The decomposition uses parts~$(iii)$ and~$(iv)$ of \Cref{prop:AE-characterization}. Conditional on the state, each sender has a possible truthful atom and an otherwise atomless reporting component. The independent-private-randomization convention in \Cref{subsec:common-primitives} makes the probabilities of the two senders' reporting events multiplicative.
	
	Let $\theta_H>0$ denote the right truthful cutoff, implicitly defined by $s(\theta_H)=\underline{r}_2 (\theta_H)$. Thus, for every $\theta\geq \theta_H$, both senders report truthfully and the receiver chooses the positive action with probability one. Likewise, for every $\theta\leq \theta_L$, both senders report truthfully and the	receiver chooses the negative action with probability one.
	
	For $\theta\in[0,\theta_H)$, let $\psi_1(\cdot\mid\theta)$ denote the density of the continuous component of sender~1's report distribution, and let $\Psi_2(\cdot\mid\theta)$ denote the sub-distribution function of the continuous component of sender~2's report distribution. These components are not normalized. For sender~$j$, their total mass is $1-\alpha_j(\theta)$. For $\theta\in(\theta_L,0]$, define $\psi_2(\cdot\mid\theta)$ and $\Psi_1(\cdot\mid\theta)$ analogously.
	
	Define the conditional mistake probabilities
	\begin{align}
		m_+(\theta)
		&\coloneqq 
		\alpha_1(\theta)\left(1-\alpha_2(\theta)\right) +\int_{\theta}^{\,s(\underline{r}_2 (\theta))} \Psi_2\!\left(s(r)\mid \theta\right)\psi_1(r\mid \theta)\,dr \; \text{ for }\; \theta\in[0,\theta_H),
		\label{eq:mplus-def}
		\\
		m_-(\theta)
		&\coloneqq 
		\alpha_2(\theta)\left(1-\alpha_1(\theta)\right) \nonumber\\
		& +\int_{s(\bar r_1(\theta))}^{\,\theta} \left[1-\alpha_1(\theta)-\Psi_1\!\left(s(r)\mid \theta\right)\right]\psi_2(r\mid \theta)\,dr \; \text{ for }\; \theta\in(\theta_L,0].
		\label{eq:mminus-def}
	\end{align}
	Here, $m_+(\theta)$ is the conditional probability that the receiver incorrectly chooses $a^-$ in a positive state, whereas $m_-(\theta)$ is the conditional probability that she incorrectly chooses $a^+$ in a negative state. In each expression, the first term covers the case in which one sender reports truthfully and the other misreports. The integral covers mistakes that arise when both reports are drawn from their continuous components and fall on the relevant side of the swing function.

The receiver's exact welfare is obtained by subtracting the expected loss from positive-state mistakes and the expected loss from negative-state mistakes from the full-information benchmark. The following proposition states this decomposition.
	\begin{proposition}
		\label{prop:exact-welfare}
		Receiver welfare in the adversarial equilibrium is
		\begin{equation}
			W_r^{\DE} = W_r^{\FI} -\int_{0}^{\theta_H} u_r(\theta)f(\theta)m_+(\theta)\,d\theta +\int_{\theta_L}^{0} u_r(\theta)f(\theta)m_-(\theta)\,d\theta,
			\label{eq:exact-welfare}
		\end{equation}
		or, alternatively,
		\begin{equation}
			W_r^{\DE} = \int_{\theta_H}^{\sup\Theta} u_r(\theta)f(\theta)\,d\theta +\int_{0}^{\theta_H} u_r(\theta)f(\theta)\left[1-m_+(\theta)\right]\,d\theta +\int_{\theta_L}^{0} u_r(\theta)f(\theta)m_-(\theta)\,d\theta.
			\label{eq:exact-welfare-alt}
		\end{equation}
	\end{proposition}


	
	The exact formula yields a useful loss decomposition. Define
	\begin{align*}
		\mathcal L_+^{\DE}
		&\coloneqq 
		\int_{0}^{\theta_H} u_r(\theta)f(\theta)m_+(\theta)\,d\theta,
		\\
		\mathcal L_-^{\DE}
		&\coloneqq 
		-\int_{\theta_L}^{0} u_r(\theta)f(\theta)m_-(\theta)\,d\theta,
		\\
		\mathcal L^{\M}
		&\coloneqq 
		-\int_{\ell^M}^{0} u_r(\theta)f(\theta)\,d\theta.
	\end{align*}
	Because $u_r(\theta)>0$ for $\theta>0$ and $u_r(\theta)<0$ for $\theta<0$, all three loss terms are non-negative. The term $\mathcal L_+^{\DE}$ (resp.~$\mathcal L_-^{\DE}$) is the expected loss generated by mistakes in positive (resp.~negative) states under competition. The term $\mathcal L^{\M}$ is the expected pooling loss in the least informative	monopoly equilibrium.
	
	\begin{corollary}
		\label{cor:exact-comparison}
		The receiver is better off under competition than in the least informative monopoly equilibrium if and only if
		\begin{equation}
			\mathcal L_+^{\DE}+\mathcal L_-^{\DE}<\mathcal L^{\M}.
			\label{eq:exact-comparison}
		\end{equation}
	\end{corollary}

	The previous corollary gives an exact necessary-and-sufficient criterion for the welfare ranking. Its limitation is that the mistake probabilities $m_+$ and $m_-$ are themselves endogenous objects, because they depend on the equilibrium swing report and on the senders' equilibrium mixed strategies. Still, the formula is useful in two ways. First, it shows precisely which pieces of the competitive equilibrium matter for welfare. Second, it clarifies why the lower-bound approach is conservative. Indeed, it collapses the two exact loss terms $\mathcal L_+^{\DE}$ and $\mathcal L_-^{\DE}$ into a simpler object that is easier to compare with monopoly but inevitably discards information.
	
	\Cref{prop:exact-welfare} gives the exact welfare formula in full generality. The expression is informative, but still quite rich, because it keeps track of the entire mistake structure of the adversarial equilibrium. Before using it for further comparative-statics arguments, it is useful to specialize it to the symmetric benchmark.
	
	This specialization serves two purposes. First, it shows how the exact formula collapses in the most transparent case, making the welfare logic easier to read directly from the expression. Second, it provides a sharp benchmark against which the earlier bound-based results can be interpreted. In the symmetric environment, the exact comparison becomes especially simple, and the next corollary records that reduction.
	
	\begin{corollary}
		\label{cor:exact-symmetric}
		Suppose the environment is symmetric, in the sense of \Cref{thm:symmetric}. Then, $\theta_H=-\theta_L$ and $m_+(\theta)=m_-(-\theta)$ for every $\theta\in[0,\theta_H)$. As a result, we obtain that
		\begin{equation}
			\mathcal L_+^{\DE}=\mathcal L_-^{\DE},
			\label{eq:symmetric-losses}
		\end{equation}
		and exact receiver welfare simplifies to
		\begin{equation}
			W_r^{\DE} = W_r^{\FI} + 2\int_{\theta_L}^{0} u_r(\theta)f(\theta)m_-(\theta)\,d\theta.
			\label{eq:exact-welfare-symmetric}
		\end{equation}
		Therefore, in the symmetric benchmark we obtain that
		\[
		W_r^{\DE}>W_r^{\M}(0) \quad\Longleftrightarrow\quad 2\int_{\theta_L}^{0} \left[-u_r(\theta)\right]f(\theta)m_-(\theta)\,d\theta < \int_{\ell^M}^{0} \left[-u_r(\theta)\right]f(\theta)\,d\theta.
		\]
	\end{corollary}

	
	\begin{figure}
	    \centering
        \includegraphics[width=0.9\linewidth]{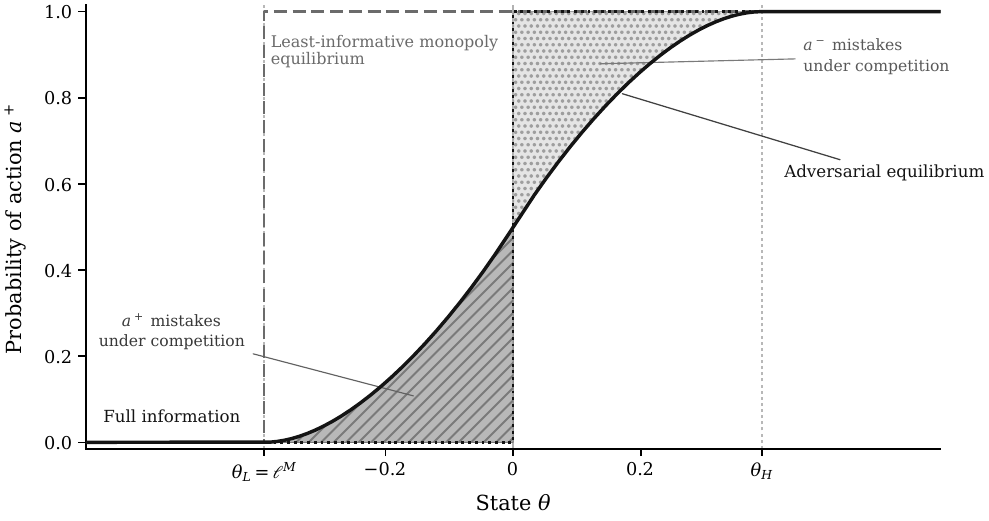}
	    \caption{Decision probabilities in the symmetric LQ benchmark. The figure assumes a standard normal prior, $(\tau_1,k_1)=(-1,1)$, and $(\tau_2,k_2)=(1,1)$. The dotted and dashed step functions represent, respectively, the receiver's full-information rule and the least-informative monopoly equilibrium, while the solid curve is $q^{\DE}(\theta)$, the probability of action $a^+$ in the adversarial equilibrium. The cutoffs satisfy $\theta_L=\ell^M=-\theta_H=-(\sqrt{17}-1)/8$. Selected monopoly chooses $a^+$ with probability one throughout $(\ell^M,0)$. The shaded regions show the two competitive decision errors, i.e., choosing $a^+$ in negative states and choosing $a^-$ in positive states.}
    \label{fig:decision-probabilities}
	\end{figure}

Figure~\ref{fig:decision-probabilities} illustrates the discipline-noise trade-off generated by competition. Under the selected monopoly outcome, the receiver chooses $a^+$ with probability one throughout $(\ell^M,0)$ and makes a systematic error in those negative states. Competition reduces this negative-state error but introduces mistakes in positive states. In the symmetric benchmark depicted in the figure, the former gain exceeds the latter loss, so competition increases receiver welfare.

The earlier lower bound should not be interpreted as simply dropping $\mathcal L_+^{\DE}$ from the exact decomposition. It evaluates a feasible decision rule under which the receiver ignores sender~2 and follows the sign of sender~1's report. This rule avoids positive-state mistakes but uses only coarse information in negative states. Because the receiver's equilibrium decision optimally uses both reports, equilibrium welfare must be at least as large as the payoff from this feasible rule. The exact decomposition instead records both types of equilibrium mistake and is therefore the appropriate basis for the comparative-static result that follows.
	
	
	\subsection{A sufficiently reliable additional sender}
	\label{subsec:high-cost-sender}
	
	\Cref{cor:exact-comparison} reduces the welfare comparison to whether the two competitive losses satisfy $\mathcal L_+^{\DE}+\mathcal L_-^{\DE}<\mathcal L^{\M}$. Although exact, this condition is expressed in terms of endogenous equilibrium objects. We now obtain a sufficient condition stated directly in terms of sender~$2$'s misreporting-cost parameter.

Fix sender~$2$'s bias $\tau_2$ and let $k_2$ grow. At every fixed state away from the receiver's cutoff, both senders eventually report truthfully. Under the regularity conditions stated below, the resulting disappearance of competitive mistakes implies that both loss terms converge to zero. By contrast, the pooling loss under the least-informative monopoly equilibrium remains strictly positive whenever its pooling region is non-degenerate. This difference yields the high-cost welfare comparison.
	
	Because the mixing regions vary with $k_2$, we first extend the conditional mistake probabilities by zero outside those regions. This allows the two welfare losses to be written over fixed state domains. For every competitive environment $(\tau_2,k_2)$, define\footnote{Because the mixing regions vary with $k_2$, we make explicit the dependence of the conditional mistake probabilities defined in \eqref{eq:mplus-def} and \eqref{eq:mminus-def} on $(\tau_2,k_2)$ and extend them by zero outside their original domains. This allows the two welfare losses to be written over fixed state domains.}
	\[
	\widetilde m_+(\theta\mid\tau_2,k_2)
	=
	\begin{cases}
		m_+(\theta\mid\tau_2,k_2), & \theta\in[0,\theta_H(\tau_2,k_2)),\\
		0, & \theta\geq \theta_H(\tau_2,k_2),
	\end{cases}
	\]
	and
	\[
	\widetilde m_-(\theta\mid\tau_2,k_2)
	=
	\begin{cases}
		m_-(\theta\mid\tau_2,k_2), & \theta\in(\theta_L(\tau_2,k_2),0],\\
		0, & \theta\leq \theta_L(\tau_2,k_2).
	\end{cases}
	\]
	Then,
	\[
	\mathcal L_+^{\DE}(\tau_2,k_2)
	=
	\int_{0}^{\sup\Theta}
	u_r(\theta)f(\theta)\widetilde m_+(\theta\mid\tau_2,k_2)\,d\theta,
	\]
	and
	\[
	\mathcal L_-^{\DE}(\tau_2,k_2)
	=
	-\int_{\inf\Theta}^{0}
	u_r(\theta)f(\theta)\widetilde m_-(\theta\mid\tau_2,k_2)\,d\theta.
	\]
	
The next lemma establishes the required point-wise convergence. The proposition converts that convergence into the disappearance of the expected welfare
losses.
	
	\begin{lemma}
\label{lem:high-k2-convergence}
Fix sender~$1$'s primitives, and a finite bias $\tau_2>0$ for sender~$2$. Then, for every fixed $\theta\neq0$, both senders report truthfully for all sufficiently large $k_2$. As a result, $q^{\DE}(\theta\mid\tau_2,k_2)\to \mathds{1}_{\{\theta\geq0\}}$ for $\theta\neq0$, the two conditional mistake probabilities converge pointwise to zero, and the expected misreporting cost of each sender converges pointwise to zero.
\end{lemma}

The lemma establishes the state-by-state effect of a higher $k_2$. Away from the receiver's cutoff, sufficiently costly misreporting eventually induces both senders to report truthfully, and the receiver takes her full-information action. Establishing the welfare comparison requires one additional step. We must show that the receiver's expected losses, aggregated across states, also disappear. The next proposition provides this step and uses it to compare competition with the positive pooling loss that remains under selected monopoly.

\begin{proposition}
	\label{prop:high-k2-threshold}
	Fix sender~$1$'s primitives, and a finite bias $\tau_2>0$ for sender~$2$. Suppose that the least informative monopoly equilibrium pools a non-degenerate interval, so that $\mathcal L^{\M}>0$. Assume also that $|u_r(\theta)|f(\theta)$ is integrable on $\Theta$, that $u_r$ is continuous at $0$, with $u_r(0)=0$, and that the high-cost convergence conclusion of \Cref{lem:high-k2-convergence} holds. Then,
	\[
	\lim_{k_2\to\infty}\mathcal L_+^{\DE}(\tau_2,k_2)=\lim_{k_2\to\infty}\mathcal L_-^{\DE}(\tau_2,k_2)=0.
	\]
	Thus, there exists a finite threshold $\bar k_2(\tau_2)$ such that $W_r^{\DE}(\tau_2,k_2)>W_r^{\M}(0)$	for all $k_2\geq \bar k_2(\tau_2)$.
\end{proposition}


	
\Cref{prop:high-k2-threshold} provides the main receiver-welfare case for competition. For every fixed finite bias of sender~$2$, sufficiently costly misreporting makes the selected competitive outcome arbitrarily close to the receiver's full-information benchmark and eventually eliminates the fixed pooling advantage of selected monopoly. The proposition establishes eventual dominance, and does not require welfare to be monotone in $k_2$ or provide an explicit expression for the threshold.

The argument relies only on the monopoly equilibrium having a strictly positive pooling loss. It applies point-wise to each non-revealing generic monopoly equilibrium. The associated threshold is allowed to depend on the equilibrium index $\lambda$.
	
	\begin{corollary}
		\label{cor:high_cost}
		Fix sender $1$'s primitives and a finite bias $\tau_2>0$ for sender $2$. Assume that $|u_r(\theta)|f(\theta)$ is integrable on $\Theta$, that $u_r$ is continuous at $0$, with $u_r(0)=0$, and that the high-cost convergence conclusion of \Cref{lem:high-k2-convergence} holds. Then, for every monopoly equilibrium indexed by $\lambda<\bar\lambda$, there exists a finite threshold such that $W_r^{\DE}(\tau_2,k_2)>W_r^{\M}(\lambda)$ for every $k_2\geq\bar k_2(\tau_2,\lambda)$.
	\end{corollary}

The quantifiers in \Cref{cor:high_cost} are point-wise in $\lambda$. By \Cref{thm:high-skepticism}, for every fixed finite competitive environment $(\tau_2,k_2)$, generic monopoly equilibria with $\lambda<\bar\lambda$ sufficiently close to $\bar\lambda$ yield higher receiver welfare than competition. Hence, no single finite cost threshold can make competition dominate the entire non-revealing generic monopoly family at once.
	
	
	\subsection{The linear-quadratic class}\label{sec:LQ}
	
	The high-$k_2$ result provides a welfare case for competition. We now use the linear-quadratic (LQ) specification to examine how sender~2's bias and misreporting cost affect the competitive equilibrium. The central object is the swing report, i.e., the report by sender~2 that exactly offsets a given report by sender~1 and leaves the receiver indifferent between her two actions.

The LQ specification delivers explicit reach functions and a tractable equation for the swing report. We proceed in three steps. First, we derive the swing equation and study how the swing report responds to sender~2's bias and misreporting cost. Second, we use these comparative statics to extend the symmetric welfare comparison. Finally, we study the non-binding inverse-reach condition under which the effect of sender~2's bias can be signed. These results clarify the forces behind the welfare comparison, but they do not produce a global welfare ranking, as receiver welfare also depends on the equilibrium cutoffs, mixing probabilities, and resulting mistakes.\footnote{The LQ specification is used for tractability. The Appendix shows which reach and swing-report comparative statics extend to broader classes of distance-based misreporting costs, including standard power costs under additional curvature restrictions.}

	
	\subsubsection{Linear-quadratic primitives and the swing equation}\label{sec:LQ_class}
	
	In the LQ case, the senders' reaches are explicit and the receiver's indifference condition can be written as a single equation for the swing report. This allows us to trace how sender~2's bias and misreporting cost change the evidence required to overturn sender~1's report.
	
	\begin{definition}\label{def:LQ}
		The environment is linear-quadratic (LQ) if $u_i(\theta)=\theta-\tau_i$ for $i\in\{1,2,r\}$, and $C_j(r_j,\theta)=(r_j-\theta)^2$ for $j\in\{1,2\}$, with $\tau_r=0$, $\tau_1<0<\tau_2$, and $k_j>0$.
	\end{definition}

	Throughout this subsection, maintain the LQ specification of \Cref{def:LQ}. For $\theta\geq\tau_1$, the reach of sender~$1$ is
	\[
	\bar r_1(\theta)=\theta+\sqrt{\frac{\theta-\tau_1}{k_1}},
	\]
	while, on the image of sender~$2$'s lower reach, its inverse reach is
	\begin{equation}\label{eq:r2-inv}
		\underline{r}_2 ^{-1}(r)=r+\frac{\sqrt{1+4k_2(\tau_2-r)}-1}{2k_2}.
	\end{equation}
	The swing report $s(r)$ is implicitly pinned down by the equilibrium conditions. Fix $r\in(0,\bar r_1(0)]$ and let $s(r)\in[\underline{r}_2 (0),0)$ be the report that makes the receiver indifferent between $a^+$ and $a^-$. Define
	\begin{equation}\label{eq:G}
		G(r,s\mid \tau_2,k_2)\coloneqq \int_{\max\{s,\bar r_1^{-1}(r)\}}^{\min\{r,\underline{r}_2 ^{-1}(s)\}} \frac{4\theta f(\theta)(s-\theta)(r-\theta)}{(\theta-\tau_1)(\theta-\tau_2)}\,d\theta.
	\end{equation}
	Then, $s(r)$ is implicitly characterized by $G\left(r,s(r)\mid \tau_2,k_2\right)=0$. Whenever the partial derivative $G_s$ is nonzero, the implicit function theorem yields
	\[
	\frac{\partial s(r)}{\partial \tau_2}=-\frac{G_{\tau_2}(r,s(r))}{G_s(r,s(r))},
	\] 
	\[
	\frac{\partial s(r)}{\partial k_2}=-\frac{G_{k_2}(r,s(r))}{G_s(r,s(r))}.
	\]
	
To sign these comparative statics, we first need the swing equation to be locally well behaved. The next lemma establishes that the derivative with respect to the swing report is strictly positive at every non-degenerate swing root, namely for $r\in(0,\bar r_1(0))$. The relevant root is locally unique, and the implicit function theorem applies. At the endpoint $r=\bar r_1(0)$, the state interval consistent with the swing pair collapses to the singleton $\{0\}$ and the corresponding derivative is zero. We handle the endpoint comparative static directly from the reach formula.

We evaluate derivatives at points where the integration bounds are locally differentiable. At a boundary-switching point, where $\underline r_2^{-1}(s(r))=r$, the swing root may have a kink. The corresponding results should then be interpreted using one-sided derivatives and continuity.
	
	
	\begin{lemma}\label{prop:Gs-positive}
    For every $r\in(0,\bar r_1(0))$, we have $G_s(r,s(r))>0$. At the endpoint, we obtain $G_s\!\left(\bar r_1(0),s(\bar r_1(0))\right)=0$, where the derivative is understood in the feasible one-sided sense.
	\end{lemma}
	

	
	We begin with sender~2's misreporting cost. The parameter $k_2$ enters the swing equation only through sender~2's inverse reach. Its effect depends on whether that reach constrains the states consistent with the observed reports. The next proposition states the resulting comparative static.
	
	\begin{proposition}
\label{prop:cost-effect}
Fix $\tau_2>0$ and $r\in(0,\bar r_1(0)]$. The adversarial-equilibrium swing report $s_{k_2}(r)$ is weakly increasing in $k_2$. At every differentiability point,
\[
\frac{\partial s_{k_2}(r)}{\partial k_2}=0 \quad\text{if }\underline r_2^{-1}(s_{k_2}(r))>r,
\]
and
\[
\frac{\partial s_{k_2}(r)}{\partial k_2}>0 \quad\text{if }\underline r_2^{-1}(s_{k_2}(r))<r.
\]
At a boundary-switching point, the corresponding one-sided derivatives are weakly non-negative whenever they exist.
\end{proposition}

	
	The cost effect is a boundary one. When sender~2's inverse reach is non-binding, a marginal change in $k_2$ leaves the swing report unchanged. When it binds, a higher misreporting cost raises the swing report. Bias operates differently. The parameter $\tau_2$ enters the indifference condition directly through the integrand and can affect the swing report even when sender~2's inverse reach is non-binding. The next proposition signs this interior effect.
	
	\begin{proposition}\label{prop:bias-effect}
		Fix $r\in(0,\bar r_1(0))$. Suppose that sender~$2$'s inverse reach is locally non-binding at the swing report, namely $\underline r_2^{-1}(s(r))>r$. Then, at every differentiability point of the swing root,
		\[
		\frac{\partial s(r)}{\partial \tau_2}>0.
		\]
	\end{proposition}


	
	\begin{figure}
		\centering
		\includegraphics[width=0.8\linewidth]{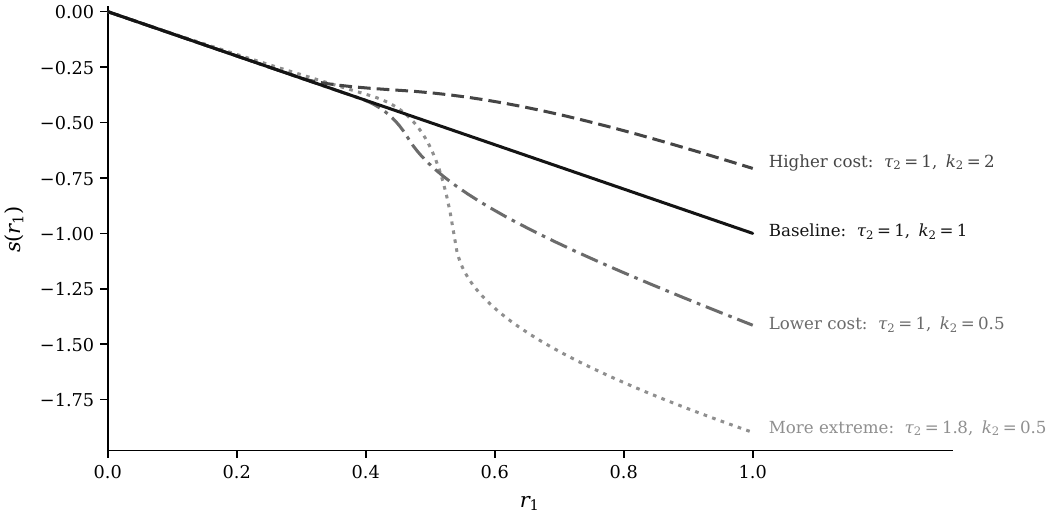}
		\caption{Swing report functions in the LQ class. The figure assumes a standard normal prior for $\theta$, and $(\tau_1,k_1)=(-1,1)$. It plots $s(r_1)$ for four competitive environments.}
		\label{fig:swing}
	\end{figure}
	

\Cref{prop:bias-effect} is a local result about the receiver's indifference condition. It does not imply that receiver welfare is globally increasing in $\tau_2$, because welfare also depends on the truthful cutoffs, equilibrium mixing, and mistake probabilities.

Figure~\ref{fig:swing} illustrates these comparative statics. In the symmetric benchmark, $s(r_1)=-r_1$. Increasing $k_2$ moves the swing function upward where sender~2's inverse reach binds and leaves it locally unchanged where that reach is non-binding. When sender~2 is more biased and faces cheaper misreporting costs, the change is nonuniform. Indeed, the swing function lies above the benchmark near the origin, but below it for stronger reports by sender~1. Thus, changes in sender~2's primitives need not shift the receiver's burden of proof uniformly across report pairs.

The comparative statics are local and may be piecewise. They are nevertheless useful for welfare analysis when evaluated at the report pair that governs the comparison with selected monopoly. We turn to that application next.
	

	\subsubsection{Welfare implications of the LQ comparative statics}\label{sec:more_bias_cost}
	
	The symmetric benchmark in \Cref{thm:symmetric} is a primitive environment in which competition strictly improves upon selected monopoly. To connect the LQ comparative statics to that result, evaluate the swing function at sender~1's monopoly pooling report $r^M$. The relevant object is
\[
s_{\tau_2,k_2}(r^M).
\]
In the symmetric benchmark, $s_{-\tau_1,k_1}(r^M)=\ell^M$. As a result, any primitive change that preserves $s_{\tau_2,k_2}(r^M)\geq\ell^M$ also keeps the competitive left cutoff weakly to the right of the monopoly pooling cutoff. The lower-bound argument in \Cref{prop:bounds} yields
\[
W_r^{\DE}(\tau_2,k_2)>W_r^{\M}(0).
\]

The cost comparative static provides an immediate application. Holding sender~2's bias at its symmetric value, a higher $k_2$ weakly raises the comparison-relevant swing report. The welfare ranking from the symmetric benchmark continues to hold.
	
	\begin{corollary}\label{cor:more-costly}
		Suppose the LQ specification holds. If $\widehat\tau_2=-\tau_1$ and $\widehat k_2\geq k_1$, then $W_r^{\DE}(-\tau_1,\widehat k_2)>W_r^{\M}(0)$.
	\end{corollary}


	The corollary extends the symmetric welfare ranking to every $\widehat k_2\geq k_1$ while maintaining symmetric biases. Extending the ranking through changes in sender~2's bias is less immediate. \Cref{prop:bias-effect} applies locally and requires sender~2's inverse reach to be non-binding. When that condition holds along the comparison, a higher $\tau_2$ cannot lower the swing report evaluated at $r^M$. We now study when the required non-binding condition is satisfied.

	
	\subsubsection{The non-binding inverse-reach condition}\label{subsec:interpretation_condition}
	
	For a report $r$ by sender~1, the \emph{non-binding inverse-reach} condition is
	\begin{equation}\label{eq:nbir_inverse}
		\underline r_2^{-1}\left(s(r)\right)\geq r.
	\end{equation}
	Since $\underline r_2(\cdot)$ is strictly increasing, this is equivalent to $s(r)\geq \underline r_2(r)$. The interpretation is simple. At the report pair $(r,s(r))$, sender~2 is not constrained by the boundary of his feasible misreporting set. The counter-report required to overturn sender~1's report is still affordable for sender~2.
	
	Condition~\eqref{eq:nbir_inverse} is not automatic. We separate two questions. First, holding the report pair fixed, how do sender~2's primitives affect his reach? Second, does the condition continue to hold once the swing report is allowed to adjust endogenously? The next lemma addresses the first question, and the proposition that follows addresses the second in an extreme LQ limit.
	
	\begin{lemma}\label{prop:condition_lq}
		Consider the LQ class. For every fixed $\theta<\tau_2$,
		\[
		\frac{\partial \underline r_2(\theta)}{\partial \tau_2}<0 < \frac{\partial \underline r_2(\theta)}{\partial k_2}.
		\]
		likewise, for every fixed report $r_2<\tau_2$ in the domain of $\underline r_2^{-1}$,
		\[
		\frac{\partial \underline r_2^{-1}(r_2)}{\partial \tau_2}>0 > \frac{\partial \underline r_2^{-1}(r_2)}{\partial k_2}.
		\]
		Hence, for any fixed report pair $(r,s(r))$, condition~\eqref{eq:nbir_inverse} becomes weakly easier to satisfy when sender~2 becomes more biased, meaning higher $\tau_2$, or cheaper to lie, meaning lower $k_2$.
	\end{lemma}


	Holding the report pair fixed, a higher bias or a lower misreporting cost expands sender~2's reach and makes the non-binding condition easier to satisfy. This observation is not sufficient by itself because the swing report also changes with the primitives. The next proposition accounts for that endogenous response. It shows that the non-binding condition eventually holds in the LQ class when sender~2 becomes arbitrarily biased and arbitrarily cheap to misreport.
	
	\begin{proposition}
		\label{prop:condition_inverse}
		Consider the LQ class. Fix any sequence of competitive environments $(\tau_2^n,k_2^n)_{n\geq1}$	such that $\tau_2^n\to+\infty$ and $k_2^n\to0^+$. For each $n$, let
		\[
		s_n\coloneqq s(r^M\mid\tau_2^n,k_2^n)
		\]
		be the comparison-relevant swing report. Then, for all sufficiently large $n$, we obtain
		\[
		\underline r_2^{-1}\left(s_n\mid\tau_2^n,k_2^n\right)\geq r^M.
		\]
	\end{proposition}

	
\Cref{prop:condition_inverse} establishes that the non-binding condition eventually holds along any sequence for which sender~2's bias diverges and his misreporting cost converges to zero. As sender~2 becomes more extreme and less constrained by misreporting costs, the equilibrium mistake probabilities may also change.

The LQ analysis clarifies how sender~2's primitives affect the receiver's burden of proof, but it does not imply that receiver welfare is globally monotone in either bias or cost. The paper's main welfare case for competition remains the high-$k_2$ result, under which competitive mistakes vanish.
	


	\section{Total welfare}
	\label{sec:total-welfare}

	The analysis so far has evaluated market structures from the receiver's perspective. We now consider total welfare, defined as the sum of the receiver's payoff and both senders' action payoffs, net of communication costs. We hold the set of players fixed across market structures. Sender~2's action payoff is counted under both monopoly and competition. Under monopoly, however, sender~2 is muted and incurs no communication cost. Under competition, both senders communicate and may incur such costs.

For a given market configuration, let $q(\theta)$ denote the probability that the receiver chooses $a^+$ in state $\theta$, and let $\mathcal C_j(\theta)$ denote sender~$j$'s expected communication cost in that state, including the cost parameter $k_j$. The comparison has two components: $(i)$ the action rule induced by each market structure, and $(ii)$ the communication costs required to sustain it.
	
	In the monopolistic environment, sender~1 is active and sender~2 is muted. Hence, total welfare in the least informative monopoly equilibrium is
	\[
	\mathcal W_{\mathrm{tot}}^{\M}(0) = \int_\Theta \left[u_r(\theta)+u_1(\theta)+u_2(\theta)\right] q^{\M}_0(\theta)f(\theta)\,d\theta - \int_\Theta \mathcal C_1^{\M,0}(\theta)f(\theta)\,d\theta,
	\]
	where $q^{\M}_0$ is the probability of action $a^+$ in the least informative monopoly equilibrium and $\mathcal C_1^{\M,0}$ is sender~1's expected communication cost in that equilibrium.
	
	In the competitive environment, both senders are active. Total welfare in the adversarial equilibrium is
	\[
	\mathcal W_{\mathrm{tot}}^{\DE}(\tau_2,k_2) = \int_\Theta \left[u_r(\theta)+u_1(\theta)+u_2(\theta)\right] q^{\DE}(\theta\mid \tau_2,k_2)f(\theta)\,d\theta - \sum_{j=1}^2 \int_\Theta \mathcal C_j^{\DE}(\theta\mid \tau_2,k_2)f(\theta)\,d\theta.
	\]
	Because the same three players are counted in both environments, any welfare difference comes from the induced action rule and communication costs, not from mechanically adding sender~2's payoff under competition.
	
	The receiver-full-information total-welfare benchmark is
	\[
	\mathcal W_{\mathrm{tot}}^{\FI} = \int_{\{\theta \,\mid\,u_r(\theta)\geq0\}} \left[u_r(\theta)+u_1(\theta)+u_2(\theta)\right]f(\theta)\,d\theta.
	\]
	This is the total welfare generated by the receiver's full-information decision rule. It is not necessarily the maximum of total welfare over all decision rules, because the receiver's cutoff need not coincide with the cutoff that maximizes the sum of players' action payoffs.
	
	The distinction between receiver welfare and total welfare is captured by the aggregate action-payoff difference
\[
S(\theta) \coloneqq u_r(\theta)+u_1(\theta)+u_2(\theta).
\]
The receiver's full-information decision is governed by the sign of $u_r(\theta)$, whereas the action that maximizes total payoffs is governed by the sign of $S(\theta)$. Hence, if $S(\theta)>0$ in some states in which $u_r(\theta)<0$, persuasion toward $a^+$ can harm the receiver while improving the allocation of actions.

Misreporting costs do not create this wedge. As the cheap-talk benchmark illustrates, the same conflict can arise even when communication is costless. In the present model, misreporting costs instead determine whether a distortion can be sustained and whether its action-payoff benefit survives after the communication costs of the active senders are taken into account.

These forces allow either market structure to maximize total welfare. When $S(\theta)$ is positive on some negative states, the distortion induced by monopoly may raise total welfare despite reducing receiver welfare. When $S(\theta)$ is negative throughout the negative states distorted by monopoly, that distortion is socially harmful, and sufficiently disciplined competition can improve both receiver welfare and total welfare. We begin by establishing the limiting level of total welfare under competition.
	
	\begin{lemma}
		\label{lem:total-welfare-high-k2}
		Consider the LQ class with fixed sender~1 primitives, fixed finite $\tau_2>0$, and fixed $k_1>0$. Suppose that $f$ has finite first moment and that the high-cost convergence conclusion of \Cref{lem:high-k2-convergence} holds. Then, $\mathcal W_{\mathrm{tot}}^{\DE}(\tau_2,k_2) \to \mathcal W_{\mathrm{tot}}^{\FI}$ as $k_2\to\infty$.
	\end{lemma}


	The lemma reduces the large-$k_2$ comparison to the ranking between the selected monopoly outcome and the receiver's full-information rule. If the selected monopoly outcome generates more total welfare than that rule, it must eventually dominate competition as $k_2$ increases. The next proposition provides conditions under which this ranking holds in the LQ class.
	
	\begin{proposition}
		\label{thm:total-welfare-monopoly-dominates}
		Consider the LQ class of \Cref{def:LQ}, with $\tau_1<0<\tau_2$. Suppose that the density $f$ is continuous at $0$, satisfies $f(0)>0$, has full support, and has finite first moment. Suppose also that the least informative monopoly equilibrium is well defined for the relevant values of $k_1$. If $-\tau_1>3\tau_2$, then there exists $\bar k_1<\infty$ such that $\mathcal W_{\mathrm{tot}}^{\M}(0)> \mathcal W_{\mathrm{tot}}^{\FI}$ for every $k_1\geq\bar k_1$.
	\end{proposition}


The proposition shows that monopoly's distortion can be socially valuable. When sender~1's preference for $a^+$ is sufficiently strong, the resulting gain in action payoffs more than compensates for sender~1's communication cost. Selected monopoly then yields strictly higher total welfare than the receiver's full-information rule.

This comparison also has implications for competition. By \Cref{lem:total-welfare-high-k2}, competitive total welfare converges to the full-information benchmark as $k_2$ increases. The strict welfare advantage identified in the proposition must persist relative to competition for all sufficiently large but finite values of $k_2$. The following corollary records this implication.

	
	\begin{corollary}
    \label{cor:finite-k2-total-welfare-dominance}
    Under the assumptions of \Cref{thm:total-welfare-monopoly-dominates}, for every $k_1\geq\bar k_1$, there exists a finite threshold $\bar k_2(k_1,\tau_2)$ such that $\mathcal W_{\mathrm{tot}}^{\M}(0)> \mathcal W_{\mathrm{tot}}^{\DE}(\tau_2,k_2)$ for every $k_2\geq\bar k_2(k_1,\tau_2)$.
\end{corollary}
	
\Cref{thm:total-welfare-monopoly-dominates} and \Cref{cor:finite-k2-total-welfare-dominance} establish that greater decision accuracy need not increase total welfare. Under the stated conditions, selected monopoly induces $a^+$ on negative states in which the resulting gain in total action payoffs is large enough to offset sender~1's communication cost. Monopoly lowers receiver welfare but raises total welfare above the receiver's full-information benchmark. Because competition converges to that benchmark as $k_2$ grows, monopoly dominates competition for all sufficiently large finite values of $k_2$.

The opposite ranking arises when the receiver and total welfare agree on the states distorted by monopoly. In the LQ class, $S(\theta)=3\theta-\tau_1-\tau_2$. If $\tau_2\geq-\tau_1$, then $S(\theta)<0$ for every negative state. Monopoly's distortion toward $a^+$ is socially harmful throughout its pooling interval. The next proposition combines this alignment with the high-$k_2$ convergence result.
	
	\begin{proposition}
		\label{thm:total-welfare-competition-dominates}
	Consider the LQ class of \Cref{def:LQ}, and fix $k_1>0$. Suppose that the density $f$ has full support and finite first moment, and suppose that the least informative monopoly equilibrium is well defined and non-degenerate. If $\tau_2\geq -\tau_1$, then $\mathcal W_{\mathrm{tot}}^{\FI}	>\mathcal W_{\mathrm{tot}}^{\M}(0)$. As a result, there exists a finite threshold $\bar k_2(k_1,\tau_2)$ such that $\mathcal W_{\mathrm{tot}}^{\DE}(\tau_2,k_2)>	\mathcal W_{\mathrm{tot}}^{\M}(0)$ for every $k_2\geq \bar k_2(k_1,\tau_2)$.
 \end{proposition}

	
	The two propositions provide sufficient conditions for opposite welfare rankings. When sender~1's preference for $a^+$ is sufficiently strong, monopoly can generate socially valuable persuasion and outperform competition. When sender~2's opposition is sufficiently strong, the receiver and total welfare agree on the negative states distorted by monopoly, and sufficiently disciplined competition improves both welfare criteria. These conditions do not exhaust the parameter space. Between them, the ranking also depends on the magnitude of the action-payoff gains and the equilibrium communication costs.

Therefore, the welfare criterion matters for policy. A regulator concerned with receiver welfare has a direct reason to favor competition when $k_2$ is high because competition improves decision accuracy. A regulator concerned with total welfare must additionally ask whether the receiver's preferred action maximizes aggregate surplus. When sender payoffs represent genuine social stakes rather than transfers or purely distributive concerns, eliminating persuasion can reduce total welfare even as it improves the receiver's decisions.
	
		
\subsection{A numerical illustration}
\label{ex:total-welfare-illustration}

The following example illustrates how selected monopoly can raise total welfare above the receiver's full-information benchmark. Let $\theta$ be standard normal, with density $\varphi$, and set $u_r(\theta)=\theta$ and
\[
u_1(\theta)=\theta+\frac{4}{5} > u_2(\theta)=\theta-\frac{1}{10}.
\]
Misreporting costs are quadratic, with $k_1=35/2$. The least informative monopoly equilibrium has
\[
\ell^M=-\frac1{10} < r^M=\frac1{10}.
\]
Indeed, symmetry of the prior gives
\[
\int_{-1/10}^{1/10}\theta\varphi(\theta)\,d\theta=0,
\]
and the lower pooling type reaches $r^M$ because
\[
\frac{35}{2} \left(\frac15\right)^2 = \frac7{10} = u_1\!\left(-\frac1{10}\right).
\]

The aggregate action-payoff gain from choosing $a^+$ is
\[
S(\theta) =u_r(\theta)+u_1(\theta)+u_2(\theta) =3\theta+\frac7{10}.
\]
Relative to the receiver's full-information rule, monopoly changes the action only on $[-1/10,0)$. The resulting gain in action payoffs is
\[
\int_{-1/10}^{0} \left(3\theta+\frac7{10}\right)\varphi(\theta)\,d\theta \approx 0.02191.
\]
Sender~$1$'s equilibrium misreporting cost is
\[
\frac{35}{2} \int_{-1/10}^{1/10} \left(\frac1{10}-\theta\right)^2\varphi(\theta)\,d\theta \approx 0.01858.
\]
Thus,
\[
\mathcal W_{\mathrm{tot}}^{\M}(0) -\mathcal W_{\mathrm{tot}}^{\FI} \approx 0.02191-0.01858 =0.00333>0.
\]
The level of total welfare under the receiver's full-information rule is
\[
\mathcal W_{\mathrm{tot}}^{\FI} = \int_0^\infty \left(3\theta+\frac7{10}\right)\varphi(\theta)\,d\theta = \frac{3}{\sqrt{2\pi}}+\frac7{20} \approx1.54683.
\]
Hence, $\mathcal W_{\mathrm{tot}}^{\M}(0) \approx1.550157$. By \Cref{lem:total-welfare-high-k2}, competitive welfare converges to $\mathcal W_{\mathrm{tot}}^{\FI}$ as $k_2\to\infty$. Hence, for all sufficiently large finite $k_2$, selected monopoly yields higher total welfare than under the selected competitive outcome in this example.\footnote{Measured relative to monopoly welfare, the limiting percentage loss from monopoly to competition is $\approx0.215\%$. This is a limiting percentage, as the exact loss at a given finite $k_2$ requires solving the corresponding adversarial equilibrium. The calculation includes the equilibrium misreporting cost and illustrates the same mechanism as the formal proposition.}

\begin{figure}
    \centering
    \includegraphics[width=0.8\linewidth]{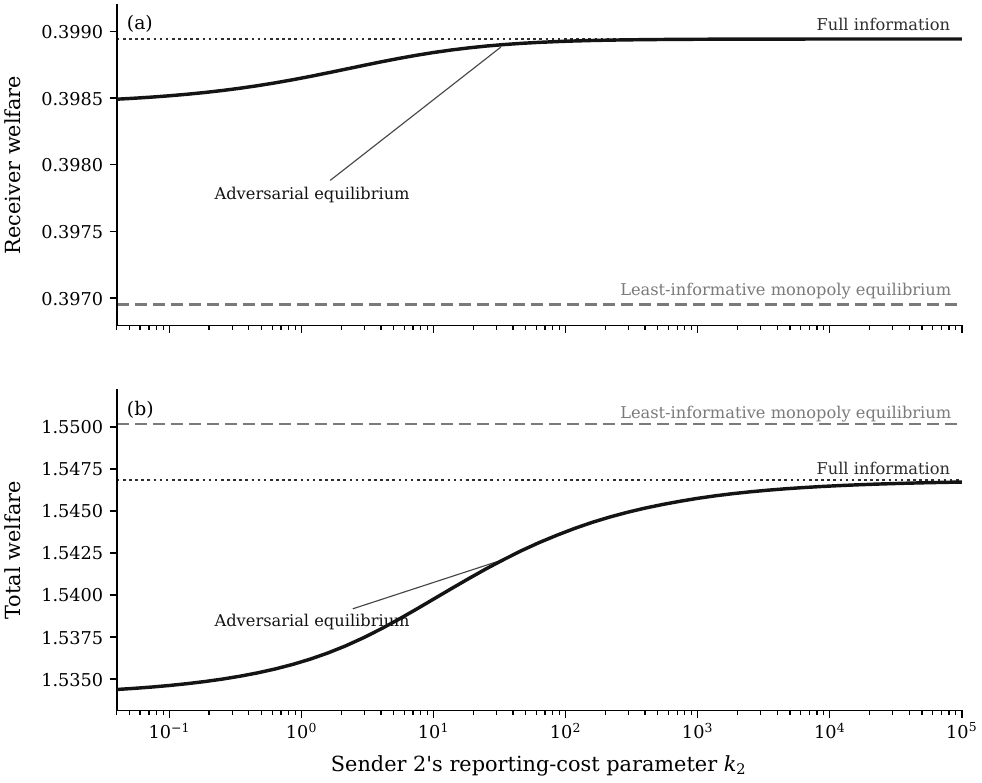}
    \caption{Receiver and total welfare as sender~2's reporting-cost parameter varies. The parameters are those of the numerical illustration, i.e.,  $\theta\sim\mathcal N(0,1)$, $(\tau_1,k_1)=(-4/5,35/2)$, and $\tau_2=1/10$. The horizontal axis is logarithmic. Panel~(a) reports receiver welfare, while panel~(b) reports total welfare, including the action payoffs of the receiver and both senders and the equilibrium reporting costs of both active senders. In each panel, the solid curve represents the adversarial equilibrium, the dotted line represents full information, and the dashed line represents the least-informative monopoly equilibrium.}
    \label{fig:welfare-by-k2}
\end{figure}

Figure~\ref{fig:welfare-by-k2} illustrates the divergence between the two welfare criteria. Panel~(a) shows that competition improves receiver welfare throughout the displayed range and converges to the receiver's full-information payoff as $k_2$ increases. Panel~(b) shows the opposite ranking for total welfare. Competition remains below selected monopoly and converges to the lower full-information benchmark. In this example, monopoly's distortion raises aggregate action payoffs by enough to offset sender~1's communication cost. Competition improves decision accuracy while reducing total welfare.


\section{Concluding remarks and discussion}

		
	\subsection{Competition, bias, and institutional design}
	\label{sec:policy-discussion}
	
	The results yield several insights for institutional design. Competition can discipline a biased source, but plurality alone does not guarantee better information. Under the selected outcomes studied in the paper, adding an oppositely biased source improves the receiver's welfare when that source faces sufficiently high costs of misreporting. The broader welfare effect remains ambiguous because the receiver's preferred action need not maximize aggregate surplus. The value of competition depends jointly on the sources' incentives, the institutions that constrain manipulation, and the welfare criterion used to evaluate the outcome.

Media regulation provides the most direct application. Existing regulatory frameworks combine concerns about plurality with constraints on reporting. Ofcom requires news to be reported with due accuracy and due impartiality, while clarifying that due impartiality does not require equal time for every view \citep{Ofcom2025}. In the United States, the FCC describes competition, localism, and viewpoint diversity as traditional public-interest goals of its review of broadcast-ownership rules \citep{FCC2025}. The European Media Freedom Act similarly protects access to a plurality of editorially independent media content and requires assessments of how media-market concentrations affect pluralism and editorial independence \citep{EMFA2024}.

The model suggests that plurality and accountability address complementary aspects. Plurality creates the possibility that one source checks another, whereas verification, liability, reputational concerns, and related institutions raise the cost of manipulating information. Competition is most effective when the additional source can discipline one-sided persuasion without cheaply generating new strategic noise. Opposition between sources can help produce this discipline, but greater bias is not valuable in itself. Indeed, our LQ analysis shows that bias changes the receiver's burden of proof without implying a globally positive welfare effect.

Adversarial legal procedure provides a second application. Its usual rationale is that opposed advocates discipline one another before an impartial decision maker, as in models of adversarial evidence and disclosure \citep{Shin1998}. At the same time, adversarial proceedings consume resources and encourage strategic presentation of evidence \citep{Tullock1980,Posner1999,Kagan2003}. The model captures this trade-off. Opposed advocates can improve the decision maker's information when misrepresentation is sufficiently constrained, but the informational benefit is not automatic. Nor does greater accuracy from the decision maker's perspective necessarily maximize total welfare when the interests of other affected parties are included. The analysis identifies conditions under which the disciplinary benefit of opposition is strongest.

Second-opinion and escalation mechanisms provide a related application. Martha's Rule in the English NHS allows patients, families, carers, and staff to request a rapid review when concerns about a patient's deterioration are not being addressed. The review is conducted by a different team, creating an additional channel through which an initial assessment can be challenged \citep{NHSEngland2026}. NHS administrative data report that activations of the rule have been followed by changes in treatment and transfers to higher levels of care. This experience illustrates the institutional role of an independent check, but it should not be interpreted as a direct test of the model. Medical experts may possess different information and need not have opposed preferences.\footnote{Additional advice is most valuable when the new source is sufficiently independent to challenge the initial recommendation, sufficiently accountable that strategic misrepresentation is costly, and sufficiently well integrated into the decision process that the additional review does not create excessive delay or confusion. The model does not imply that more opinions, or more strongly opposed opinions, are always desirable.}

A key implication of our findings is that competition, bias, and accountability should be evaluated jointly. Opposed incentives can discipline persuasion, but the welfare effect of greater bias remains ambiguous. Competition improves receiver welfare under the paper's high-cost condition, but it need not improve total welfare. Institutional design should consider how costly manipulation is, how do the sources' incentives interact, and whose welfare the institution seeks to maximize.

	\section{Conclusion}
	
	This paper studies whether competition improves information provision when biased sources can misrepresent facts at a cost. The comparison is complicated by equilibrium multiplicity. Both monopoly and competition can support outcomes with very different informational properties. We address this problem using a common belief-based selection criterion. The criterion selects the least informative equilibrium from the monopoly family and the adversarial-equilibrium outcome under two-sender competition. The resulting welfare comparison applies the same criterion to beliefs in both market structures.

Competition does not uniformly dominate monopoly. A sufficiently skeptical receiver can obtain nearly full information from a monopolist and may be better off than under adversarial competition. Under the selected outcomes, however, competition improves receiver welfare in the symmetric benchmark and, more generally, when the additional sender faces sufficiently high misreporting costs. For any fixed finite bias of sender~2, a higher misreporting cost eventually eliminates the mistakes generated by that sender while preserving his ability to discipline sender~1. Competitive welfare approaches the receiver's full-information benchmark, whereas selected monopoly retains a positive pooling loss.

Greater decision accuracy need not imply greater total welfare. The receiver's preferred action may differ from the action that maximizes the sum of all players' payoffs. Monopoly can induce persuasion that harms the receiver but improves the allocation of actions from an aggregate perspective. When this benefit exceeds the associated communication cost, selected monopoly can generate more total welfare than competition. When the receiver's ranking and the aggregate-surplus ranking are aligned on the states distorted by monopoly, sufficiently disciplined competition improves both welfare criteria.

The analysis deliberately focuses on a binary decision and two sources that observe the same state. The binary action space isolates the threshold decision at the center of the model, while costly misreporting captures an intermediate case between cheap talk and verifiable disclosure. The many-sender extension shows that truthful full revelation can be supported when at least three senders report and a unilateral outlier cannot overturn the agreement of the others.\footnote{This is an existence result under unilateral consistency, not a selection or uniqueness result. Richer action spaces, heterogeneous information, endogenous source entry, and coordination among aligned senders provide natural directions for further research.}

Competition is most valuable when an additional source can challenge one-sided persuasion but faces effective constraints on manipulation. Plurality and accountability are complementary, as adding voices can improve decisions, but its value depends on the sources' incentives, the credibility of their reports, and the welfare criterion used to evaluate the outcome.

	
	\newpage
	
	\appendix

	
	\section{Proofs}


\begin{proof}[Proof of \Cref{lem:first-branch-automatic}]
		Take $\theta<\theta'$ in $\left[\bar r_1^{-1}(r),\,\bar r_1^{-1}(r^*(\lambda))\right]$. On this interval the report $r$ lies weakly above both types, so	$|r-\theta|>|r-\theta'|$. By the assumed distance-monotonicity of the cost function, we have $C_1(r,\theta)>C_1(r,\theta')$. Since $u_1(\theta)$ is strictly increasing, it follows that $u_1(\theta)<u_1(\theta')$. Hence,
		\[
		\frac{k_1C_1(r,\theta)}{u_1(\theta)} > \frac{k_1C_1(r,\theta')}{u_1(\theta')}.
		\]
		Since the relevant interval lies above sender 1's threshold $\tau_1$, we have $u_1(\theta)>0$ throughout it. Therefore, $q_r(\theta)>q_r(\theta')$, so the first segment is strictly decreasing.
	\end{proof}


\begin{proof}[Proof of \Cref{prop:Gamma-costs}]
		By \Cref{lem:first-branch-automatic}, the first piece of $q_r$ is strictly	decreasing.	It remains to study the second piece,
		\[
		q_r(\theta) = 1-\frac{k_1\left(C_1(r^*(\lambda),\theta)-C_1(r,\theta)\right)}{u_1(\theta)},
		\]
		for $\theta\in\left[\bar r_1^{-1}(r^*(\lambda)),\,\tilde\theta(r,\lambda)\right]$ and where $\tilde\theta(r,\lambda)$ is defined by $C_1(r,\tilde\theta)=C_1(r^*(\lambda),\tilde\theta)$. Set
		\[
		A(\theta)\coloneqq C_1(r^*(\lambda),\theta)-C_1(r,\theta).
		\]
		Since $r<r^*(\lambda)$ and $\theta\leq \tilde\theta(r,\lambda)$, we have $A(\theta)\geq 0$. We claim that $A(\theta)$ is weakly decreasing on the whole interval.
		
		First consider $\theta\leq r$. Then, $A(\theta)=\Gamma(r^*(\lambda)-\theta)-\Gamma(r-\theta)$. Because $\Gamma$ is weakly convex, its derivative is weakly increasing wherever it exists. Hence,
		\[
		\frac{dA(\theta)}{d\theta} = -\Gamma'(r^*(\lambda)-\theta)+\Gamma'(r-\theta) \leq 0,
		\]
		since $r^*(\lambda)-\theta>r-\theta$.
		
		Now consider $r\leq \theta\leq \tilde\theta(r,\lambda)$. Then, $A(\theta)=\Gamma(r^*(\lambda)-\theta)-\Gamma(\theta-r)$. Therefore,
		\[
		\frac{dA(\theta)}{d\theta} = -\Gamma'(r^*(\lambda)-\theta)-\Gamma'(\theta-r)<0,
		\]
		because $\Gamma$ is strictly increasing, so $\Gamma'>0$ wherever it exists.	Thus, $A(\theta)$ is weakly decreasing on the whole second segment.
		
		Since $u_1(\theta)$ is strictly increasing and positive on the same interval, the ratio	$A(\theta)/u_1(\theta)$ is strictly decreasing. As a result,
		\[
		q_r(\theta)=1-\frac{k_1A(\theta)}{u_1(\theta)}
		\]
		is strictly increasing. This proves \Cref{ass:sc}.
	\end{proof}


\begin{proof}[Proof of \Cref{thm:nwbr}]
		Fix a generic monopoly equilibrium indexed by $\lambda$. For readability write
		\[
		r^*\coloneqq r^*(\lambda),
		\]
 \[
		\ell\coloneqq \ell(\lambda),
		\]
 \[
		\underline\theta(r)\coloneqq \bar r_1^{-1}(r).
		\]
		Fix an arbitrary off-path report $r\in(\ell,r^*)$. We show that one round of the Never-a-Weak-Best-Response test deletes every type except $\ell$ and that, once the receiver assigns the deviation to type $\ell$, the deviation is not profitable. This proves that the equilibrium survives the refinement.
		
		The sender's equilibrium payoff is
		\[
		W_\lambda(\theta)=
		\begin{cases}
			0, & \theta<\ell,\\
			u_1(\theta)-k_1 C_1(r^*,\theta), & \theta\in[\ell,r^*),\\
			u_1(\theta), & \theta\geq r^*.
		\end{cases}
		\]
		
		First, we observe that only types in $[\underline\theta(r),\vartheta(r,\lambda)]$ can ever weakly gain from deviating to $r$. If $\theta<\underline\theta(r)$, then even if the receiver chooses $a^+$ with probability one,
		\[
		w(r,1,\theta)=u_1(\theta)-k_1 C_1(r,\theta)<0=W_\lambda(\theta),
		\]
		by definition of the inverse reach $\underline\theta(r)=\bar r_1^{-1}(r)$. Hence, no such type can weakly prefer the deviation.
		
		If $\theta\in(\vartheta(r,\lambda),r^*)$, then $C_1(r,\theta)>C_1(r^*,\theta)$, so even when action $a^+$ is chosen with probability one,
		\[
		w(r,1,\theta)=u_1(\theta)-k_1 C_1(r,\theta) < u_1(\theta)-k_1 C_1(r^*,\theta)=W_\lambda(\theta).
		\]
		Thus, no such type can weakly prefer the deviation either.
		
		Finally, if $\theta\geq r^*$, then
		\[
		w(r,1,\theta)=u_1(\theta)-k_1 C_1(r,\theta)<u_1(\theta)=W_\lambda(\theta),
		\]
		again because $C_1(r,\theta)>0$. Therefore, only types in the interval $[\underline\theta(r),\vartheta(r,\lambda)]$ can ever be weak best responses to the off-path report $r$.
		
		Next, we show that the receiver probability that makes a type indifferent is uniquely minimized at $\ell$. For $\theta\in[\underline\theta(r),\ell]$, the equilibrium payoff is $W_\lambda(\theta)=0$. The indifference condition between deviating to $r$ and staying in equilibrium is 
		\[
		0=\sigma u_1(\theta)-k_1 C_1(r,\theta),
		\]
		which yields
		\[
		\sigma=q_r(\theta)=\frac{k_1 C_1(r,\theta)}{u_1(\theta)}.
		\]
		Hence, for such a type, we have $D_0(\theta,r)=\{q_r(\theta)\}$ and $D(\theta,r)=(q_r(\theta),1]$.
		
		For $\theta\in[\ell,\vartheta(r,\lambda)]$, the equilibrium payoff is $W_\lambda(\theta)=u_1(\theta)-k_1 C_1(r^*,\theta)$. The indifference condition becomes
		\[
		u_1(\theta)-k_1 C_1(r^*,\theta)=\sigma u_1(\theta)-k_1 C_1(r,\theta),
		\]
		so
		\[
		\sigma=q_r(\theta)=1-\frac{k_1\left(C_1(r^*,\theta)-C_1(r,\theta)\right)}{u_1(\theta)}.
		\]
		Hence, again, we get $D_0(\theta,r)=\{q_r(\theta)\}$ and $D(\theta,r)=(q_r(\theta),1]$. \Cref{ass:sc} implies that the function $q_r$ is strictly decreasing on $[\underline\theta(r),\ell]$ and strictly increasing on $[\ell,\vartheta(r,\lambda)]$. Therefore, $q_r$ has a unique minimum at $\theta=\ell$.
		
		Finally, we obtain that one round of $\NWBR$ deletes every type except $\ell$.	Fix any type $\theta\neq \ell$ in $[\underline\theta(r),\vartheta(r,\lambda)]$. Since $q_r(\theta)>q_r(\ell)$ and $D(\ell,r)=(q_r(\ell),1]$, we have $D_0(\theta,r)=\{q_r(\theta)\}\subseteq D(\ell,r)$. Hence, the pair $(\theta,r)$ is deleted by the Never-a-Weak-Best-Response test. This argument applies to every type $\theta\neq \ell$ that survived a deviation to $r$. Thus, every pair $(\theta,r)$ with $\theta\neq \ell$ is eliminated in one round. By contrast, $D_0(\ell,r)=\{q_r(\ell)\}$ is not contained in $\bigcup_{\theta\neq \ell} D(\theta,r)$, because every interval $D(\theta,r)=(q_r(\theta),1]$ begins strictly above $q_r(\ell)$.	So $(\ell,r)$ is the unique surviving type-report pair.
		
		Once all other types have been deleted, the only type consistent with the off-path report $r$ is $\ell$. Because $\ell<0$, the receiver's unique best reply to beliefs concentrated on type $\ell$ is action $a^-$. Equivalently, the receiver chooses $a^+$ with probability $\sigma=0$. Under that reply,
		\[
		w(r,0,\ell)=-k_1 C_1(r,\ell)<0.
		\]
		Since type $\ell$ obtains equilibrium payoff $W_\lambda(\ell)=0$, the deviation is not profitable. Because the report $r\in(\ell,r^*)$ was arbitrary, every off-path report fails to be profitable after the Never-a-Weak-Best-Response refinement is applied. Therefore, the equilibrium survives the refinement.
	\end{proof}


\begin{proof}[Proof of \Cref{cor:lq-nwbr}]
		Fix $r\in(\ell(\lambda),r^*(\lambda))$. For $\theta\in[\bar r_1^{-1}(r),\ell(\lambda)]$ we have
		\[
		q_r(\theta)=\frac{k_1(r-\theta)^2}{\theta-\tau_1}.
		\]
		Differentiating yields
		\[
		q_r'(\theta) = -k_1\frac{(r-\theta)(r+\theta-2\tau_1)}{(\theta-\tau_1)^2}<0,
		\]
		because $r>\theta>\tau_1$. For $\theta\in[\ell(\lambda),\vartheta(r,\lambda)]$ we have
		\[
		q_r(\theta)=1-k_1\frac{(r^*(\lambda)-\theta)^2-(r-\theta)^2}{\theta-\tau_1} =1-k_1\left(r^*(\lambda)-r\right)\frac{r^*(\lambda)+r-2\theta}{\theta-\tau_1}.
		\]
		Differentiating gives
		\[
		q_r'(\theta) = k_1\left(r^*(\lambda)-r\right)\frac{r^*(\lambda)+r-2\tau_1}{(\theta-\tau_1)^2}>0,
		\]
		because $r^*(\lambda)>r>\tau_1$. Hence, \Cref{ass:sc} holds. The conclusion follows from \Cref{thm:nwbr}.
	\end{proof}


\begin{proof}[Proof of \Cref{thm:common-refinement-monopoly}]
Write $r^M\coloneqq r^*(0)$ and $\ell^M\coloneqq\bar r_1^{-1}(r^M)$. The proof has two parts. We first construct admissible off-path beliefs for $E_0$. For a report below the pooling report, we will let the receiver infer a single state $t^M(r)$. This inferred state rises continuously from $\ell^M$ to zero as the report rises from $\ell^M$ to $r^M$. Formally, for $\ell^M\leq r<r^M$, let
\[
t^M(r) \coloneqq \ell^M\frac{r^M-r}{r^M-\ell^M}.
\]
For $0\leq r<r^M$, monotonicity of the inverse reach and the fact that $0<-\ell^M<r^M-\ell^M$ give
\[
\bar r_1^{-1}(r) \leq \ell^M \leq t^M(r) \leq r.
\]
Hence, $t^M(r)\in P_1(r)=[\bar r_1^{-1}(r),r]$. The inferred state can indeed rationally generate report $r$.

Retain the Bayesian posterior at every on-path report and assign the following posterior off path, i.e., $\mu_r=\delta_{t^M(r)}$ for $\ell^M<r<r^M$. Together with the on-path beliefs, this specification generates the following expected payoff difference for the receiver,
\begin{equation*}
\widetilde V(r)
=
\begin{cases}
u_r(r), & r\leq\ell^M,\\[3pt]
u_r\!\left(t^M(r)\right), & \ell^M<r<r^M,\\[3pt]
0, & r=r^M,\\[3pt]
u_r(r), & r>r^M.
\end{cases}
\end{equation*}

The construction has the desired economic implication. Below $r^M$, the inferred state is below zero, so the receiver chooses $a^-$. At $r^M$, she is indifferent and chooses $a^+$ under the maintained tie-breaking rule. Above $r^M$, reports are truthful and she again chooses $a^+$. Moreover, $\widetilde V$ is weakly increasing everywhere and strictly increasing within each atom-free component. The component below the pooling atom is
\[
C_-\coloneqq\left[0,r^M\right).
\]
Extend $t^M$ to the pooling report by setting $t^M(r^M)=0$, and define for $r\in[0,r^M]$
\[
\overline V_{C_-}(r)
\coloneqq
u_r\!\left(t^M(r)\right).
\]
This function agrees with $\widetilde V$ throughout $C_-$ and is continuous at
$r^M$. Its component-wise derivative satisfies
\[
\partial_1^{C_-}\overline V_{C_-}(r)
=
u_r'\!\left(t^M(r)\right)
\frac{-\ell^M}{r^M-\ell^M}
>0
\]
for every $r\in[0,r^M]$, with the appropriate one-sided derivative at the
endpoints. In particular,
\[
\partial_1^{C_-}\overline V_{C_-}(r^M)
=
u_r'(0)\frac{-\ell^M}{r^M-\ell^M}
>0.
\]

If $r^M<\bar r_1(0)$, there is also an upper component defined by
\[
C_+\coloneqq \left(r^M,\bar r_1(0)\right).
\]
Define its extension for all $r\in[r^M,\bar r_1(0))$ by
\[
\overline V_{C_+}(r)
\coloneqq
u_r(r).
\]
This extension agrees with $\widetilde V$ throughout $C_+$ and satisfies $\partial_1^{C_+}\overline V_{C_+}(r)
=
u_r'(r)>0$, including at $r^M$ in the feasible one-sided sense.

The lower extension assigns value $\overline V_{C_-}(r^M)=u_r(0)=0$. When the upper component is present, its extension instead assigns $\overline V_{C_+}(r^M)=u_r(r^M)>0$. Therefore, the skeptical-envelope map selects the lower-component value, i.e.,
\[
\left(\mathcal L_{\nu^{E_0}}\widetilde V\right)(r^M)
=
0
=
\widetilde V(r^M).
\]
Thus, requirement~$(iii)$ holds at the pooling atom and throughout both atom-free components, and the skeptical-envelope condition holds at the atom. All assigned beliefs are supported on states that can rationally generate the observed report. In particular, whenever a report identifies a unique state, the receiver assigns probability one to it.

It remains to check that these off-path beliefs do not create a profitable deviation. Types below $\ell^M$ cannot induce $a^+$ at $r^M$ or above because those reports are beyond their reach. Types in the pooling interval obtain a weakly positive payoff at $r^M$. A lower report leads to $a^-$, while a higher report leads to the same action $a^+$ at a higher reporting cost. Types above $r^M$ report truthfully and already pay no reporting cost. Thus, the proposed beliefs preserve all equilibrium incentives, proving part~$(i)$.

Now fix $\lambda>0$ and write $r^*=r^*(\lambda)$. At the pooling atom, $V^{E_\lambda}(r^*)=\lambda$. If reports immediately above $r^*$ belong to $Z_1$, the value approached from above is $u_r(r^*)$. This exceeds $\lambda$ because the pooling posterior also assigns positive probability to states below $r^*$. Because sender~$1$ is upward biased, the skeptical-envelope condition requires the value at the atom to equal the smaller value approached from its neighboring regions. Hence, the value approached from below must be
\[
\lim_{r\uparrow r^*}V^{E_\lambda}(r)=\lambda>0.
\]
The same conclusion is immediate when the pooling report can be approached only from below within $Z_1$.

Because this limiting value is positive, every report sufficiently close to $r^*$ from below also induces $a^+$. Choose one such report $r>0$ and a pooled type
\[
\theta\in \left(r,\frac{r+r^*}{2}\right).
\]
The deviation from $r^*$ to $r$ leaves the receiver's action unchanged but strictly reduces the reporting cost, because $|r-\theta|<|r^*-\theta|$. Hence,
\[
u_1(\theta)-k_1C_1(r,\theta) > u_1(\theta)-k_1C_1(r^*,\theta),
\]
contradicting sender optimality. No $E_\lambda$ with $\lambda>0$ is skeptical-envelope admissible. This proves part~$(ii)$.
\end{proof}


\begin{proof}[Proof of \Cref{lem:operator-competition-geometry}]
Sender~1's report $r_1\geq0$ can rationally originate from the states in $[\bar r_1^{-1}(r_1),r_1]$. Sender~2's report $r_2\leq0$ can rationally originate from the states in $[r_2,\underline r_2^{-1}(r_2)]$. Their intersection gives the stated formula. Throughout $Z_2$, this intersection contains zero. It is equal to $\{0\}$ only at the two stated corners.

The set of atoms of a probability distribution is at most countable. Removing a countable set of points from a rectangle in $\mathbb R^2$ does not prevent us from drawing a path between any two remaining points. To see this, choose an intermediate point in the rectangle and connect each of the original points to it with a straight line. The intermediate point can be chosen so that neither line passes through a removed point, because the excluded choices lie on only countably many lines. Hence, $K_2\setminus\mathcal A(\nu)$ is path connected.

To establish density, fix any $\bm r\in Z_2$ and any relative neighborhood $U$ of $\bm r$ in $Z_2$. The set $U\cap K_2$ is uncountable, whereas $\mathcal A(\nu)$ is countable. Hence,
\[
U\cap\left(K_2\setminus\mathcal A(\nu)\right)\neq\varnothing.
\]
Therefore, every point of $Z_2$ belongs to the closure of $K_2\setminus\mathcal A(\nu)$.
\end{proof}


\begin{proof}[Proof of \Cref{lem:AE-atomless}]
By part~$(iv)$ of \Cref{prop:AE-characterization}, conditional on a state $\theta$, sender $j$'s only possible report atom is at the truthful report $r_j=\theta$. Every non-truthful report belongs to the atomless component with density $\psi_j(\cdot\mid\theta)$.

By the independent-private-randomization convention in \Cref{subsec:common-primitives}, for every fixed report pair $(r_1,r_2)$, we have
\[
\nu^{\DE}\!\left(\{(r_1,r_2)\}\right) = \int_\Theta \phi_1(\{r_1\}\mid\theta) \phi_2(\{r_2\}\mid\theta) f(\theta)\,d\theta.
\]
The integrand can be positive only if both senders use their truthful atoms, which requires $r_1=r_2=\theta$. For a fixed report pair, this can occur at no more than one state. Because the prior has a density, every singleton state has probability zero. Therefore,
\[
\nu^{\DE}\!\left(\{(r_1,r_2)\}\right)=0.
\]
Since the report pair was arbitrary,
$\mathcal A(\nu^{\DE})=\varnothing$.
\end{proof}


\begin{proof}[Proof of \Cref{cor:common-refinement-AE}]
Let $E$ be skeptical-envelope admissible. Requirement~$(ii)$ is imposed on the full report space $\Theta^2$ and is exactly condition $(wM)$ in \citet{Vaccari2023JET}. Set
\[
C^E
\coloneqq
K_2\setminus\mathcal A(\nu^E).
\]
By \Cref{lem:operator-competition-geometry}, $C^E$ is path connected and dense in $Z_2$. Hence, $\mathscr C(\nu^E)=\{C^E\}$ and $\overline{C^E}=Z_2$. In particular, every report atom borders the unique component $C^E$.

At every non-atomic point of $K_2$, the extension $\overline V_{C^E}$ agrees with $V^E$ by requirement~$(iii)$. At every atom $\bm a\in\mathcal A(\nu^E)\cap K_2$, there is only one bordering component. Therefore,
\[
\left(\mathcal L_{\nu^E}V^E\right)(\bm a) = \overline V_{C^E}(\bm a) = \left(\mathcal U_{\nu^E}V^E\right)(\bm a).
\]
Skeptical-envelope admissibility then gives $V^E(\bm a)=\overline V_{C^E}(\bm a)$. As a result, $V^E$ and $\overline V_{C^E}$ agree throughout $K_2$.
Component-wise regularity implies
\[
\frac{\partial V^E(r_1,r_2)}{\partial r_j} = \partial_j^{C^E}\overline V_{C^E}(r_1,r_2) >0
\]
for every $j\in\{1,2\}$ whenever $\underline r_2(0)<r_2\leq0\leq r_1<\bar r_1(0)$. At the axes, the relevant derivative is understood in the feasible one-sided
sense. This is condition $(sM)$ in \citet{Vaccari2023JET}.

Plausible-state support, together with \Cref{lem:operator-competition-geometry}, implies
\[
\mu_{(0,0)}^E = \mu_{(\bar r_1(0),\underline r_2(0))}^E = \delta_0.
\]
As a result,
\[
V^E(0,0) = V^E(\bar r_1(0),\underline r_2(0)) = u_r(0) =0,
\]
which is condition $(Dom)$. Thus, the three regularity requirements together with the skeptical-envelope condition imply $(wM)$, $(sM)$, and $(Dom)$. The essential-uniqueness result in \citet{Vaccari2023JET} proves part~$(i)$.

For part~$(ii)$, take an adversarial equilibrium. Its receiver payoff difference satisfies $(wM)$, $(sM)$, and $(Dom)$. By \Cref{lem:AE-atomless}, its joint-report distribution has no point atoms, so both envelope maps are the identity on $Z_2$.

It remains to verify plausible-state support. The adversarial-equilibrium characterization implies that every strictly conflicting report pair $(r_1,r_2)\in Z_2$ with $r_2<0<r_1$ is on path. At every such pair, Bayes' rule assigns probability only to states in which both reports are used by the corresponding equilibrium strategies. Every such state belongs to $P_2(r_1,r_2)$. Hence,
\[
\mu^E_{(r_1,r_2)}\!\left(P_2(r_1,r_2)\right)=1
\]
at every strictly conflicting report pair.

It remains to complete beliefs at zero-probability boundary profiles. Write
\[
\theta_-(r_1,r_2) \coloneqq \max\left\{\bar r_1^{-1}(r_1),r_2\right\},
\]
\[
\theta_+(r_1,r_2) \coloneqq \min\left\{r_1,\underline r_2^{-1}(r_2)\right\}.
\]
Thus, $P_2(r_1,r_2) = [\theta_-(r_1,r_2),\theta_+(r_1,r_2)]$. Because $V^E$ and the reach functions are continuous, and because every such boundary profile can be approached by strictly conflicting on-path profiles, we have that $u_r\!\left(\theta_-(r_1,r_2)\right) \leq V^E(r_1,r_2) \leq u_r\!\left(\theta_+(r_1,r_2)\right)$. Whenever $\theta_-(r_1,r_2)<\theta_+(r_1,r_2)$, define
\[
\omega(r_1,r_2) \coloneqq \frac{ V^E(r_1,r_2)-u_r(\theta_-(r_1,r_2)) }{ u_r(\theta_+(r_1,r_2))-u_r(\theta_-(r_1,r_2)) } \in[0,1]
\]
and set
\[
\mu^E_{(r_1,r_2)} = \left(1-\omega(r_1,r_2)\right) \delta_{\theta_-(r_1,r_2)} + \omega(r_1,r_2) \delta_{\theta_+(r_1,r_2)}.
\]
At the two profiles for which $P_2(r_1,r_2)=\{0\}$, set $\mu^E_{(r_1,r_2)}=\delta_0$.

This completion is supported on $P_2(r_1,r_2)$ and generates exactly the original value $V^E(r_1,r_2)$. It leaves the receiver's action, all sender incentives, and conditions $(wM)$, $(sM)$, and $(Dom)$ unchanged. By \Cref{lem:AE-atomless}, both envelope maps are the identity on $Z_2$. Thus, the resulting adversarial-equilibrium assessment is skeptical-envelope admissible.

Because $\mathcal A(\nu^{\DE})=\varnothing$, the unique atom-free component is $K_2$ itself. Therefore, requirement~$(iii)$ imposes no additional continuation condition at report atoms and coincides on $K_2$ with condition $(sM)$, which the adversarial equilibrium already satisfies.
\end{proof}


\begin{proof}[Proof of \Cref{prop:mono-welfare}]
		Under full information the receiver chooses $a^+$ if and only if $\theta\geq 0$. In the monopoly equilibrium indexed by $\lambda$, the only states in which the receiver behaves differently are those in the pooling interval $(\ell(\lambda),r^*(\lambda))$. After observing the pooling report $r^*(\lambda)$, the receiver chooses $a^+$. Therefore, the only welfare loss relative to full information comes from the negative states in the pooling interval, namely the states in $(\ell(\lambda),0)$. The loss is exactly
		\[
		\int_{\ell(\lambda)}^0 u_r(\theta)f(\theta)\,d\theta,
		\]
		which is weakly negative because $u_r(\theta)<0$ for every $\theta<0$. This proves the formula and the inequality. The inequality is strict unless the pooling interval collapses at the receiver threshold, which is precisely the full-information case.
	\end{proof}


\begin{proof}[Proof of \Cref{thm:high-skepticism}]
		By \Cref{prop:mono-welfare},
		\[
		W_r^{\M}(\lambda) = W_r^{\FI} + \int_{\ell(\lambda)}^0 u_r(\theta)f(\theta)\,d\theta.
		\]
		As $\lambda$ approaches the most informative equilibrium, the pooling interval shrinks to the receiver's threshold and, therefore, $\ell(\lambda)\uparrow0$. By dominated convergence,
		\[
		\lim_{\lambda\to\bar\lambda}W_r^{\M}(\lambda)=W_r^{\FI}.
		\]
		Since $W_r^{\DE}<W_r^{\FI}$, there exists $\lambda_0<\bar\lambda$ such that $W_r^{\M}(\lambda)>W_r^{\DE}$ for every $\lambda\in(\lambda_0,\bar\lambda]$.
	\end{proof}


\begin{proof}[Proof of \Cref{prop:bounds}]
		For the lower bound, consider the feasible decision rule under which the receiver ignores sender~$2$ and chooses $a^+$ if and only if sender~$1$'s report is weakly positive. Call the welfare generated by this rule $\underline W_r^{\DE}$. Since the receiver's actual adversarial-equilibrium strategy is sequentially rational, it weakly dominates every feasible rule, so $W_r^{\DE}\geq \underline W_r^{\DE}$. Thus, it is enough to compute $\underline W_r^{\DE}$.
		
		If $\theta\geq 0$, then sender~$1$ never reports below the truth. Hence, he never reports a negative number. Under the auxiliary rule, the receiver chooses the correct action and obtains payoff $u_r(\theta)$. If $\theta\leq \theta_L$, sender~$1$ is truthful, so his report is negative and the receiver again chooses the correct action, obtaining payoff zero. If $\theta\in(\theta_L,0)$, sender~$1$ reports truthfully with probability $\alpha_1(\theta)$ and upward-misreports with probability $1-\alpha_1(\theta)$. Under the auxiliary rule the receiver therefore chooses $a^-$ with probability $\alpha_1(\theta)$ and $a^+$ with probability $1-\alpha_1(\theta)$. Her expected payoff in that state is
		\[
		0\cdot \alpha_1(\theta)+u_r(\theta)\left(1-\alpha_1(\theta)\right)=u_r(\theta)\left(1-\alpha_1(\theta)\right).
		\]
		Integrating across states yields the stated formula for $\underline W_r^{\DE}$.
		
		For the upper bound, consider a state $\theta\in(\theta_L,0)$. If sender~$2$ reports truthfully, then his report is $r_2=\theta$. Every upward misreport of sender~$1$ in that state is weakly above $s(\theta)$ by construction of the adversarial equilibrium, so sender~$1$ swings the truthful report of sender~$2$ and the receiver chooses $a^+$. Thus, on the event where sender~1 misreports upward and sender~2 reports truthfully, the receiver necessarily makes a mistake. The probability of that event is $(1-\alpha_1(\theta))\alpha_2(\theta)$. If one ignores all other mistakes the receiver may make, welfare can only go up. Therefore,
		\[
		W_r^{\DE}\leq W_r^{\FI}+\int_{\theta_L}^0 u_r(\theta)f(\theta)\left(1-\alpha_1(\theta)\right)\alpha_2(\theta)\,d\theta=\overline W_r^{\DE}.
		\]
		The strict inequalities follow because $u_r(\theta)<0$ on $(\theta_L,0)$ and because $\alpha_1,\alpha_2\in(0,1)$ on a set of positive measure.
	\end{proof}


\begin{proof}[Proof of \Cref{lem:cutoff-sufficient-condition}]
		By the lower bound in \Cref{prop:bounds},
		\[
		W_r^{\DE} \geq W_r^{\FI} + \int_{\theta_L}^0 u_r(\theta)f(\theta)\left(1-\alpha_1(\theta)\right)\,d\theta.
		\]
		By \Cref{prop:mono-welfare},
		\[
		W_r^{\M}(0) = W_r^{\FI} + \int_{\ell^M}^0 u_r(\theta)f(\theta)\,d\theta.
		\]
		Therefore,
		\[
		W_r^{\DE}-W_r^{\M}(0) \geq -\int_{\ell^M}^{\theta_L}u_r(\theta)f(\theta)\,d\theta - \int_{\theta_L}^0u_r(\theta)f(\theta)\alpha_1(\theta)\,d\theta.
		\]
		Since $u_r(\theta)<0$ for every $\theta<0$, both terms on the right-hand side are weakly positive. At least one is strictly positive under the stated assumptions. Hence, $W_r^{\DE}>W_r^{\M}(0)$.
	\end{proof}


\begin{proof}[Proof of \Cref{thm:symmetric}]
		Under the symmetric benchmark, reflecting states and reports around zero maps the competitive environment into itself while exchanging the two senders. Hence, if $s$ is an adversarial-equilibrium swing function, then
\[
\widetilde s(r)\coloneqq-s(-r)
\]
is also a swing function. Essential uniqueness of the adversarial-equilibrium outcome implies $s(r)=\widetilde s(r)$, and therefore $s(-r)=-s(r)$. This gives oddness, but one more step is needed. For $r>0$, define $g(r)\coloneqq-s(r)$. The function $g$ is strictly increasing. Oddness and the involution property $s(s(r))=r$ imply $g(g(r))=r$. An increasing involution must be the identity. If $g(r)>r$, monotonicity would give $g(g(r))>g(r)$, contrary to $g(g(r))=r$, and the case $g(r)<r$ is analogous. Thus $g(r)=r$ and
\[
s(r)=-r
\]
on the relevant domain.
		
		The left truthful cutoff is characterized by $s\left(\bar r_1(\theta_L)\right)=\theta_L$. Using $s(r)=-r$, this becomes $\bar r_1(\theta_L)=-\theta_L$. Alternatively, we can write $\bar r_1^{-1}(-\theta_L)=\theta_L$.
		
		By symmetry of $f$ and anti-symmetry of $u_r$, we have
		\begin{equation}\label{eq:monop_cond}
		\int_{\theta_L}^{-\theta_L}u_r(\theta)f(\theta)\,d\theta=0.
		\end{equation}
		Thus, the interval $(\theta_L,-\theta_L)$, pooled at report $-\theta_L$, satisfies the least-informative monopoly indifference condition \eqref{eq:monop_cond}, where $\theta_L=\bar r_1^{-1}(-\theta_L)$. By uniqueness of the least-informative monopoly pooling report, we obtain $r^M=-\theta_L=-\ell^M$.
		The welfare ranking now follows from \Cref{lem:cutoff-sufficient-condition}, because $\theta_L=\ell^M$.
	\end{proof}


\begin{proof}[Proof of \Cref{prop:exact-welfare}]
		Let $q(\theta)$ denote the equilibrium probability that the receiver chooses the positive action in state $\theta$. By definition,
		\[
		W_r^{\DE}=\int_{\Theta} u_r(\theta)f(\theta)q(\theta)\,d\theta.
		\]
		
		Consider first a state $\theta\geq \theta_H$. By the definition of the right truthful cutoff, sender~$2$ cannot profitably swing a truthful report by sender~$1$. Hence, both senders report truthfully with probability one and the receiver chooses the positive action with probability one. Therefore, $q(\theta)=1$ for every $\theta\geq \theta_H$. Now consider a state $\theta\leq \theta_L$. By the definition of the left truthful cutoff, sender~$1$ cannot profitably swing a truthful report by sender~$2$. Hence, both senders again report truthfully with probability one, and the receiver chooses the positive action with probability zero. Therefore, $q(\theta)=0$ for every $\theta\leq \theta_L$. It remains to characterize $q(\theta)$ on the mixing regions.
		
		\medskip
		
		\noindent
		\textbf{Positive states.} Fix $\theta\in[0,\theta_H)$.	In this region sender~$1$ is truthful with probability $\alpha_1(\theta)$ and continuously misreports upward with density $\psi_1(\cdot\mid \theta)$ on $\left[\theta,s(\underline{r}_2 (\theta))\right)$. Sender~$2$ is truthful with probability $\alpha_2(\theta)$ and continuously misreports downward with cumulative distribution function $\Psi_2(\cdot\mid \theta)$ on $\left[\underline{r}_2 (\theta),s(\theta)\right)$.
		
		The receiver makes a mistake if and only if sender~$2$ swings sender~$1$'s report. There are two disjoint cases. First, sender~$1$ reports truthfully, which occurs with probability $\alpha_1(\theta)$. Conditional on that event, sender~$2$ swings the truthful report $\theta$ if and only if he misreports. That happens with probability $1-\alpha_2(\theta)$. Hence, this case contributes $\alpha_1(\theta)\left(1-\alpha_2(\theta)\right)$ to the conditional mistake probability.
		
		Second, sender~$1$ misreports to some $r\in[\theta,s(\underline r_2(\theta)))$. Conditional on that report, sender~$2$ swings if and only if his report lies below $s(r)$. Since $\Psi_2(\cdot\mid\theta)$ is the sub-distribution function of sender~$2$'s unconditional continuous component, that probability is $\Psi_2(s(r)\mid\theta)$. Integrating against sender~$1$'s unconditional continuous density yields the second term in \eqref{eq:mplus-def}. Hence, the total conditional mistake probability in a positive state is exactly $m_+(\theta)$, and therefore $q(\theta)=1-m_+(\theta)$ for every $\theta\in[0,\theta_H)$.
		
		\medskip
		\noindent
		\textbf{Negative states.} Fix $\theta\in(\theta_L,0]$. In this region sender~$2$ is truthful with probability $\alpha_2(\theta)$ and continuously misreports downward with density $\psi_2(\cdot\mid \theta)$ on $\left[s(\bar r_1(\theta)),\theta\right)$. Sender~$1$ is truthful with probability $\alpha_1(\theta)$ and continuously misreports upward with cumulative distribution function $\Psi_1(\cdot\mid \theta)$ on $\left[s(\theta),\bar r_1(\theta)\right)$.
		
		The receiver now makes a mistake if and only if sender~$1$ swings sender~$2$'s report. Again there are two disjoint cases. First, sender~$2$ reports truthfully, which occurs with probability $\alpha_2(\theta)$. Conditional on that event, sender~$1$ swings the truthful report $\theta$ if and only if he misreports. That occurs with probability $1-\alpha_1(\theta)$. Hence, this case contributes $\alpha_2(\theta)\left(1-\alpha_1(\theta)\right)$ to the conditional mistake probability.
		
		Second, sender~$2$ misreports to some $r\in[s(\bar r_1(\theta)),\theta)$. Conditional on that report, sender~$1$ swings if and only if his continuous report is weakly above $s(r)$. Since sender~$1$'s continuous component has total mass $1-\alpha_1(\theta)$, and since $\Psi_1(s(r)\mid\theta)$ is the mass of that component weakly below $s(r)$, this probability is $1-\alpha_1(\theta)-\Psi_1(s(r)\mid\theta)$. Integrating against sender~$2$'s unconditional continuous density yields the second term in \eqref{eq:mminus-def}. Hence, the total conditional mistake probability in a negative state is $m_-(\theta)$, and therefore $q(\theta)=m_-(\theta)$ for every $\theta\in(\theta_L,0]$.
		
		Substituting the piece-wise characterization of $q(\theta)$ into
		\[
		W_r^{\DE}=\int_{\Theta} u_r(\theta)f(\theta)q(\theta)\,d\theta
		\]
		gives \eqref{eq:exact-welfare-alt}.	Subtracting and adding
		\[
		\int_{0}^{\theta_H}u_r(\theta)f(\theta)\,d\theta
		\]
		then yields \eqref{eq:exact-welfare}, because
		\[
		W_r^{\FI} = \int_{0}^{\theta_H} u_r(\theta)f(\theta)\,d\theta +\int_{\theta_H}^{\sup\Theta} u_r(\theta)f(\theta)\,d\theta.
		\]
		This completes the proof.
	\end{proof}


\begin{proof}[Proof of \Cref{cor:exact-comparison}]
		By \Cref{prop:exact-welfare}, we have $W_r^{\DE}=W_r^{\FI}-\mathcal L_+^{\DE}-\mathcal L_-^{\DE}$. By the monopoly welfare formula, $W_r^{\M}(0)=W_r^{\FI}-\mathcal L^{\M}$. Subtracting the second identity from the first yields
		\[
		W_r^{\DE}-W_r^{\M}(0) = \mathcal L^{\M}-\mathcal L_+^{\DE}-\mathcal L_-^{\DE}.
		\]
		Therefore, $W_r^{\DE}>W_r^{\M}(0)$ if and only if $\mathcal L_+^{\DE}+\mathcal L_-^{\DE}<\mathcal L^{\M}$. This proves \eqref{eq:exact-comparison}.
	\end{proof}


\begin{proof}[Proof of \Cref{cor:exact-symmetric}]
		Under the symmetry assumptions of \Cref{thm:symmetric}, the prior is symmetric around zero, the receiver's payoff satisfies $u_r(-\theta)=-u_r(\theta)$, the truthful cutoffs satisfy $\theta_H=-\theta_L$, and the adversarial-equilibrium strategies are mirror images of each other. In particular, if $\theta>0$, then the probability of a mistaken negative decision in state $\theta$ is equal to the probability of a mistaken positive decision in state $-\theta$. Hence, $m_+(\theta)=m_-(-\theta)$.
		
		Using the change of variable $x=-\theta$ in the positive-state loss term, together with the symmetry of $f$ and $u_r$, yields
		\begin{align*}
			\mathcal L_+^{\DE}
			&=
			\int_{0}^{\theta_H} u_r(\theta)f(\theta)m_+(\theta)\,d\theta \\
			&=
			\int_{0}^{-\theta_L} u_r(\theta)f(\theta)m_-(-\theta)\,d\theta \\
			&=
			\int_{\theta_L}^{0} \left[-u_r(x)\right]f(x)m_-(x)\,dx
			=
			\mathcal L_-^{\DE}.
		\end{align*}
		This proves \eqref{eq:symmetric-losses}. Substituting $\mathcal L_+^{\DE}=\mathcal L_-^{\DE}$ into $W_r^{\DE}=W_r^{\FI}-\mathcal L_+^{\DE}-\mathcal L_-^{\DE}$ gives
		\[
		W_r^{\DE} = W_r^{\FI}-2\mathcal L_-^{\DE} = W_r^{\FI} + 2\int_{\theta_L}^{0} u_r(\theta)f(\theta)m_-(\theta)\,d\theta,
		\]
		which is \eqref{eq:exact-welfare-symmetric}. The final inequality is then just the exact comparison criterion in \Cref{cor:exact-comparison} specialized to the symmetric case.
	\end{proof}


\begin{proof}[Proof of \Cref{lem:high-k2-convergence}]
Part~$(iii)$ of \Cref{prop:AE-characterization} implies that both senders report truthfully outside $(\theta_L,\theta_H)$. We show that this mixing interval shrinks to zero as $k_2$ grows.

First fix $\theta>0$. If $\theta\geq\tau_2$, then $\theta\geq\theta_H$ for every $k_2$, so both senders are already truthful. Now suppose $0<\theta<\tau_2$. Sender~$2$'s lower reach satisfies
\[
k_2 C_2\!\left(\underline r_2(\theta),\theta\right) =-u_2(\theta).
\]
The right-hand side is fixed. Hence $C_2(\underline r_2(\theta),\theta)\to0$. Because the lower reach is below $\theta$ and the cost is strictly increasing in the distance from the truth, this implies $\underline r_2(\theta)\to\theta$. Thus, $\underline r_2(\theta)>0$ for all sufficiently large $k_2$. If $\theta$ lies in the swing-function domain, then $s(\theta)<0$ and, therefore,
\[
s(\theta)-\underline r_2(\theta)<0.
\]
Because $s(x)-\underline r_2(x)$ is strictly decreasing and vanishes at $x=\theta_H$, this inequality gives $\theta>\theta_H$. If $\theta$ lies above the swing-function domain, the same conclusion is immediate. Hence, both senders are truthful at every fixed $\theta>0$ once $k_2$ is large enough.

Now fix $\theta<0$. If $\theta\leq\tau_1$, then $\theta<\theta_L$ because $\theta_L>\tau_1$, and both senders report truthfully. Suppose, therefore, that $\tau_1<\theta<0$, so that $\bar r_1(\theta)$ is well defined. If $\bar r_1(\theta)\leq0$, sender~$1$ cannot reach a positive report. Since $\bar r_1(\theta_L)=s(\theta_L)>0$ and $\bar r_1$ is increasing, we have $\theta<\theta_L$. Suppose instead that $\bar r_1(\theta)>0$. At state zero, sender~$2$'s lower reach satisfies
\[
k_2 C_2\!\left(\underline r_2(0),0\right)=-u_2(0),
\]
so $C_2(\underline r_2(0),0)\to0$. The same distance argument gives $\underline r_2(0)\to0$. Moreover, $s\!\left(\bar r_1(\theta)\right) \geq \underline r_2(0)$. For all sufficiently large $k_2$, the right-hand side is greater than the fixed negative state $\theta$. Hence, $s\!\left(\bar r_1(\theta)\right)-\theta>0$. The function $x\mapsto s(\bar r_1(x))-x$ is strictly decreasing and vanishes at $x=\theta_L$. Therefore, $\theta<\theta_L$, and both senders again report truthfully.

Thus, for every fixed $\theta\neq0$, equilibrium reports are eventually truthful. The receiver then takes her full-information action, both mistake probabilities are zero, and both misreporting costs are zero. This proves all the stated limits.
\end{proof}


\begin{proof}[Proof of \Cref{prop:high-k2-threshold}]
	By \Cref{lem:high-k2-convergence}, we have $\widetilde m_+(\theta\mid\tau_2,k_2)\to0$ for every $\theta>0$ and $\widetilde m_-(\theta\mid\tau_2,k_2)\to0$ for every $\theta<0$. Moreover, recall that $0\leq \widetilde m_+(\theta\mid\tau_2,k_2)\leq1$ and $0\leq \widetilde m_-(\theta\mid\tau_2,k_2)\leq1$. Since $|u_r(\theta)|f(\theta)$ is integrable, dominated convergence gives \[
	\lim_{k_2\to\infty}\mathcal L_+^{\DE}(\tau_2,k_2)=\lim_{k_2\to\infty}\mathcal L_-^{\DE}(\tau_2,k_2)=0.
	\]
	
	By \Cref{cor:exact-comparison}, $W_r^{\DE}(\tau_2,k_2)>W_r^{\M}(0)$	if and only if $\mathcal L_+^{\DE}(\tau_2,k_2)+\mathcal L_-^{\DE}(\tau_2,k_2)<\mathcal L^{\M}$. Since $\mathcal L^{\M}>0$ and the left-hand side converges to zero, there exists a finite $\bar k_2(\tau_2)$ such that the inequality holds for every $k_2\geq\bar k_2(\tau_2)$.
\end{proof}


\begin{proof}[Proof of \Cref{cor:high_cost}]
		Fix $\lambda<\bar\lambda$. Define the monopoly loss relative to full information by
		\[
		\mathcal L^{\M}(\lambda) \coloneqq -\int_{\ell(\lambda)}^0u_r(\theta)f(\theta)\,d\theta.
		\]
		Since $\lambda<\bar\lambda$, the monopoly equilibrium is not fully revealing, so $\mathcal L^{\M}(\lambda)>0$. By \Cref{prop:exact-welfare}, we have
		\[
		W_r^{\DE}(\tau_2,k_2) = W_r^{\FI} - \mathcal L_+^{\DE}(\tau_2,k_2) - \mathcal L_-^{\DE}(\tau_2,k_2),
		\]
		while by \Cref{prop:mono-welfare} we obtain $W_r^{\M}(\lambda)=W_r^{\FI}-\mathcal L^{\M}(\lambda)$. Therefore,
		\[
		W_r^{\DE}(\tau_2,k_2)-W_r^{\M}(\lambda) = \mathcal L^{\M}(\lambda) - \mathcal L_+^{\DE}(\tau_2,k_2) - \mathcal L_-^{\DE}(\tau_2,k_2).
		\]
		By \Cref{prop:high-k2-threshold}, $\mathcal L_+^{\DE}(\tau_2,k_2)+\mathcal L_-^{\DE}(\tau_2,k_2)\to0$. Since $\mathcal L^{\M}(\lambda)>0$, the desired strict inequality holds for all sufficiently large $k_2$.
	\end{proof}


\begin{proof}[Proof of \Cref{prop:Gs-positive}]
		Fix $r\in(0,\bar r_1(0))$ and write $s\coloneqq s(r)$. Let
		\begin{equation}\label{eq:I}
			I(r,s)\coloneqq \left[\max\{s,\bar r_1^{-1}(r)\},\min\{r,\underline{r}_2 ^{-1}(s)\}\right].
		\end{equation}
		Differentiating \eqref{eq:G} with respect to $s$ and using Leibniz' rule gives
		\begin{align*}
			G_s(r,s)
			&= \int_{I(r,s)} \frac{4\theta f(\theta)(r-\theta)}{(\theta-\tau_1)(\theta-\tau_2)}\,d\theta \\
			&\qquad +\mathds{1}_{\left\{\underline{r}_2 ^{-1}(s)<r\right\}} \frac{4\underline{r}_2 ^{-1}(s)f(\underline{r}_2 ^{-1}(s))(s-\underline{r}_2 ^{-1}(s))(r-\underline{r}_2 ^{-1}(s))}{(\underline{r}_2 ^{-1}(s)-\tau_1)(\underline{r}_2 ^{-1}(s)-\tau_2)} \frac{\partial \underline{r}_2 ^{-1}(s)}{\partial s}.
		\end{align*}
		There is no lower-bound term because the lower limit either does not depend on $s$ or makes the integrand vanish through the factor $(s-\theta)$. Likewise, the upper-bound term appears only when the upper limit is $\underline{r}_2 ^{-1}(s)$. If the upper limit is $r$, the factor $(r-\theta)$ makes the boundary contribution vanish.
		
		We first show that the integral term is strictly positive. Since $G(r,s)=0$, we have that
		\[
		\int_{I(r,s)} \frac{4\theta f(\theta)(s-\theta)(r-\theta)}{(\theta-\tau_1)(\theta-\tau_2)}\,d\theta=0.
		\]
		Rearranging yields
		\[
		\int_{I(r,s)} \frac{4\theta f(\theta)(r-\theta)}{(\theta-\tau_1)(\theta-\tau_2)}\,d\theta = \frac{1}{s} \int_{I(r,s)} \frac{4\theta^2 f(\theta)(r-\theta)}{(\theta-\tau_1)(\theta-\tau_2)}\,d\theta.
		\]
		On $I(r,s)$ we have $r-\theta>0$, $\theta-\tau_1>0$, and $\theta-\tau_2<0$. Therefore, the integrand on the right-hand side is strictly negative. Since $s<0$, the right-hand side is strictly positive. Hence, the integral term is strictly positive.
		
		Now consider the boundary term. From \eqref{eq:r2-inv} one obtains
		\[
		\frac{\partial \underline{r}_2 ^{-1}(s)}{\partial s}=1-\frac{1}{\sqrt{1+4k_2(\tau_2-s)}}>0.
		\]
		Whenever the indicator is one, $\underline{r}_2 ^{-1}(s)$ belongs to $(0,\tau_2)$, and so $\underline{r}_2 ^{-1}(s)>0$, $s<\underline{r}_2 ^{-1}(s)$, and $r>\underline{r}_2 ^{-1}(s)$. As a result, the fraction multiplying $\partial \underline{r}_2 ^{-1}(s)/\partial s$ is strictly positive. Thus, the boundary term is weakly positive. Combining the two parts yields $G_s(r,s(r))>0$.

        At the endpoint, strict monotonicity of the swing function on $[\underline r_2(0),\bar r_1(0)]$ implies $s\!\left(\bar r_1(0)\right)=\underline r_2(0)$. Moreover, $\bar r_1^{-1}\!\left(\bar r_1(0)\right)=0=\underline r_2^{-1}\!\left(\underline r_2(0)\right)$. As a result, $I\!\left(\bar r_1(0),\underline r_2(0)\right)=[0,0]=\left\{ 0 \right\}$. Therefore, the integral term in the displayed formula for $G_s$ is zero. The boundary term is also zero because its numerator contains the factor $\underline r_2^{-1}(\underline r_2(0))=0$. Hence, $G_s\!\left(\bar r_1(0),s(\bar r_1(0))\right)=0$.	
\end{proof}


\begin{proof}[Proof of \Cref{prop:cost-effect}]
First consider $r=\bar r_1(0)$. The endpoint identity for the swing function and the LQ reach formula give
\[
s_{k_2}\!\left(\bar r_1(0)\right) = \underline r_2(0\mid\tau_2,k_2) = -\sqrt{\frac{\tau_2}{k_2}}.
\]
Therefore,
\[
\frac{\partial s_{k_2}(\bar r_1(0))}{\partial k_2} = \frac{\sqrt{\tau_2}}{2k_2^{3/2}} >0.
\]
Moreover,
\[
\underline r_2^{-1}\!\left( s_{k_2}(\bar r_1(0)) \right) = 0 < \bar r_1(0),
\]
so this is precisely the binding case in the proposition. This proves the endpoint conclusion directly, without invoking the implicit function theorem.

For the remainder of the proof, fix $r\in(0,\bar r_1(0))$ and write $s\coloneqq s_{k_2}(r)$.

In the LQ swing equation, $k_2$ enters only through the upper integration limit
\[
b_{k_2}(s) \coloneqq \min\left\{r,\underline r_2^{-1}(s)\right\}.
\]
The inverse reach $\underline r_2^{-1}(s)$ is strictly decreasing in $k_2$. Whenever it binds, it lies in $(0,r)$, where the integrand in \eqref{eq:G} is strictly positive. It follows that $G(r,s\mid\tau_2,k_2)$ is weakly decreasing in $k_2$ for every fixed admissible $s$, and strictly decreasing whenever the inverse-reach boundary moves in the interior.

The adversarial-equilibrium characterization gives a unique swing root, and \Cref{prop:Gs-positive} shows that $G$ crosses zero from below at that root. Therefore, if $\widehat k_2\geq k_2$, the point-wise inequality
\[
G\left(r,s\mid\tau_2,\widehat k_2\right) \leq G\left(r,s\mid\tau_2,k_2\right)
\]
implies $s_{\widehat k_2}(r)\geq s_{k_2}(r)$. If the old root lies below the admissible report interval for $\widehat k_2$, this inequality holds immediately from the movement of that interval's lower endpoint.

At differentiability points, implicit differentiation gives
\[
\frac{\partial s_{k_2}(r)}{\partial k_2} =- \frac{G_{k_2}(r,s_{k_2}(r))} {G_s(r,s_{k_2}(r))}.
\]
Here $G_s>0$, while $G_{k_2}=0$ when the inverse reach is non-binding and $G_{k_2}<0$ when it binds. This gives the derivative statements. The same argument applied from either side gives the claim at boundary-switching points.
\end{proof}


\begin{proof}[Proof of \Cref{prop:bias-effect}]
		Under the strict non-binding condition $\underline r_2^{-1}(s(r))>r$, the upper integration limit in \eqref{eq:G} is locally equal to $r$. Hence, there is no boundary term when differentiating with respect to $\tau_2$. We obtain
		\[
		G_{\tau_2}(r,s)=\int_{I(r,s)} \frac{4\theta f(\theta)(s-\theta)(r-\theta)}{(\theta-\tau_1)(\theta-\tau_2)^2}\,d\theta.
		\]
		Using again the identity $G(r,s)=0$, rewrite
		\[
		0=\int_{I(r,s)} \frac{4\theta f(\theta)(s-\theta)(r-\theta)(\theta-\tau_2)}{(\theta-\tau_1)(\theta-\tau_2)^2}\,d\theta.
		\]
		Expanding $\theta(\theta-\tau_2)=\theta^2-\tau_2\theta$ and rearranging yields
		\[
		\tau_2 G_{\tau_2}(r,s)=\int_{I(r,s)} \frac{4\theta^2 f(\theta)(s-\theta)(r-\theta)}{(\theta-\tau_1)(\theta-\tau_2)^2}\,d\theta.
		\]
		On $I(r,s)$ the factors satisfy $\theta^2>0$, $s<\theta$, $r>\theta$, $\theta>\tau_1$, and $(\theta-\tau_2)^2>0$. Hence, the integrand on the right-hand side is strictly negative. Because $\tau_2>0$, it follows that $G_{\tau_2}(r,s)<0$. \Cref{prop:Gs-positive} then gives
		\[
		\frac{\partial s(r)}{\partial \tau_2}=-\frac{G_{\tau_2}(r,s(r))}{G_s(r,s(r))}>0.
		\]
	\end{proof}


\begin{proof}[Proof of \Cref{cor:more-costly}]
			By \Cref{thm:symmetric}, the symmetric benchmark satisfies $s_{-\tau_1,k_1}(r^M)=\ell^M$. Holding $\tau_2=-\tau_1$ fixed and increasing $k_2$, \Cref{prop:cost-effect} implies that, along the continuous adversarial-equilibrium swing root path,
			\[
			s_{-\tau_1,\widehat k_2}(r^M) \geq s_{-\tau_1,k_1}(r^M) = \ell^M.
			\]
		
			Let $\widehat\theta_L$ denote the left truthful cutoff in the adversarial equilibrium with parameters $(-\tau_1,\widehat k_2)$. Define
		\[
		\Lambda(\theta) \coloneqq s_{-\tau_1,\widehat k_2}\left(\bar r_1(\theta)\right)-\theta .
		\]
		Because $\bar r_1$ is strictly increasing and the swing function is strictly decreasing, $\Lambda$ is strictly decreasing. The cutoff $\widehat\theta_L$ is characterized by $\Lambda(\widehat\theta_L)=0$. Moreover, since $\bar r_1(\ell^M)=r^M$, we have
		\[
		\Lambda(\ell^M) = s_{-\tau_1,\widehat k_2}(r^M)-\ell^M \geq0.
		\]
		Because $\Lambda$ is strictly decreasing, this implies $\widehat\theta_L\geq\ell^M$. The conclusion follows from \Cref{lem:cutoff-sufficient-condition}.
		\end{proof}


\begin{proof}[Proof of \Cref{prop:condition_lq}]
		Under the LQ setting, sender~2's lower reach is obtained from $k_2(\theta-r_2)^2=\tau_2-\theta$ for $\theta\leq\tau_2$ and $r_2\leq\theta$. Thus,
		\begin{equation}\label{eq:lq_reach_sender2}
			\underline r_2(\theta) = \theta-\sqrt{\frac{\tau_2-\theta}{k_2}}.
		\end{equation}
		Differentiating gives
		\[
		\frac{\partial \underline r_2(\theta)}{\partial \tau_2} = -\frac{1}{2\sqrt{k_2}\sqrt{\tau_2-\theta}} <0
		\]
		and
		\[
		\frac{\partial \underline r_2(\theta)}{\partial k_2} = \frac{\sqrt{\tau_2-\theta}}{2k_2^{3/2}} >0.
		\]
		Since $\underline r_2(\cdot)$ is strictly increasing, its inverse exists on its image. Differentiating the identity $\underline r_2\left(\underline r_2^{-1}(r_2)\right)=r_2$ with respect to $\tau_2$ and $k_2$ gives \[ \frac{\partial \underline r_2^{-1}(r_2)}{\partial \tau_2} = - \frac{\partial \underline r_2(\theta)/\partial\tau_2} {\partial \underline r_2(\theta)/\partial\theta} \Bigg|_{\theta=\underline r_2^{-1}(r_2)} >0
		\]
		and
		\[
		\frac{\partial \underline r_2^{-1}(r_2)}{\partial k_2} = - \frac{\partial \underline r_2(\theta)/\partial k_2} {\partial \underline r_2(\theta)/\partial\theta} \Bigg|_{\theta=\underline r_2^{-1}(r_2)} <0,
		\]
		because $\partial \underline r_2(\theta)/\partial\theta>0$. The final statement follows immediately.
	\end{proof}


\begin{proof}[Proof of \Cref{prop:condition_inverse}]
		Write $r\coloneqq r^M$ and $\ell\coloneqq \ell^M$.	In the LQ class, sender~$2$'s lower reach at state $r$ is
		\[
		\underline r_2(r\mid\tau_2,k_2) = r-\sqrt{\frac{\tau_2-r}{k_2}}.
		\]
		For each $n$, define
		\[
		d_n \coloneqq \sqrt{\frac{\tau_2^n-r}{k_2^n}}.
		\]
		Since $\tau_2^n\to+\infty$ and $k_2^n\to0^+$, we have $d_n\to+\infty$. Let
		\[
		s_n^0\coloneqq r-d_n = \underline r_2(r\mid\tau_2^n,k_2^n).
		\]
		Then, $\underline r_2^{-1}(s_n^0\mid\tau_2^n,k_2^n)=r$. It is enough to show that $s_n>s_n^0$ for all sufficiently large $n$.
		
		For all sufficiently large $n$, $s_n^0<\ell$. At $s=s_n^0$, the upper endpoint of the swing integral is $r$, while the lower endpoint is $\ell$. Hence,
		\[
		G(r,s_n^0\mid\tau_2^n,k_2^n) = \int_\ell^r \frac{ 4\theta f(\theta)(s_n^0-\theta)(r-\theta) }{ (\theta-\tau_1)(\theta-\tau_2^n) } \,d\theta.
		\]
		Substituting $s_n^0=r-d_n$, we get
		\[
		G(r,s_n^0\mid\tau_2^n,k_2^n) = \int_\ell^r \frac{ 4\theta f(\theta)\left(r-\theta-d_n\right)(r-\theta) }{ (\theta-\tau_1)(\theta-\tau_2^n) } \,d\theta.
		\]
		Multiply by the positive scalar $\tau_2^n/d_n$ and define
		\[
		\widehat G_n \coloneqq \frac{\tau_2^n}{d_n} G(r,s_n^0\mid\tau_2^n,k_2^n).
		\]
		Then,
		\[
		\widehat G_n = \int_\ell^r 4\theta f(\theta) \frac{r-\theta}{\theta-\tau_1} \left(1-\frac{r-\theta}{d_n}\right) \frac{\tau_2^n}{\tau_2^n-\theta} \,d\theta.
		\]
		The two factors, $1-(r-\theta)/d_n$ and $\tau_2^n/(\tau_2^n-\theta)$ converge uniformly to one on $[\ell,r]$. Therefore,
		\[
		\widehat G_n \to 4\int_\ell^r \theta f(\theta) \frac{r-\theta}{\theta-\tau_1} \,d\theta.
		\]
		Let
		\[
		H(\theta)\coloneqq \frac{r-\theta}{\theta-\tau_1}.
		\]
		Then, $H$ is strictly decreasing on $[\ell,r]$, since
		\[
		H'(\theta) = -\frac{r-\tau_1}{(\theta-\tau_1)^2}<0.
		\]
		By the least-informative monopoly indifference condition in the linear receiver-payoff case,
		\[
		\int_\ell^r \theta f(\theta)\,d\theta=0.
		\]
		Hence,
		\[
		\int_\ell^r \theta f(\theta)H(\theta)\,d\theta = \int_\ell^r \theta f(\theta)\left[H(\theta)-H(0)\right]\,d\theta.
		\]
		For $\theta<0$, $H(\theta)>H(0)$, so $\theta\left[H(\theta)-H(0)\right]<0$. For $\theta>0$, $H(\theta)<H(0)$, so again $\theta\left[H(\theta)-H(0)\right]<0$. Since $f$ has full support and $\ell<0<r$, the integrand is strictly negative on a set of positive measure. Therefore,
		\[
		\int_\ell^r \theta f(\theta)H(\theta)\,d\theta<0.
		\]
		It follows that $\widehat G_n<0$ for all sufficiently large $n$. Since $\tau_2^n/d_n>0$, we have $G(r,s_n^0\mid\tau_2^n,k_2^n)<0$ for all sufficiently large $n$.
		
		The adversarial-equilibrium characterization gives a unique zero $s_n$ of $G(r,\cdot\mid\tau_2^n,k_2^n)$. By \Cref{prop:Gs-positive}, this zero is a positive crossing. Continuity and uniqueness imply that $G$ is negative below $s_n$ and positive above it. Since $G(r,s_n^0\mid\tau_2^n,k_2^n)<0$, we must have $s_n^0<s_n$. Recalling that $s_n^0=\underline r_2(r\mid\tau_2^n,k_2^n)$, we obtain that $s_n>\underline r_2(r\mid\tau_2^n,k_2^n)$. Because $\underline r_2(\cdot\mid\tau_2^n,k_2^n)$ is strictly increasing, this is equivalent to
\[
\underline r_2^{-1}(s_n\mid\tau_2^n,k_2^n)>r.
\]
Substituting $r=r^M$ proves the result.
	\end{proof}


\begin{proof}[Proof of \Cref{lem:total-welfare-high-k2}]
		In the LQ class, we have $u_r(\theta)=\theta$, $u_1(\theta)=\theta-\tau_1$, and $u_2(\theta)=\theta-\tau_2$. Therefore, 
		\[
		S(\theta) \coloneqq u_r(\theta)+u_1(\theta)+u_2(\theta) = 3\theta-\tau_1-\tau_2
		\]
		is affine in $\theta$. Since $f$ has finite first moment, $S(\theta)f(\theta)$ is integrable.
		
		By \Cref{lem:high-k2-convergence}, for every fixed $\theta\neq0$,
		\[
		q^{\DE}(\theta\mid\tau_2,k_2) \to \mathds{1}_{\{\theta\geq0\}}.
		\]
		Since $0\leq q^{\DE}(\theta\mid\tau_2,k_2)\leq1$, dominated convergence gives
		\[
		\int_\Theta S(\theta) q^{\DE}(\theta\mid\tau_2,k_2)f(\theta)\,d\theta \to \int_0^\infty S(\theta)f(\theta)\,d\theta = \mathcal W_{\mathrm{tot}}^{\FI}.
		\]
		
		It remains to show that expected communication costs vanish in expectation. By \Cref{lem:high-k2-convergence}, for every fixed $\theta\neq0$, we have $\mathcal C_j^{\DE}(\theta\mid\tau_2,k_2)\to0$ for $j\in\{1,2\}$. Moreover, truthful reporting is feasible and costless. Hence, any report used with positive probability in equilibrium must give sender~$j$ an expected payoff at least as high as truthful reporting. Since the action payoff difference for sender~$j$ is bounded in absolute value by $|u_j(\theta)|$, the expected communication cost at state $\theta$ is bounded by
		\[
		0 \leq \mathcal C_j^{\DE}(\theta\mid\tau_2,k_2) \leq |u_j(\theta)|.
		\]
		Because $u_j$ is affine and $f$ has finite first moment, $|u_j(\theta)|f(\theta)$ is integrable. Dominated convergence implies, for $j\in\{1,2\}$,
		\[
		\int_\Theta \mathcal C_j^{\DE}(\theta\mid\tau_2,k_2)f(\theta)\,d\theta \to0.
		\]
		Combining the convergence of the action-payoff term with the convergence of communication costs proves $\mathcal W_{\mathrm{tot}}^{\DE}(\tau_2,k_2) \to \mathcal W_{\mathrm{tot}}^{\FI}$.
	\end{proof}


\begin{proof}[Proof of \Cref{thm:total-welfare-monopoly-dominates}]
		Recall that $S(\theta)=u_r(\theta)+u_1(\theta)+u_2(\theta)$. Under \Cref{def:LQ}, $S(\theta)=3\theta-\tau_1-\tau_2$. Since $u_r(\theta)=\theta$, full information induces action $a^+$ if and only if $\theta\geq0$. Therefore,
		\[
		\mathcal W_{\mathrm{tot}}^{\FI} = \int_0^\infty S(\theta)f(\theta)\,d\theta.
		\]
		
		Fix $-\tau_1>3\tau_2$. Let $\ell^M$ and $r^M$ denote the lower pooling type and pooling report in the least informative monopoly equilibrium. To emphasize their dependence on $k_1$, write $\ell_k\coloneqq \ell^M$, $r_k\coloneqq r^M$, and $d_k\coloneqq r_k-\ell_k$. The least informative monopoly equilibrium is characterized by
		\[
		\int_{\ell_k}^{r_k}\theta f(\theta)\,d\theta=0
		\]
		and by the reach condition
		\[
		r_k=\ell_k+\sqrt{\frac{\ell_k-\tau_1}{k_1}}.
		\]
		Hence, $d_k^2=(\ell_k-\tau_1)/k_1$. Since $\ell_k<0$, we have $d_k^2\leq (-\tau_1)/k_1$, and, therefore, $d_k\to0$ as $k_1\to\infty$. Because $\ell_k<0<r_k$ in the least informative monopoly equilibrium, this implies $\ell_k\to0$ and $r_k\to0$.
		
		We next derive the relative location of the two endpoints. Since $f$ is continuous at zero and $f(0)>0$,
		\[
		0 = \int_{\ell_k}^{r_k}\theta f(\theta)\,d\theta = \frac{f(0)}{2}\left(r_k^2-\ell_k^2\right)+o\left(d_k^2\right).
		\]
		Alternatively, we can write
		\[
		0 = \frac{f(0)}{2}d_k(r_k+\ell_k)+o\left(d_k^2\right).
		\]
		Dividing by $d_k>0$ gives $r_k+\ell_k=o(d_k)$. Thus, $r_k/d_k\to\frac12$ and $\ell_k/d_k\to-\frac12$. Moreover, from $k_1d_k^2=\ell_k-\tau_1$ and $\ell_k\to0$, we get $k_1d_k^2\to -\tau_1$.
		
		In the least informative monopoly equilibrium, the receiver chooses $a^+$ if and only if $\theta\geq \ell_k$. Sender~1 pools the interval $(\ell_k,r_k)$ at report $r_k$, and therefore incurs cost $k_1(r_k-\theta)^2$	on that interval. Sender~2 is muted and incurs no communication cost. Hence, the difference between total welfare in the least informative monopoly equilibrium and full-information total welfare is
		\[
		\mathcal W_{\mathrm{tot}}^{\M}(0)-\mathcal W_{\mathrm{tot}}^{\FI} = \int_{\ell_k}^{0}S(\theta)f(\theta)\,d\theta - k_1\int_{\ell_k}^{r_k}(r_k-\theta)^2f(\theta)\,d\theta.
		\]
		We evaluate this difference as $k_1\to\infty$.
		
		First,
		\[
		\frac{1}{d_k} \int_{\ell_k}^{0}S(\theta)f(\theta)\,d\theta \to \frac{f(0)}{2}(-\tau_1-\tau_2),
		\]
		because $S(\theta)=-\tau_1-\tau_2+3\theta$ and the interval $[\ell_k,0]$ has length asymptotic to $d_k/2$.
		
		Second,
		\[
		\int_{\ell_k}^{r_k}(r_k-\theta)^2f(\theta)\,d\theta = \frac{f(0)}{3}d_k^3+o\left(d_k^3\right).
		\]
		Therefore,
		\[
		\frac{1}{d_k} k_1 \int_{\ell_k}^{r_k}(r_k-\theta)^2f(\theta)\,d\theta \to \frac{f(0)}{3}(-\tau_1),
		\]
		because $k_1d_k^2\to -\tau_1$.
		
		Combining the two limits gives
		\[
		\frac{ \mathcal W_{\mathrm{tot}}^{\M}(0)-\mathcal W_{\mathrm{tot}}^{\FI} }{d_k} \to f(0) \left[ \frac{-\tau_1-\tau_2}{2} - \frac{-\tau_1}{3} \right] = \frac{f(0)}{6}(-\tau_1-3\tau_2).
		\]
		By assumption, $-\tau_1>3\tau_2$, so the limit is strictly positive. Hence, there exists $\bar k_1<\infty$ such that $\mathcal W_{\mathrm{tot}}^{\M}(0)>
		\mathcal W_{\mathrm{tot}}^{\FI}$ for every $k_1\geq \bar k_1$. This proves the proposition.
	\end{proof}


\begin{proof}[Proof of
\Cref{cor:finite-k2-total-welfare-dominance}]
Fix $k_1\geq\bar k_1$. By
\Cref{thm:total-welfare-monopoly-dominates}, we have that
\[
\mathcal W_{\mathrm{tot}}^{\M}(0)> \mathcal W_{\mathrm{tot}}^{\FI}.
\]
By \Cref{lem:total-welfare-high-k2}, $\mathcal W_{\mathrm{tot}}^{\DE}(\tau_2,k_2)$ converges to $\mathcal W_{\mathrm{tot}}^{\FI}$ as $k_2\to\infty$. The strict inequality therefore implies that there exists a finite threshold $\bar k_2(k_1,\tau_2)$ such that
\[
\mathcal W_{\mathrm{tot}}^{\M}(0)> \mathcal W_{\mathrm{tot}}^{\DE}(\tau_2,k_2)
\]
for every $k_2\geq\bar k_2(k_1,\tau_2)$.
\end{proof}


\begin{proof}[Proof of \Cref{thm:total-welfare-competition-dominates}]
		Let $S(\theta)$	denote the \emph{total} action-payoff difference from choosing $a^+$ rather than $a^-$. In the LQ class, we have $S(\theta)=3\theta-\tau_1-\tau_2$. If $\tau_2\geq-\tau_1$, then, for every $\theta<0$, we obtain $S(\theta)=3\theta-\tau_1-\tau_2\leq 3\theta<0$.
		
		Let $\ell^M<0<r^M$ denote the lower pooling type and pooling report in the least informative monopoly equilibrium. Under the receiver's full-information decision rule, action $a^+$ is chosen if and only if $\theta\geq0$. Under the least informative monopoly equilibrium, action $a^+$ is chosen if and only if $\theta\geq\ell^M$. Thus, relative to the receiver's full-information decision rule, the least informative monopoly equilibrium differs only on the interval $(\ell^M,0)$, where it induces action $a^+$ instead of $a^-$.
		
		Therefore,
		\[
		\mathcal W_{\mathrm{tot}}^{\M}(0)-\mathcal W_{\mathrm{tot}}^{\FI} = \int_{\ell^M}^{0} S(\theta)f(\theta)\,d\theta - k_1\int_{\ell^M}^{r^M}(r^M-\theta)^2f(\theta)\,d\theta.
		\]
		The first term is strictly negative because $S(\theta)<0$ for every $\theta<0$ and $f$ has full support. The second term is weakly negative, and it is strictly negative because the least informative monopoly equilibrium is non-degenerate, so sender~1 incurs positive misreporting costs on a set of positive measure. Hence, $\mathcal W_{\mathrm{tot}}^{\M}(0)-\mathcal W_{\mathrm{tot}}^{\FI}<0$, or, alternatively, $\mathcal W_{\mathrm{tot}}^{\FI}> \mathcal W_{\mathrm{tot}}^{\M}(0)$.
		
		By \Cref{lem:total-welfare-high-k2}, we have that $\mathcal W_{\mathrm{tot}}^{\DE}(\tau_2,k_2) \to \mathcal W_{\mathrm{tot}}^{\FI}$ as $k_2\to\infty$. Since $\mathcal W_{\mathrm{tot}}^{\FI}> \mathcal W_{\mathrm{tot}}^{\M}(0)$, there exists a finite threshold $\bar k_2(k_1,\tau_2)$ such that $\mathcal W_{\mathrm{tot}}^{\DE}(\tau_2,k_2) > \mathcal W_{\mathrm{tot}}^{\M}(0)$ for every $k_2\geq\bar k_2(k_1,\tau_2)$. This proves the proposition.
	\end{proof}



\section{Supplementary material}


	\subsection{Benchmarks: cheap talk and verifiable disclosure}
\label{app:polar-benchmarks}

Before analyzing costly misreporting, we consider two polar communication technologies. In the cheap-talk benchmark, messages are costless and unverifiable. In the verifiable-disclosure benchmark, senders cannot lie but may withhold information or disclose only coarse evidence. These benchmarks clarify what partial verifiability adds to the comparison between monopoly and competition. They also anticipate two themes of the analysis: the importance of equilibrium selection and the distinction between improving the receiver's decisions and increasing total welfare.

This distinction arises because the receiver's preferred action need not maximize aggregate surplus. Let $S(\theta)\coloneqq u_r(\theta)+u_1(\theta)+u_2(\theta)$ denote the aggregate action-payoff gain from choosing $a^+$ rather than $a^-$. The receiver chooses according to the sign of $u_r(\theta)$, whereas total welfare depends on the sign of $S(\theta)$. Thus, if $S(\theta)>0$ in some states in which $u_r(\theta)<0$, persuasion may harm the receiver while increasing total welfare. The benchmarks show that this wedge does not originate in misreporting costs. In the baseline model, those costs instead determine whether such persuasion is feasible and whether its allocative benefit survives after communication costs are taken into account.


	\subsubsection{Cheap talk}
	\label{subsec:cheap_talk_benchmark}
	
	Set $k_1=k_2=0$. Reports are then pure cheap talk. Since messages have no intrinsic meaning, it is useful to describe equilibria by their payoff-relevant outcome. For a market structure $E$, let $q^E(\theta)\in[0,1]$ denote the probability that the receiver chooses $a^+$ in state $\theta$.
	
	We do not attempt to characterize the full set of cheap-talk equilibria, which may also include mixed-strategy and payoff-equivalent variants. For the benchmark comparison, it is enough to distinguish babbling outcomes, one-sender outcomes, and the canonical both-senders-active binary advocacy outcome.
	
	The unrestricted cheap-talk game has the usual multiplicity. Babbling equilibria always exist in both market structures: all sender types send payoff-irrelevant messages, the receiver ignores messages, and she chooses a best reply to the prior. This observation is standard, so we do not state it as a formal result. It immediately implies that no all-equilibrium pro-competition result is possible in cheap talk. A comparison of market structures must either compare equilibrium sets or impose a selection on informative equilibria.
	
	There is also a second source of multiplicity that is specific to competition. In cheap talk, one sender can always be ignored. For instance, the receiver may use only sender 1's message and ignore sender 2. Sender 2 then has no payoff-relevant effect on the receiver's action, so any message rule, including babbling, is optimal for him. Given sender 2's babbling, it is optimal for the receiver to ignore him. Hence, the sender-1 monopoly outcome is always embedded in the competitive cheap-talk game.
	
	This embedding is a useful contrast with the costly-misreporting model. With positive misreporting costs, an ignored sender does not generally babble. If the receiver's action is independent of sender 2's report, sender 2 minimizes reporting costs by reporting truthfully. But a truthful report is informative; once the receiver observes it, it is no longer sequentially optimal to ignore it. Thus, the cheap-talk logic by which competition can replicate monopoly through one sender's babbling does not carry over to the baseline model with costly misreporting. This is one reason why competition is strategically sharper under partial verifiability than under cheap talk.
	
	We now record the payoff-relevant cheap-talk outcomes that will be used as benchmarks. The description below abstracts from redundant message labels and unused messages.
	
	\paragraph{Monopoly-like outcomes.}
	Consider first a one-sender cheap-talk outcome generated by sender $j$. Since sender $j$ wants $a^+$ if and only if $\theta\geq \tau_j$, all types on the same side of $\tau_j$ have identical preferences over messages. Hence, any payoff-relevant one-sender outcome has the form
	\[
	q_j(\theta)=
	\begin{cases}
		p_j^-, & \theta<\tau_j,\\
		p_j^+, & \theta\geq \tau_j,
	\end{cases}
	\qquad p_j^-\leq p_j^+,
	\]
	where each on-path action probability is a receiver best reply to the posterior induced by the corresponding cell. For sender 1, since $\tau_1<0$, the low cell lies entirely below the receiver's cutoff, so $p_1^-=0$. Let
	\[
	T_r \coloneqq\int_{\tau_1}^{\infty}u_r(\theta)f(\theta)d\theta.
	\]
	The sender-1 monopoly-like outcome therefore gives
	\[
	q_1(\theta)=
	\begin{cases}
		0, & \theta<\tau_1,\\
		p_1^+, & \theta\geq \tau_1,
	\end{cases}
	\qquad
	p_1^+\in BR_r(T_r),
	\]
	where, consistently with the maintained tie-breaking convention, $BR_r(x)=\{1\}$ if $x\geq0$ and $BR_r(x)=\{0\}$ if $x<0$. In monopoly this is the canonical informative cheap-talk outcome whenever $T_r\geq0$. In competition the same outcome is supported by having the receiver ignore sender 2 and sender 2 babble.
	
	Analogously, competition also admits sender-2 monopoly-like outcomes, in which the receiver listens only to sender 2 and sender 1 babbles. These outcomes are not the comparison benchmark, but they are part of the competitive cheap-talk equilibrium set. They also make clear why one should not expect all competitive equilibria to dominate monopoly: equilibrium selection remains essential.
	
	\paragraph{The two-sender advocacy outcome.}
	The natural informative competitive benchmark is the binary advocacy outcome in which each sender recommends the action he prefers. Sender 1 sends a pro-$a^+$ message if and only if $\theta\geq \tau_1$, while sender 2 sends a pro-$a^+$ message if and only if $\theta\geq \tau_2$. The induced cells are $(-\infty,\tau_1)$, $[\tau_1,\tau_2)$, and $[\tau_2,\infty)$. On the first cell both senders favor $a^-$, while on the third cell both favor $a^+$. On the middle cell sender 1 favors $a^+$ and sender 2 favors $a^-$. Define
	\[
 A_r(\tau_2) \coloneqq\int_{\tau_1}^{\tau_2}u_r(\theta)f(\theta)d\theta,
	\]
 \[
	B_r(\tau_2) \coloneqq\int_{\tau_2}^{\infty}u_r(\theta)f(\theta)d\theta.
	\]
	Since $\tau_2>0$, $B_r(\tau_2)>0$. The advocacy outcome is
	\[
	q_A(\theta)=
	\begin{cases}
		0, & \theta<\tau_1,\\
		p_A, & \theta\in[\tau_1,\tau_2),\\
		1, & \theta\geq\tau_2,
	\end{cases}
	\qquad
	p_A\in BR_r(A_r(\tau_2)).
	\]
	This outcome is an equilibrium. On the agreement cells, the receiver's action coincides with both senders' preferred action. On the disagreement cell, sender 1's message is the most favorable report available for inducing $a^+$, while sender 2's message is the most favorable report available for inducing $a^-$. Thus, neither sender can profitably deviate given the other's recommendation. Off-path beliefs after the unused profile in which sender 1 recommends $a^-$ and sender 2 recommends $a^+$ can be chosen so that neither sender has a profitable deviation.
	
	The advocacy outcome is not the unique non-babbling equilibrium of the unrestricted cheap-talk game. As just discussed, competition also admits babbling outcomes, sender-1 monopoly-like outcomes, sender-2 monopoly-like outcomes, and redundant refinements with payoff-irrelevant messages. The advocacy outcome should instead be understood as the canonical both-senders-active binary action-recommendation outcome. Within binary equilibria in which sender 1's recommendation is payoff-relevant when sender 2 recommends $a^-$, and sender 2's recommendation is payoff-relevant when sender 1 recommends $a^+$, the cutoffs must be $\tau_1$ and $\tau_2$: sender 1 is indifferent between recommendations only at $\tau_1$, and sender 2 is indifferent only at $\tau_2$.
	
	\paragraph{Receiver welfare.}
	Receiver welfare under the sender-1 monopoly-like cheap-talk outcome is
	\[
	W_r^{CT,M}=\max\{T_r,0\}.
	\]
	The formula also covers $T_r=0$, as the receiver chooses $a^+$ on the high cell but obtains zero expected payoff there. Receiver welfare under the advocacy outcome is
	\[
	W_r^{CT,A}=B_r(\tau_2)+\max\{A_r(\tau_2),0\}.
	\]
	Since $T_r=A_r(\tau_2)+B_r(\tau_2)$ and $B_r(\tau_2)>0$, we obtain that $W_r^{CT,A}\geq W_r^{CT,M}$. The inequality is strict whenever the receiver strictly dislikes the disagreement cell, $A_r(\tau_2)<0$. In that case competition allows sender 2 to block sender 1's pro-$a^+$ persuasion on a cell where the receiver would prefer $a^-$. If $A_r(\tau_2)\geq0$, the receiver chooses $a^+$ on the disagreement cell as well, so the advocacy outcome coincides with the sender-1 monopoly-like action rule.
	
	The effect of a more downward-biased sender 2 is also transparent in the advocacy outcome. Holding all other primitives fixed, increasing $\tau_2$ expands the disagreement cell and shrinks the agreement-on-$a^+$ cell. Since $W_r^{CT,A}=B_r(\tau_2)+\max\{A_r(\tau_2),0\}$, receiver welfare is locally constant when $A_r(\tau_2)>0$ and strictly decreasing when $A_r(\tau_2)<0$. Thus, a higher $\tau_2$ never helps the receiver in the advocacy benchmark. In the sender-1 monopoly-like outcome it has no effect, because sender 2 is ignored. In sender-2 monopoly-like outcomes it weakly worsens the receiver's information by moving sender 2's cutoff farther to the right. Hence, across the canonical cheap-talk outcomes, a more biased second sender never improves receiver welfare. Its effect is either neutral or harmful.

	
	\paragraph{Total welfare.}
	Because communication is costless in cheap talk, total welfare depends only on the states in which the receiver chooses $a^+$. Using $S(\theta)$ defined above, let
	\[
	A_S(\tau_2) \coloneqq\int_{\tau_1}^{\tau_2}S(\theta)f(\theta)d\theta,
	\]
 \[
	B_S(\tau_2) \coloneqq\int_{\tau_2}^{\infty}S(\theta)f(\theta)d\theta,
	\]
	and
 \[
 T_S \coloneqq A_S(\tau_2)+B_S(\tau_2)
 \]
 Since all players prefer $a^+$ on $[\tau_2,\infty)$, $B_S(\tau_2)>0$.
	
	Under the sender-1 monopoly-like outcome, total welfare is $W_{tot}^{CT,M}=p_1^+T_S$, where $p_1^+\in BR_r(T_r)$. Under the advocacy outcome, it is $W_{tot}^{CT,A}=B_S(\tau_2)+p_AA_S(\tau_2)$, where $p_A\in BR_r(A_r(\tau_2))$. Under the maintained tie-breaking convention, competition yields strictly lower total welfare than sender-1 monopoly if and only if $T_r\geq0$ and $A_r(\tau_2)<0<A_S(\tau_2)$. The interpretation is immediate. The first inequality says that sender 1's monopoly-like message induces $a^+$ on the pooled high cell. The second says that, on the disagreement cell, the receiver would prefer $a^-$. The third says that total welfare would nevertheless prefer $a^+$ on that same disagreement cell. Competition then improves receiver welfare by blocking sender 1's persuasion, but it lowers total welfare because the blocked persuasion was socially valuable.
	
	This is the main lesson of the cheap-talk benchmark. Since communication is costless, any total-welfare advantage of one-sided advice cannot be due to savings in misreporting costs. It comes from the action-allocation effect: monopoly may induce $a^+$ in states where the receiver dislikes $a^+$, but where aggregate action surplus favors $a^+$. Misreporting costs in the baseline model do not create this wedge. They determine whether the socially valuable distortion can be sustained and whether its benefit survives net of communication costs.
	
	In the LQ class, we have $S(\theta)=3\theta-\tau_1-\tau_2$. Therefore,
	\[
	A_S(\tau_2)=3A_r(\tau_2)-(\tau_1+\tau_2)\Pr\{\theta\in[\tau_1,\tau_2)\}.
	\]
	If $\tau_2\geq-\tau_1$, then $A_r(\tau_2)<0$ implies $A_S(\tau_2)<0$. Hence, under the natural restriction that the monopolist is the less biased sender, the advocacy outcome weakly improves total welfare relative to sender-1 monopoly. If instead sender 1's upward bias is sufficiently strong relative to sender 2's opposition, then $A_S(\tau_2)$ can be positive even when $A_r(\tau_2)$ is negative. In that case monopoly can dominate competition in total welfare even though competition improves receiver welfare.
	

	\subsubsection{Verifiable disclosure}
	\label{subsec:disclosure_benchmark}
	
	We now consider the opposite polar case. Disclosure is costless, and the evidence structure is rich: in state $\theta$, a sender may disclose any set $E\subseteq\Theta$ containing $\theta$, including the singleton $\{\theta\}$, or disclose no information. A sender can therefore withhold information or disclose coarse evidence, but cannot submit evidence that excludes the true state.
	
	The competitive disclosure benchmark has a strong all-equilibrium implication. In every state $\theta>0$, sender 1 and the receiver both prefer $a^+$. If an equilibrium induced $a^-$ with positive probability at such a state, sender 1 could disclose the singleton evidence $\{\theta\}$ and force the receiver to choose $a^+$. Hence, every competitive disclosure equilibrium has $q^{D,C}(\theta)=1$ for $\theta>0$. Similarly, in every state $\theta<0$, sender 2 and the receiver both prefer $a^-$. If an equilibrium induced $a^+$ with positive probability, sender 2 could disclose $\{\theta\}$ and force $a^-$. Hence,
	\[
	q^{D,C}(\theta)=\mathds{1}_{\{\theta\geq0\}}
	\]
	for every $\theta\neq0$ in every competitive disclosure equilibrium.
	
	In monopoly with only sender 1 active, the same logic is one-sided. If $\theta>0$, sender 1 can disclose $\{\theta\}$ and force $a^+$, so $q^{D,M}(\theta)=1$. If $\theta<\tau_1$, sender 1 and the receiver both prefer $a^-$, so no equilibrium can induce $a^+$ there. Possible mistakes are confined to the conflict interval $(\tau_1,0)$, where sender 1 prefers $a^+$ but the receiver prefers $a^-$. Therefore, competition weakly improves receiver welfare in every disclosure equilibrium, and it strictly improves receiver welfare whenever the monopolist induces $a^+$ with positive probability on a positive-measure subset of $(\tau_1,0)$.
	
	For total welfare, the same alignment issue as in cheap talk remains. Since there are no misreporting costs in the disclosure benchmark, the welfare difference between competitive disclosure and a monopoly disclosure equilibrium is
	\[
	W_{tot}^{D,C}-W_{tot}^{D,M} = -\int_{\tau_1}^{0}S(\theta)q^{D,M}(\theta)f(\theta)d\theta.
	\]
	Thus, competition weakly improves total welfare in every disclosure equilibrium if $S(\theta)\leq0$ for every $\theta\in(\tau_1,0)$. In the LQ class this condition is implied by $\tau_2\geq-\tau_1$. When this sign condition fails, the total-welfare comparison is generally ambiguous. By the displayed formula, competition reduces total welfare exactly when the weighted integral of $S(\theta)$ over the states in the conflict interval on which monopoly induces $a^+$ is positive.

	
	\subsubsection{Takeaways from the benchmarks}
	\label{subsec:benchmark_lessons}
	
	The two benchmarks deliver three takeaways. First, cheap talk has too much equilibrium multiplicity to support an all-equilibrium pro-competition result. Babbling equilibria always exist, and the competitive game can always reproduce the sender-1 monopoly outcome by having the receiver ignore sender 2 and sender 2 babble. This mechanism is specific to cheap talk. With positive misreporting costs, an ignored sender reports truthfully to avoid costs, making it impossible for the receiver to ignore him on path.

	Second, in the canonical informative cheap-talk outcome, competition improves receiver welfare by converting one sender's one-sided persuasion into a disagreement cell. A more downward-biased sender 2 never helps the receiver in this benchmark: it either has no effect because sender 2 is ignored, or it weakly worsens the receiver's payoff by expanding the region in which sender 2 withholds support for $a^+$.
	
	Third, the total-welfare comparison is governed by $S(\theta)$, not by $u_r(\theta)$. Cheap talk shows this most starkly because there are no communication costs. If the receiver and total welfare rank the relevant disagreement cell in the same way, competition improves both receiver welfare and total welfare. If they rank it differently, competition can improve receiver welfare while lowering total welfare. The baseline costly-misreporting model adds a second layer: finite misreporting costs determine how much persuasion and adversarial discipline are feasible, and whether the action-surplus benefit of one-sided advice survives net of communication costs.


\subsection{More than two senders}
\label{app:many-senders}

This appendix provides a simple existence result for environments with at least three senders. It is distinct from the selection analysis in \Cref{sec:common_refinement}. When all but one sender submit the same report, the receiver can treat that common report as identifying the state. A unilateral deviation then cannot affect the receiver's action, and truthful reporting is optimal because it minimizes reporting costs.

This mechanism differs from the skeptical-envelope selection argument. When all other senders agree, one sender's report is redundant, whereas the skeptical-envelope criterion requires local responsiveness to each report in its domain. The result below establishes the existence of full revelation under a separate belief restriction, yet it does not claim that full revelation is selected by the skeptical-envelope refinement.

Consider an extension of the model with a finite set of senders $N=\{1,\ldots,n\}$, where $n\geq3$. All senders observe the same state $\theta$, report simultaneously, and have payoffs and reporting costs as in the baseline model. Thus, sender $j$'s payoff is
\[
w_j(r_j,a,\theta) = \mathds{1}_{\{a=a^+\}}u_j(\theta)-k_jC_j(r_j,\theta),
\]
where $k_j>0$, $C_j(\theta,\theta)=0$, and $C_j(r_j,\theta)$ is strictly increasing in $|r_j-\theta|$. The ordering of the senders' biases is irrelevant for the result.

To support truthful reporting, we impose a restriction on beliefs following a unilateral deviation: the common report of all other senders determines the receiver's posterior.

\begin{definition}
\label{def:unilateral-consistency}
An assessment is \emph{unilaterally consistent} if, for every $x\in\Theta$ and every report profile $\bm r\in\Theta^n$ at which at least $n-1$ reports are equal to $x$, the receiver's posterior is
\[
\mu_{\bm r}=\delta_x.
\]
\end{definition}

Because $n\geq3$, at most one value can appear in at least $n-1$ coordinates, so the restriction is well defined.

\begin{proposition}
\label{prop:many-senders-full-revelation}
Suppose that $n\geq3$. There exists a perfect Bayesian equilibrium in which every sender reports truthfully, the receiver learns the state, and the assessment is unilaterally consistent.
\end{proposition}

\begin{proof}[Proof of \Cref{prop:many-senders-full-revelation}]
Let every sender report truthfully, i.e., $r_j(\theta)=\theta$ for every $j\in N$ and every $\theta\in\Theta$. At any report profile at which at least $n-1$ reports are equal to some $x\in\Theta$, set the receiver's posterior equal to $\delta_x$. Complete beliefs arbitrarily at all other off-path profiles, and let the receiver take a best response to her posterior at every report profile.

On the equilibrium path, the receiver observes $\bm r=(\theta,\ldots,\theta)$. Bayes' rule therefore gives $\mu_{\bm r}=\delta_\theta$, so the receiver learns the state and takes the full-information action.

Now fix sender $j$ and state $\theta$, and suppose that every other sender reports truthfully. After any unilateral deviation $r_j'$, the remaining $n-1$ reports are still equal to $\theta$. Unilateral consistency therefore gives
\[
\mu_{(r_j',\bm r_{-j})}=\delta_\theta.
\]
The deviation cannot change the receiver's action. It can only change sender $j$'s reporting cost, and
\[
k_jC_j(r_j',\theta) \geq k_jC_j(\theta,\theta)=0,
\]
with strict inequality whenever $r_j'\neq\theta$. Hence no sender has a profitable unilateral deviation. The proposed strategies, beliefs, and receiver actions form a perfect Bayesian equilibrium, and the equilibrium is fully revealing.
\end{proof}

\Cref{prop:many-senders-full-revelation} establishes the existence of a truthful perfect Bayesian equilibrium under unilateral consistency. It does not establish uniqueness, robustness to coordinated deviations, or admissibility under the skeptical-envelope refinement. The mechanism is that, following any unilateral deviation, the remaining $n-1$ truthful reports continue to identify the state. The deviation therefore cannot affect the receiver's action and can only increase the deviating sender's reporting cost. Equilibrium selection and coordinated deviations in many-sender environments remain outside the scope of the result.

	
	\subsection{NWBR without the single-crossing condition}
	\label{app:nwbr-arbitrary-qr}
	
	This appendix characterizes the $\NWBR$ deletion procedure when \Cref{ass:sc} is not imposed. Without single crossing, the types that survive after an off-path report are the minimizers of the indifference cutoff $q_r$. The results are not used in the welfare analysis, but they identify precisely when the refinement can eliminate a generic monopoly equilibrium.
	
	Fix a generic monopoly equilibrium indexed by $\lambda$, with pooling report $r^*(\lambda)$ and lower pooling type $\ell(\lambda)$. For any off-path report $r\in(\ell(\lambda),r^*(\lambda))$, write $\underline\theta(r)\coloneqq \bar r_1^{-1}(r)$, and let $\vartheta(r,\lambda)$ be the unique type in $(r,r^*(\lambda))$ satisfying
	\[
	C_1(r,\vartheta(r,\lambda)) = C_1(r^*(\lambda),\vartheta(r,\lambda)).
	\]
	On the set of types that can weakly benefit from deviating to $r$, define $q_r$ as in Section~\ref{sec:monopoly}. That is, for $\theta\in[\underline\theta(r),\ell(\lambda)]$,
	\[
	q_r(\theta) = \frac{k_1C_1(r,\theta)}{u_1(\theta)}
	\]
	and, for $\theta\in[\ell(\lambda),\vartheta(r,\lambda)]$,
	\[
	q_r(\theta) = 1- \frac{k_1\left(C_1(r^*(\lambda),\theta)-C_1(r,\theta)\right)} {u_1(\theta)}.
	\]
	Define the set of minimizing types,
	\begin{equation*}
		M_r^\lambda \coloneqq \arg\min_{\theta\in[\underline\theta(r),\vartheta(r,\lambda)]} q_r(\theta).
	\end{equation*}
	
	\begin{proposition}
		\label{prop:nwbr-survivors}
		Fix a generic monopoly equilibrium indexed by $\lambda$ and an off-path report $r\in(\ell(\lambda),r^*(\lambda))$. After one round of the Never-a-Weak-Best-Response test, the surviving type-report pairs are exactly $\{(\theta,r) \mid \theta\in M_r^\lambda\}$.
	\end{proposition}

	
	\begin{proof}[Proof of \Cref{prop:nwbr-survivors}]
		Fix $r\in(\ell(\lambda),r^*(\lambda))$. As in the proof of \Cref{thm:nwbr}, only types in the interval $[\underline\theta(r),\vartheta(r,\lambda)]$	can ever weakly gain from deviating to $r$. Hence, every pair $(\theta,r)$ with $\theta\notin[\underline\theta(r),\vartheta(r,\lambda)]$ is deleted immediately.
		
		Now take $\theta\in[\underline\theta(r),\vartheta(r,\lambda)]$. For such a type, the receiver probability that makes the type indifferent between following the equilibrium strategy and deviating to $r$ is $q_r(\theta)$. Hence, $D_0(\theta,r)=\{q_r(\theta)\}$ and $D(\theta,r)=(q_r(\theta),1]$. Therefore,
		\[
		D_0(\theta,r) \subseteq \bigcup_{\theta'\neq\theta}D(\theta',r) \iff \exists\, \theta'\neq\theta \; \text{ such that }\; q_r(\theta')<q_r(\theta).
		\]
		Alternatively, we can say that $(\theta,r)$ survives one round of $\NWBR$ if and only if $q_r(\theta)$ attains the minimum of $q_r$ on $[\underline\theta(r),\vartheta(r,\lambda)]$. By definition, this means $\theta\in M_r^\lambda$.
	\end{proof}

	
	\Cref{prop:nwbr-survivors} characterizes the deletion step. Survival of the equilibrium is a further requirement: once the surviving types have been identified, the receiver must be able to hold a posterior supported on them and respond in a way that makes the deviation unprofitable. The next corollary gives the exact criterion.
	
	\begin{corollary}
		\label{cor:nwbr-criterion-no-sc}
		A generic monopoly equilibrium indexed by $\lambda$ survives the Never-a-Weak-Best-Response refinement if and only if, for every off-path report $r\in(\ell(\lambda),r^*(\lambda))$, there exists a posterior belief $\mu_r$ supported on $M_r^\lambda$ such that the receiver has a best reply assigning probability $\sigma_r$ to action $a^+$ with
		\[
		\sigma_r \leq \min_{\theta\in M_r^\lambda}q_r(\theta).
		\]
	\end{corollary}

	
	\begin{proof}[Proof of \Cref{cor:nwbr-criterion-no-sc}]
		Fix an off-path report $r$. By \Cref{prop:nwbr-survivors}, after one round of $\NWBR$ the only surviving types are those in $M_r^\lambda$. Any off-path belief consistent with the refinement must therefore assign probability one to $M_r^\lambda$.
		
		Let $\sigma_r$ be the probability with which the receiver's best reply chooses action $a^+$ under some posterior belief $\mu_r$ supported on $M_r^\lambda$. For any surviving type $\theta\in M_r^\lambda$, the deviation payoff is strictly above the equilibrium payoff if and only if $\sigma_r>q_r(\theta)$. Since all types in $M_r^\lambda$ attain the same minimum value of $q_r$, the deviation is not profitable for any surviving type if and only if
		\[
		\sigma_r \leq \min_{\theta\in M_r^\lambda}q_r(\theta).
		\]
		This condition must hold for every off-path report $r$, which proves the claim.
	\end{proof}
	
	The criterion in \Cref{cor:nwbr-criterion-no-sc} shows why the single-crossing condition used in the main text is useful. Under \Cref{ass:sc}, the minimizing set is always the singleton $\{\ell(\lambda)\}$. Since $\ell(\lambda)<0$, the receiver can assign the deviation to a negative type and choose $a^-$. Without single crossing, the minimizing set may move to higher types, and the survival of a generic monopoly equilibrium depends on whether some receiver best reply supported on $M_r^\lambda$ can still deter the deviation.

	
	\subsection{Distance-cost extensions}\label{sec:lc}
	
	This appendix examines which LQ comparative statics extend beyond quadratic costs. The extension concerns equilibrium geometry, not welfare. In the linear-convex class, the reach formulas and the local cost effect continue to hold. A stronger Arrow--Pratt regularity condition also yields $G_s>0$ and preserves the local bias effect when sender~2's inverse reach is non-binding.

The appendix proceeds in two steps. \Cref{subsec:LC-class} derives the reach and cost results for the linear-convex class, while \Cref{app:costeffect} studies the comparative statics on misreporting cost for that class. \Cref{subsec:LAR-class} introduces the stronger curvature restriction and establishes the two remaining sign properties.
	
	
	\subsubsection{The linear-convex class}
	\label{subsec:LC-class}
	
	The linear-convex distance-cost class, or LC class, retains linear utilities but allows general convex distance costs. Formally, $u_i(\theta)=\theta-\tau_i$ for $i\in\{1,2,r\}$, with $\tau_1<\tau_r=0<\tau_2$, and $C_j(r,\theta)=\Gamma_j(|r-\theta|)$ for $j\in\{1,2\}$. Each function $\Gamma_j:[0,\infty)\to[0,\infty)$ satisfies
\begin{enumerate}[label=($\roman*$)]
    \item $\Gamma_j(0)=0$;
    \item $\Gamma_j$ is strictly increasing on $(0,\infty)$;
    \item $\Gamma_j$ is continuously differentiable on $(0,\infty)$;
    \item $\Gamma_j$ is weakly convex.
\end{enumerate}

Weak convexity is sufficient for the reach and cost results below, but not for the sign of the bias effect or for $G_s>0$.
	
	The senders' reaches admit a simple representation in the linear-convex class.
	
	\begin{lemma}
		\label{prop:lc-reaches}
		Suppose $u_i(\theta)=\theta-\tau_i$ for $i\in\{1,2,r\}$, $\tau_1<\tau_r=0<\tau_2$, and $C_j(r,\theta)=\Gamma_j(|r-\theta|)$, with $\Gamma_j$ strictly increasing and $\Gamma_j(0)=0$. Then, sender~1's upper reach is, for $\theta\geq \tau_1$,
		\[
		\bar r_1(\theta) = \theta+\Gamma_1^{-1}\!\left(\frac{\theta-\tau_1}{k_1}\right),
		\]
		and sender~2's lower reach is, for $\theta\leq \tau_2$,
		\[
		\underline{r}_2(\theta) = \theta-\Gamma_2^{-1}\!\left(\frac{\tau_2-\theta}{k_2}\right).
		\]
	\end{lemma}

	
	\begin{proof}[Proof of \Cref{prop:lc-reaches}]
		For sender~1, the largest upward report that is weakly better than truthful reporting when	action $a^+$ is induced with certainty solves, for $r\geq \theta$,
		\[
		\theta-\tau_1 = k_1\Gamma_1(r-\theta).
		\]
		Since $\Gamma_1$ is strictly increasing, the solution is unique and equals
		\[
		r=\theta+\Gamma_1^{-1}\!\left(\frac{\theta-\tau_1}{k_1}\right).
		\]
		This gives $\bar r_1(\theta)$.
		
		For sender~2, the most extreme downward report that is weakly better than truthful reporting when action $a^-$ is induced with certainty solves, for $r\leq \theta$,
		\[
		\tau_2-\theta = k_2\Gamma_2(\theta-r).
		\]
		Again by strict monotonicity of $\Gamma_2$, the solution is unique and equals
		\[
		r=\theta-\Gamma_2^{-1}\!\left(\frac{\tau_2-\theta}{k_2}\right).
		\]
		This gives $\underline{r}_2(\theta) $.
	\end{proof}
	
	The corresponding comparative statics are immediate.
	
	\begin{proposition}
		\label{prop:lc-reach-cs}
		Consider the LC class. Under the assumptions of \Cref{prop:lc-reaches}, for every fixed
		$\theta<\tau_2$,
		\[
		\frac{\partial \underline{r}_2(\theta) }{\partial \tau_2}<0< \frac{\partial \underline{r}_2(\theta) }{\partial k_2}.
		\]
		Likewise, for every fixed report $r_2<\tau_2$ in the domain of $\underline{r}_2 ^{-1}$,
		\[
		\frac{\partial \underline{r}_2^{-1}(r_2) }{\partial \tau_2}>0> \frac{\partial \underline{r}_2^{-1}(r_2) }{\partial k_2}.
		\]
	\end{proposition}

	
	\begin{proof}[Proof of \Cref{prop:lc-reach-cs}]
		Consider the following reach,
		\[
		\underline{r}_2(\theta) =\theta-\Gamma_2^{-1}\!\left(\frac{\tau_2-\theta}{k_2}\right).
		\]
		Because $\Gamma_2$ is differentiable, strictly increasing, and convex, $\Gamma_2'(x)>0$ for every $x>0$. Hence, $\Gamma_2^{-1}$ is differentiable on $(0,\infty)$ and $(\Gamma_2^{-1})'(x)>0$. Differentiating gives
		\[
		\frac{\partial \underline{r}_2(\theta) }{\partial \tau_2} = -(\Gamma_{2}^{-1})'\!\left(\frac{\tau_2-\theta}{k_2}\right)\frac{1}{k_2}<0,
		\]
		and
		\[
		\frac{\partial \underline{r}_2(\theta) }{\partial k_2} = -(\Gamma_{2}^{-1})'\!\left(\frac{\tau_2-\theta}{k_2}\right) \left(-\frac{\tau_2-\theta}{k_2^2}\right) >0.
		\]
		
		For $r_2<\tau_2$, strict monotonicity and $\underline r_2(\tau_2)=\tau_2$ imply $\underline r_2^{-1}(r_2)<\tau_2$. Therefore, its inverse satisfies $\underline{r}_2 (\underline{r}_2^{-1}(r_2) )=r_2$. Implicit differentiation yields
		\[
		\frac{\partial \underline{r}_2^{-1}(r_2) }{\partial \tau_2} = - \frac{\partial \underline{r}_2(\theta) /\partial \tau_2} {\partial \underline{r}_2(\theta) /\partial \theta} \bigg|_{\theta=\underline{r}_2^{-1}(r_2) } >0,
		\]
		and
		\[
		\frac{\partial \underline{r}_2^{-1}(r_2) }{\partial k_2} = - \frac{\partial \underline{r}_2(\theta) /\partial k_2} {\partial \underline{r}_2(\theta) /\partial \theta} \bigg|_{\theta=\underline{r}_2^{-1}(r_2) } <0,
		\]
		because $\partial \underline{r}_2(\theta) /\partial \theta>0$.
	\end{proof}
	
	The adversarial equilibrium swing report remains characterized by an implicit zero condition. For $r>0$ and $s(r)<0$, define $I(r,s)$ as in \eqref{eq:I}, and
	\begin{equation}
		G(r,s\mid \tau_2,k_2) \coloneqq \int_{I(r,s)} \frac{\theta f(\theta)\Gamma_1'(r-\theta)\Gamma_2'(\theta-s)} {(\theta-\tau_1)(\tau_2-\theta)} \,d\theta.
		\label{eq:lc-swing}
	\end{equation}
	The swing report is characterized by $G(r,s(r)\mid \tau_2,k_2)=0$.
	

	\subsubsection{The cost effect}\label{app:costeffect}
	
	The cost parameter $k_2$ enters the swing equation only through sender~2's reach. This delivers the cleanest comparative-statics result in the LC class.
	
	\begin{proposition}
		\label{prop:LC-cost-effect}
		Consider the LC class. Fix $r>0$ and let $s(r)$ be an interior adversarial-equilibrium swing report. Suppose that $G_s(r,s(r))>0$. At every differentiability point of the swing root,
		\[
		\frac{\partial s(r)}{\partial k_2}=0 \quad\text{if}\quad \underline r_2^{-1}(s(r))>r,
		\]
		and
		\[
		\frac{\partial s(r)}{\partial k_2}>0 \quad\text{if}\quad \underline r_2^{-1}(s(r))<r.
		\]
		At boundary-switching points, where $\underline r_2^{-1}(s(r))=r$, the corresponding one-sided derivatives are weakly non-negative whenever they exist.
	\end{proposition}

	
	\begin{proof}[Proof of \Cref{prop:LC-cost-effect}]
		The LC swing equation can be written as
		\[
		G(r,s\mid\tau_2,k_2) = \int_{I(r,s)} \frac{ \theta f(\theta)\Gamma_1'(r-\theta)\Gamma_2'(\theta-s) }{ (\theta-\tau_1)(\tau_2-\theta) } \,d\theta,
		\]
		where
		\[
		I(r,s) = \left[ \max\{s,\bar r_1^{-1}(r)\}, \min\{r,\underline r_2^{-1}(s)\} \right].
		\]
		The parameter $k_2$ enters this equation only through the upper integration bound $\underline r_2^{-1}(s)$. Thus, at differentiability points,
		\[
		G_{k_2}(r,s) = \mathds{1}_{\left\{\underline r_2^{-1}(s)<r\right\}} g(r,s) \frac{\partial \underline r_2^{-1}(s)}{\partial k_2},
		\]
		where
		\[
		g(r,s) = \frac{ \underline r_2^{-1}(s) f(\underline r_2^{-1}(s)) \Gamma_1'(r-\underline r_2^{-1}(s)) \Gamma_2'(\underline r_2^{-1}(s)-s) }{ (\underline r_2^{-1}(s)-\tau_1)(\tau_2-\underline r_2^{-1}(s)) }.
		\]
		Whenever the upper bound binds, $0<\underline r_2^{-1}(s)<r$ and $\underline r_2^{-1}(s)<\tau_2$, so $g(r,s)>0$. The LC reach formula implies
		\[
		\frac{\partial \underline r_2^{-1}(s)}{\partial k_2}<0.
		\]
		Therefore, $G_{k_2}(r,s)\leq0$, with strict inequality when the inverse reach binds.
		
		Since $G_s(r,s(r))>0$, implicit differentiation gives
		\[
		\frac{\partial s(r)}{\partial k_2} = -\frac{G_{k_2}(r,s(r))}{G_s(r,s(r))} \geq0,
		\]
		with strict inequality when $\underline r_2^{-1}(s(r))<r$. At boundary-switching points the conclusion follows by one-sided differentiation.
	\end{proof}


	The LC result is local and conditional on $G_s>0$. It is also a boundary effect: $k_2$ affects the swing equation only when sender~2's inverse reach binds. Two LQ sign properties do not follow from weak convexity alone. First, the sign of $G_s$ is not automatic. Second, the bias effect
\[
\frac{\partial s(r)}{\partial \tau_2}
\]
cannot generally be signed from weak convexity of $\Gamma_2$, because differentiating \eqref{eq:lc-swing} with respect to $\tau_2$ produces an interior term whose sign depends on the shape of $\Gamma_2'$. We next impose a stronger curvature restriction that recovers both properties.

	
	\subsubsection{Arrow-Pratt regular distance costs}
	\label{subsec:LAR-class}
	
	The stronger class restricts the local curvature of distance costs. The key object is the curvature-to-slope ratio
\[
\rho_j(x)\coloneqq\frac{\Gamma_j''(x)}{\Gamma_j'(x)},
\qquad x>0.
\]
This is formally analogous to the Arrow--Pratt coefficient of absolute risk aversion, except that it is applied to a convex cost function rather than to a concave utility function. We call the resulting class the \emph{linear Arrow--Pratt regular} (LAR) class.

\begin{assumption}
    \label{ass:LAR}
    Utilities satisfy $u_i(\theta)=\theta-\tau_i$ for
    $i\in\{1,2,r\}$, with $\tau_1<\tau_r=0<\tau_2$. For each sender
    $j\in\{1,2\}$, $C_j(r,\theta)=\Gamma_j(|r-\theta|)$, where
    $\Gamma_j:[0,\infty)\to[0,\infty)$ satisfies:
    \begin{enumerate}[label=(\roman*)]
        \item $\Gamma_j(0)=0$ and $\Gamma_j'(0)=0$;
        \item $\Gamma_j$ is strictly increasing on $(0,\infty)$;
        \item $\Gamma_j\in C^1([0,\infty))\cap C^2((0,\infty))$;
        \item $\Gamma_j''(x)\geq0$ for every $x>0$;
        \item the ratio
        $\rho_j(x)=\Gamma_j''(x)/\Gamma_j'(x)$ is strictly decreasing
        on $(0,\infty)$.
    \end{enumerate}
\end{assumption}

\Cref{ass:LAR} adds smoothness at zero and a decreasing curvature-to-slope restriction. Standard power costs $\Gamma_j(x)=x^\alpha$ with $\alpha>1$ satisfy the assumption because
\[
\rho_j(x)=\frac{\alpha-1}{x},
\]
which is strictly decreasing. The LQ class corresponds to $\alpha=2$. Linear costs are excluded because $\Gamma_j'(0)\neq0$.
	
	We now establish the two additional sign properties used in \Cref{sec:LQ}. For $r>0$ and $s<0$, recall that
	\begin{equation}
		G(r,s\mid \tau_2,k_2) = \int_{I(r,s)} \frac{\theta f(\theta)\Gamma_1'(r-\theta)\Gamma_2'(\theta-s)} {(\theta-\tau_1)(\tau_2-\theta)} \,d\theta,
		\label{eq:lar-swing}
	\end{equation}
	where $I(r,s)$ is defined by \eqref{eq:I}. The swing report satisfies $G(r,s(r)\mid \tau_2,k_2)=0$.
	
	\begin{lemma}
		\label{prop:lar-Gs}
		Suppose \Cref{ass:LAR} holds. Then, for every $r\in(0,\bar r_1(0))$ and every corresponding swing report $s(r)$, we have $G_s(r,s(r)\mid\tau_2,k_2)>0$. At the endpoint, we obtain $G_s\!\left( \bar r_1(0),s(\bar r_1(0))\mid\tau_2,k_2 \right)=0$, where the derivative is understood in the feasible one-sided sense.
	\end{lemma}

	
	\begin{proof}[Proof of \Cref{prop:lar-Gs}]
		Fix $r\in(0,\bar r_1(0))$ and write $s\coloneqq s(r)$. Define
		\[
		a(r,s)\coloneqq \max\left\{s,\bar r_1^{-1}(r)\right\},
		\]
    \[
		b(r,s)\coloneqq \min\left\{r,\underline{r}_2 ^{-1}(s)\right\},
		\]
		so that $I(r,s)=[a(r,s),b(r,s)]$. Also define
		\[
		H_r(\theta,s) \coloneqq \frac{\theta f(\theta)\Gamma_1'(r-\theta)\Gamma_2'(\theta-s)} {(\theta-\tau_1)(\tau_2-\theta)}.
		\]
		Then,
		\[
		G(r,s\mid\tau_2,k_2)=\int_{a(r,s)}^{b(r,s)} H_r(\theta,s)\,d\theta.
		\]
		On the interval $I(r,s)$, all factors except $\theta$ are strictly positive, so $\operatorname{sign} H_r(\theta,s)=\operatorname{sign}\theta$. Moreover,
\[
\partial_s H_r(\theta,s) = -H_r(\theta,s)\rho_2(\theta-s).
\]
		
		Differentiate with respect to $s$. By Leibniz' rule,
		\begin{align}\label{eq:Gs-Leibniz} 
			G_s(r,s)
			&=
			\int_{a(r,s)}^{b(r,s)} \partial_s H_r(\theta,s)\,d\theta + \mathds{1}_{\left\{b(r,s)=\underline{r}_2 ^{-1}(s)\right\}}H_r(b(r,s),s)\frac{\partial \underline{r}_2 ^{-1}(s)}{\partial s} \\ \nonumber
			&\qquad -\mathds{1}_{\left\{a(r,s)=s\right\}}H_r(a(r,s),s).
		\end{align}
        At a boundary-switching point, the same expression is obtained from the corresponding one-sided derivatives. The relevant boundary integrand vanishes because \Cref{ass:LAR} imposes $\Gamma_1'(0)=\Gamma_2'(0)=0$.
		
		We now study the two boundary terms. First, if $a(r,s)=\bar r_1^{-1}(r)$, then $a(r,s)$ does not depend on $s$, so the lower-bound contribution is zero. If instead $a(r,s)=s$, then
		\[
		H_r(s,s) = \frac{s f(s)\Gamma_1'(r-s)\Gamma_2'(0)} {(s-\tau_1)(\tau_2-s)} =0,
		\]
		because \Cref{ass:LAR} imposes $\Gamma_2'(0)=0$. Hence, the lower-bound term in \eqref{eq:Gs-Leibniz} is always zero. Second, for the upper-bound term, if $b(r,s)=r$, then $b(r,s)$ is constant in $s$, so again the contribution is zero. If instead $b(r,s)=\underline{r}_2 ^{-1}(s)$, then
		\[
		\frac{\partial b(r,s)}{\partial s} = \frac{\partial \underline{r}_2 ^{-1}(s)}{\partial s}.
		\]
		Therefore, the upper-bound contribution is
		\begin{equation}
			\mathds{1}_{\left\{\underline{r}_2 ^{-1}(s)<r\right\}}\, H_r\left(\underline{r}_2 ^{-1}(s),s\right)\, \frac{\partial \underline{r}_2 ^{-1}(s)}{\partial s}.
			\label{eq:upper-boundary-term}
		\end{equation}
		
		We claim that \eqref{eq:upper-boundary-term} is weakly positive. Indeed, when the indicator is zero, the term is zero. When the indicator is one, we have $0<\underline{r}_2 ^{-1}(s)<r$ and $\underline{r}_2 ^{-1}(s)<\tau_2$, so every factor in
		\[
		H_r\left(\underline{r}_2 ^{-1}(s),s\right) = \frac{\underline{r}_2 ^{-1}(s)\,f\left(\underline{r}_2 ^{-1}(s)\right) \Gamma_1'\left(r-\underline{r}_2 ^{-1}(s)\right) \Gamma_2'\left(\underline{r}_2 ^{-1}(s)-s\right)} {\left(\underline{r}_2 ^{-1}(s)-\tau_1\right)\left(\tau_2-\underline{r}_2 ^{-1}(s)\right)}
		\]
		is strictly positive. Moreover, $\underline{r}_2 ^{-1}(s)$ is increasing in $s$, and differentiating the defining identity
		\[
		\tau_2-\underline{r}_2 ^{-1}(s) = k_2\Gamma_2(\underline{r}_2 ^{-1}(s)-s)
		\]
		yields
		\[
		\frac{\partial \underline{r}_2 ^{-1}(s)}{\partial s} = \frac{k_2\Gamma_2'(\underline{r}_2 ^{-1}(s)-s)} {1+k_2\Gamma_2'(\underline{r}_2 ^{-1}(s)-s)} >0.
		\]
		Hence, the upper-bound term in \eqref{eq:upper-boundary-term} is strictly positive when it is active, and zero otherwise. Thus, it is weakly positive.	As a result, \eqref{eq:Gs-Leibniz} becomes
		\[
		G_s(r,s) = \int_{a(r,s)}^{b(r,s)} \partial_s H_r(\theta,s)\,d\theta + \mathds{1}_{\left\{\underline{r}_2 ^{-1}(s)<r\right\}}\, H_r\left(\underline{r}_2 ^{-1}(s),s\right)\, \frac{\partial \underline{r}_2 ^{-1}(s)}{\partial s},
		\]
		with the second term weakly positive.
		
		Since $s$ is a swing report,
		\[
		G(r,s\mid \tau_2,k_2)=\int_{a(r,s)}^{b(r,s)} H_r(\theta,s)\,d\theta=0.
		\]
		Therefore, for any constant $c$,
		\[
		\int_{a(r,s)}^{b(r,s)} H_r(\theta,s)\rho_2(\theta-s)\,d\theta = \int_{a(r,s)}^{b(r,s)} H_r(\theta,s)\left(\rho_2(\theta-s)-c\right)\,d\theta.
		\]
		Choose $c=\rho_2(-s)$. Because $\rho_2$ is strictly decreasing, we have
		\[
		\theta>0 \implies \theta-s>-s \implies \rho_2(\theta-s)-\rho_2(-s)<0,
		\]
		and
		\[
		\theta<0 \implies \theta-s<-s \implies \rho_2(\theta-s)-\rho_2(-s)>0.
		\]
		Since $H_r(\theta,s)$ has the sign of $\theta$, it follows that $H_r(\theta,s)\left(\rho_2(\theta-s)-\rho_2(-s)\right)<0$ for every $\theta\neq 0$ in $I(r,s)$. Hence,
		\[
		\int_{a(r,s)}^{b(r,s)} H_r(\theta,s)\rho_2(\theta-s)\,d\theta<0.
		\]
		Therefore,
\[
\int_{a(r,s)}^{b(r,s)} \partial_s H_r(\theta,s)\,d\theta = -\int_{a(r,s)}^{b(r,s)} H_r(\theta,s)\rho_2(\theta-s)\,d\theta >0.
\]
Adding the weakly positive boundary term yields $G_s(r,s)>0$.

At $r=\bar r_1(0)$, we have $s(r)=\underline r_2(0)$ and $a(r,s)=b(r,s)=0$. Hence, the integral in \eqref{eq:Gs-Leibniz} is zero. Its upper-boundary term
is also zero because $H_{\bar r_1(0)}\!\left(0,\underline r_2(0)\right)=0$. Thus, the feasible one-sided derivative at the endpoint equals zero.
	\end{proof}
	
	The next result shows that the sign of the bias effect also extends beyond the quadratic case, provided sender~2's inverse reach does not bind at the comparison-relevant report.
	
	\begin{proposition}
		\label{prop:lar-bias}
    Suppose \Cref{ass:LAR} holds. Fix $r\in(0,\bar r_1(0))$, and suppose that sender~2's inverse reach is locally non-binding at the corresponding swing report, i.e.,
    \[
    \underline r_2^{-1}(s(r))>r.
    \]
    Then,
    \[
    \frac{\partial s(r)}{\partial \tau_2}>0.
    \]
	\end{proposition}

	
	\begin{proof}[Proof of \Cref{prop:lar-bias}]
		Fix $r\in(0,\bar r_1(0))$ and write $s\coloneqq s(r)$. Under the strict non-binding condition $\underline r_2^{-1}(s)>r$, the upper integration limit in \eqref{eq:lar-swing} is locally equal to $r$. Hence, there is no boundary term when differentiating $G$ with respect to $\tau_2$, and
		\[
		G_{\tau_2}(r,s) = -\int_{a(r,s)}^{r} \frac{\theta f(\theta)\Gamma_1'(r-\theta)\Gamma_2'(\theta-s)} {(\theta-\tau_1)(\tau_2-\theta)^2} \,d\theta.
		\]
		Define
		\[
		K_r(\theta,s) \coloneqq \frac{f(\theta)\Gamma_1'(r-\theta)\Gamma_2'(\theta-s)} {\theta-\tau_1},
		\]
		so that $K_r(\theta,s)>0$ on the relevant interval. Then,
		\[
		G(r,s)=\int_{a(r,s)}^{r}\frac{\theta K_r(\theta,s)}{\tau_2-\theta}\,d\theta=0.
		\]
		Rewrite this identity as
		\[
		0 = \int_{a(r,s)}^{r} \frac{\theta(\tau_2-\theta)K_r(\theta,s)}{(\tau_2-\theta)^2}\,d\theta = \tau_2\int_{a(r,s)}^{r}\frac{\theta K_r(\theta,s)}{(\tau_2-\theta)^2}\,d\theta - \int_{a(r,s)}^{r}\frac{\theta^2 K_r(\theta,s)}{(\tau_2-\theta)^2}\,d\theta.
		\]
		The first integral is exactly $-G_{\tau_2}(r,s)$. Hence,
		\[
		\tau_2\left(-G_{\tau_2}(r,s)\right) = \int_{a(r,s)}^{r} \frac{\theta^2 K_r(\theta,s)}{(\tau_2-\theta)^2}\,d\theta.
		\]
		The right-hand side is strictly positive, and $\tau_2>0$, so $G_{\tau_2}(r,s)<0$. By \Cref{prop:lar-Gs}, $G_s(r,s)>0$. Implicit differentiation of $G(r,s(r)\mid \tau_2,k_2)=0$ yields
		\[
		\frac{\partial s(r)}{\partial \tau_2} = -\frac{G_{\tau_2}(r,s(r))}{G_s(r,s(r))}>0.
		\]
	\end{proof}
	
	\Cref{prop:lar-Gs,prop:lar-bias} show that the main sign arguments of the LQ analysis are not artifacts of quadratic costs. Under \Cref{ass:LAR}, the swing condition is locally well behaved and, when sender~2's inverse reach is locally non-binding, an increase in his downward bias raises the swing report. The next corollary collects these comparative statics.

\begin{corollary}
    \label{cor:LAR-sign-logic}
    Suppose \Cref{ass:LAR} holds. Fix $r\in(0,\bar r_1(0))$ and let $s(r)$ be a corresponding interior swing report. Then,
	    \begin{enumerate}[label=(\roman*)]
        \item the swing equation is locally well behaved:
        \[
        G_s(r,s(r))>0;
        \]
        \item at differentiability points,
        \[
        \frac{\partial s(r)}{\partial k_2}\geq0,
        \]
        with strict inequality when $\underline r_2^{-1}(s(r))<r$ and equality when $\underline r_2^{-1}(s(r))>r$; at a boundary-switching point, the corresponding one-sided derivatives are weakly non-negative whenever they exist; \item if $\underline r_2^{-1}(s(r))>r$, then
        \[
        \frac{\partial s(r)}{\partial \tau_2}>0.
        \]
    \end{enumerate}
\end{corollary}

\begin{proof}[Proof of \Cref{cor:LAR-sign-logic}]
Part~$(i)$ follows from \Cref{prop:lar-Gs}. Part~$(ii)$ follows from \Cref{prop:LC-cost-effect}, because the LAR class is a subclass of the LC class and part~$(i)$ supplies $G_s>0$. Part~$(iii)$ follows from \Cref{prop:lar-bias}.
\end{proof}

\Cref{cor:LAR-sign-logic} is a mechanism result, not a welfare result. Receiver welfare also depends on the truthful cutoffs, equilibrium mixing probabilities, and mistake probabilities, none of which is controlled by these local properties of the swing equation.



\clearpage
\bibliographystyle{aer}
\bibliography{biblio_news}

\end{document}